\documentclass[12pt,  oneside, a4paper]{article}
\usepackage[utf8]{inputenc}
\usepackage[a4paper, total={6.5in,9in}]{geometry}

\usepackage{hyperref}
\usepackage{amsmath}
\usepackage{mathtools}
\usepackage{lipsum}
\usepackage{blindtext}
\usepackage{titlesec}
\usepackage{graphicx}
\usepackage{amssymb}
\usepackage{amsthm}
\usepackage{rotating}
\usepackage{subcaption}
\usepackage{caption}
\usepackage{svg}
\usepackage{verbatim}
\usepackage{enumerate}
\usepackage{enumitem}
\usepackage{subfloat}
\usepackage[square,sort]{natbib}
\usepackage{url}
\usepackage{cleveref}

\newtheorem{theorem}{Theorem}[section]
\newtheorem*{theorem*}{Theorem}
\newtheorem{definition}[theorem]{Definition}
\newtheorem{lemma}[theorem]{Lemma}
\newtheorem*{lemma*}{Lemma}
\newtheorem{claim}[theorem]{Claim}

\newtheorem{corr}[theorem]{Corollary}

\newtheorem{conj}[theorem]{Conjecture}

\newtheorem{defn}{Definition}[section]
\newtheorem{remark}{Remark}[section]

\newcommand{\Qone}{\textcolor{blue}{\textsf{Category 1}}}
\newcommand{\Qtwo}{\textcolor{olive}{\textsf{Category 2}}}
\newcommand{\Qthree}{\textcolor{yellow}{\textsf{Category 3}}}
\newcommand{\Qfour}{\textcolor{red}{\textsf{Category 4}}}

\newcommand{\cP}{\mathcal{P}}
\newcommand{\w}{\bf{w}}
\newcommand{\pr}{\mathbb{P}}
\newcommand{\M}{\mathcal{M}}

\newcommand{\footremember}[2]{%
    \footnote{#2}
    \newcounter{#1}
    \setcounter{#1}{\value{footnote}}%
}

\title{Differentially Private Data Release on Graphs: Inefficiencies and Unfairness}
\author{Ferdinando Fioretto\footremember{alley}{University of Virginia. Email: fioretto@virginia.edu} 
\and Diptangshu Sen\footremember{trailer}{Georgia Institute of Technology. Email: dsen30@gatech.edu (Primary student author)}
\and Juba Ziani\footremember{somethingelse}{Georgia Institute of Technology. Email: jziani3@gatech.edu}
}
\date{\today}

\begin{document}

\maketitle
\begin{abstract}
Networks are crucial components of many sectors, including telecommunications, healthcare, finance, energy, and transportation.
The information carried in many such networks often contains sensitive user data, such as location data for commuters and packet data for online users. Therefore, when considering data release for networks, one must ensure that data release mechanisms do not leak excessive information about individual users, quantified in a precise mathematical sense. Differential Privacy (DP) is the widely accepted, formal, state-of-the-art technique, which has found use in a variety of real-life settings including the 2020 U.S. Census, Apple users' device data, or Google's location data, to name a few. 

Yet, the use of differential privacy comes with new challenges, as the noise added for privacy introduces inaccuracies or biases. Such biases are unavoidable for any ``reasonable'' privacy technique; the issue, however, is that DP techniques can also \emph{distribute these biases disproportionately across different populations, inducing fairness issues.} The goal of this paper is to characterize the impact of differential privacy on bias and unfairness in the context of releasing information about networks, taking a departure from previous work which has studied these effects in the context of private population counts release (such as in the U.S. Census). To this end, we consider a \emph{network release problem} where the network structure is known to all, but the \emph{weights} on edges must be released privately. We consider the impact of this private release on a simple downstream decision-making task run by a third-party, which is to find the \emph{shortest path} between any two pairs of nodes and recommend the best route to users. 
This setting is of highly practical relevance, mirroring scenarios in transportation networks, where preserving privacy while providing accurate routing information is crucial.
Our work provides theoretical foundations and empirical evidence into the bias and unfairness arising due to privacy in these networked decision problems.
\end{abstract}

\newpage
\section{Introduction}\label{sec:intro}
Networks underlie many crucial application domains, such as telecommunications, social networks, energy grids, infrastructure, and transportation. Understanding their properties is, therefore, crucial, and there is often a need for publishing network information to serve a multitude of purposes including but not limited to navigation and routing (transportation and computer networks), predictive network maintenance (computer and infrastructure networks), understanding (mis-)information propagation (social networks), for research and development purposes (e.g., energy grids), or to inform public policy. 

However the release of network data poses a key challenge since it often contains sensitive information and needs to be used and released carefully. For example, releasing data about energy and infrastructure can provide malicious entities insights into system vulnerabilities; data from social network and telecommunication can expose personal information about individuals' preferences, social interactions, and activities; transportation data can inadvertently reveal sensitive personal details like home addresses, healthcare-related visits, and other personal information allowing targeting of workers in high-security and confidential sectors~\cite{nyt1,nyt2}.  

Therefore, when releasing network data, protecting potentially sensitive information is crucial. To this end, \emph{Differential Privacy} \cite{dwork:06} has emerged as the leading paradigm for preserving individual-level privacy in aggregate-level data release. Notably, this privacy framework has been adopted in various deployments, including the 2020 U.S.~Census ~\cite{us_census}, Apple's device data collection and federated learning frameworks ~\cite{apple_dp}, and Google's location data and maps services ~\cite{google_dp}. 

In a nutshell, differential privacy relies on noise addition on the outputs of a computation to provide strong privacy guarantees. However, while this process ensures that the amount of sensitive information that can be ``leaked'' remain bounded, the added noise can introduce biases and inaccuracies, potentially impacting the reliability of the data. 
While these biases are a natural consequence of any privacy-preserving method, a concerning issue with DP is that it can distribute errors and biases \emph{unevenly} across different groups, leading to concerns about fairness. 

Our work investigates the implications of DP on bias and fairness in network data release, focusing on routing recommendations. This constitutes a departure from previous research that primarily centered on the release of population histograms (e.g., in the U.S.~Census) absent such network structure. 
Specifically, we examine the common scenario where the network structure is known but the edge weights need to be released privately. Our analysis shows how these perturbations influence tasks such as computing the shortest path and recommending optimal routes. Figure~\ref{fig:privacy_model} presents an overview of our privacy model and data release, which we introduce in more detail in Section \ref{sec:model}.

\begin{figure}[!t]
      \centering
      \includegraphics[width=0.9\textwidth]{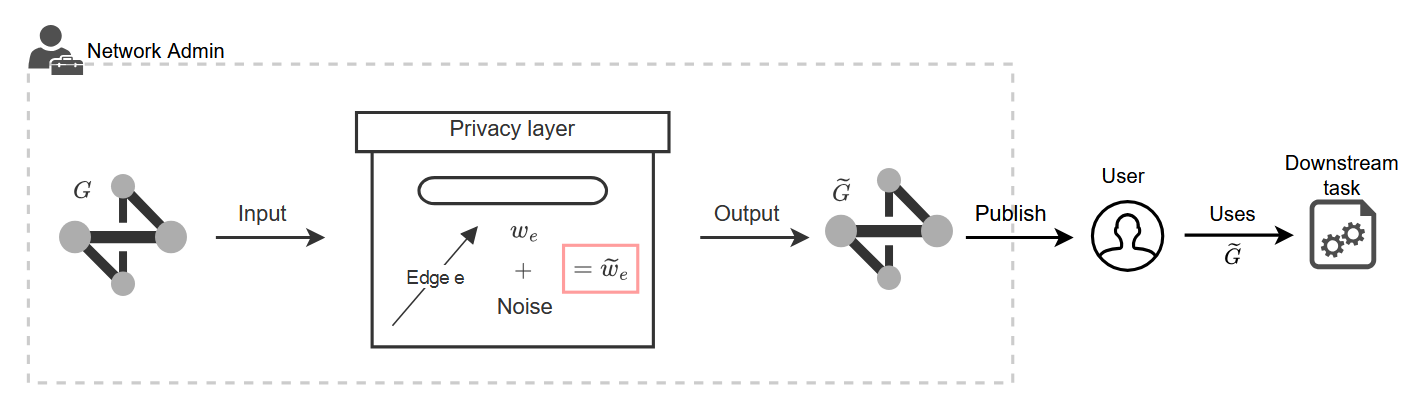}
      \caption{Schematic of privacy model: The network administrator privatizes graph $G$ by adding calibrated noise to each edge weight $w_e$ and publishes the privatized graph $\widetilde G$ with perturbed edge weights $\tilde w_e$. Users then use $\widetilde G$ to run downstream optimization tasks, such as shortest path computations.}
      \label{fig:privacy_model}
\end{figure}

Our work offers both theoretical insights and practical evidence on how differential privacy can introduce bias and unfairness in network-related decisions. Identifying these impacts is a first step towards developing more equitable and effective privacy-preserving techniques for network data. Beyond our characterizations of unfairness and biases due to DP, the paper also provides some understanding of how the network's structure affects unfairness and how different network structures are more or less robust to this unfairness, providing initial guidance toward network-design-oriented mitigation techniques.

\smallskip \noindent \textbf{Summary of contributions.} The main contributions of our work are as follows: 
\begin{enumerate}
    \item 
    We propose a model for differentially private network data release, assuming common knowledge of graph topology but requiring protection of sensitive edge weights through calibrated noise addition. This setup is detailed in
    Section~\ref{sec:model}. 
    \item  We investigate the bias and unfairness effects of using private (noisy) graph data to solve downstream optimization problems -- particularly, the problem of computing shortest paths on the graph and recommending best routes to users. To the best of our knowledge, we are the first who seek to understand the tradeoff between privacy and fairness in the context of private graph data release. 
    \item In Section~\ref{sec:theory}, we develop a theoretical framework that explains how DP-induced biases could disproportionately affect certain groups, particularly through the mechanics of noise accumulation over different path lengths and the availability of alternate routes. 
    \item Finally, in Section~\ref{sec:exp}, details extensive simulations conducted on diverse network topologies to demonstrate how privacy-related disruptions can vary by network type. This analysis also identifies network structures that are inherently more resilient to privacy-induced biases.
\end{enumerate}



\section{Related Work}
Observations that algorithms can mimic and amplify biases in data have resulted in a whole new research area that has focused on defining, analyzing, and mitigating unfairness (for surveys and summaries of the fairness literature, please refer to~\citep{barocas2023fairness,survey1,survey2}). 
The source of the observed unfairness has often been attributed to either i) data properties 
or ii) different aspects of the model's properties. 
For example, imbalance in groups' size is commonly argued to create disparities in the task's performance \cite{mehrabi2021survey}. It has also been shown that constraining the model's hypothesis space to satisfy privacy \cite{bagdasaryan2019differential,Fioretto:NeurIPS21a}, sparsity \cite{hooker2019selective,hooker2020characterising,TFKN:neurips22}, or robustness \cite{xu2021robust,nanda2021fairness,TZFvH:ijcai24} can result in disparate outcomes. 

Particularly relevant to our work is the study of the disparate impacts caused by privacy-preserving algorithms, which has recently seen several important developments \cite{FHZ:ijcai22}. Much of this line of research, similarly to our work, focuses on \emph{differential privacy}~\cite{dwork:06,dwork2014algorithmic} as the formal notion of privacy leading to unfairness.
In particular, in the context of learning tasks, \citet{ekstrand:18} raise questions about the trade-offs involved between privacy and fairness. Subsequently, \citet{cummings:19} study the trade-offs arising between differential privacy and equal opportunity, a fairness notion requiring a classifier to produce equal true positive rates across different groups. They show that there exists no classifier that simultaneously achieves $(\epsilon,0)$-DP, satisfies equal opportunity, and has accuracy better than a constant classifier. This development has risen the question of whether one can practically build fair models while retaining sensitive information private, which culminated in a variety of proposals, including \citep{jagielski:18,mozannar2020fair,tran2020differentially,Fioretto:ArXiv22d,Fioretto:NeurIPS21a,TF:ijcai23,TZFvH:ijcai23,TFvH:aaai21}. 

In the context of private data release (which involves revealing a full, privatized version of a dataset as opposed to simply releasing targeted statistics), \citet{pujol:20} were the first to show, empirically, that decision tasks made using DP datasets may disproportionately affect some groups of individuals over others. They noticed that the use of DP census data to allocate funds to school district produces unbalanced allocation errors, with some school districts systematically receiving more (or less) than what warranted. These observations were then attributed theoretically to two main factors: (1) the ``shape'' of the decision problem  \cite{TFHY:ijcai21} and (2) the presence of non-negativity constraints in post-processing steps \cite{ZHF:aaai21,ZFH:ijcai22}. 

To the best of the authors knowledge, no other work has observed nor studied the tension between privacy and fairness in downstream tasks performed on differentially-privately released \emph{network} data. Directly related to our work,~\citet{sealfon2016shortest} and ~\citet{chen2023differentially} do study the problem of differentially privately computing shortest paths, but i) they do not study the problem of releasing a private version of the entire network, just shortest path statistics, and ii) do not concern themselves with bias and fairness. Our paper thus builds on the body of work at the intersection of privacy and fairness and provides an analysis for the unfairness in a new context involving potentially complex network structures. 

\section{Preliminaries: Differential Privacy}\label{sec:prelim}
Differential privacy (DP)~\cite{dwork:06,dwork2014algorithmic} represents the forefront of techniques designed to safeguard individual data privacy. DP introduces a conceptual framework wherein two hypothetical scenarios are considered for each individual; these scenarios differ solely in the presence or absence of an individual’s data. The principle of differential privacy mandates that an adversary should not be able to reliably discern between these two scenarios based solely on the output distributions of a computation. In essence, the precise data value of an individual exerts a minimal influence on the computation’s result, thereby obscuring any single data point from being inferred with significant certainty. 

\noindent\textbf{Formal definition.} 
Formally, consider a mechanism $\mathcal{M}$ that operates on a dataset $x$ to derive a specific property $\mathcal{M}(x)$. For a dataset comprising $n$ individuals, we represent $x$ as a vector $(x_1, \ldots, x_n)$, where $x_i$ corresponds to the data associated with the $i$-th individual. We begin by defining the concept of neighboring datasets, which is fundamental in the context of DP:
\begin{definition}[Neighboring datasets]
Two datasets $x$ and $x’$ are said to be neighboring if they differ solely by the data of a single individual. That is, there exists an index $j \in [n]$ such that $x_j \neq x_j’$, while $x_i = x_i’$ for all $i \neq j$.
\end{definition}

Differential privacy, as informally described above, requires that the outputs of mechanism $\M$ exhibit minimal variability when applied to any two neighboring datasets, $x$ and $x’$. The formal criterion for this requirement is articulated as follows:
\begin{definition}[$(\varepsilon,\delta)$-differential privacy]
Let $\varepsilon,~\delta > 0$. A randomized algorithm $\M$ satisfies $(\varepsilon,\delta)$-differential privacy if, for any set of outcomes $O$ in the range of $\M$, the following inequality holds, for all neighboring databases $x, x'$:
\[
\Pr \left[\M(x) \in O \right] \leq \exp(\varepsilon) \Pr \left[\M(x') \in O \right] + \delta.
\]
\end{definition}
The parameter $\varepsilon$ plays a key role in controlling the level of privacy provided by the mechanism on each individual. As $\varepsilon$ decreases, the privacy constraint becomes increasingly stringent, enhancing individual privacy protection. Specifically, as $\varepsilon \to 0$, differential privacy requires that $\Pr \left[\M(x) = o \right]$ approaches $\Pr \left[\M(x') = o \right]$; i.e., the outcome of the mechanism becomes independent of the input data and thus perfectly preserves privacy (but, most likely, also provides no utility). Conversely, as $\varepsilon$ approaches $\infty$, the privacy constraint is trivially satisfied, effectively offering no privacy safeguard.

The underlying mechanism adopted by algorithms satisfying differential privacy involves the addition of noise to computations that interact with the original data. This noise injection is designed around the concept of the \emph{sensitivity of a function}, which is formally defined as follows:
\begin{align*}
\Delta f =  \max&~~\Vert f(x) - f(x') \Vert
\\\text{s.t.}&~~x,~x'~~\text{are neighbors}.
\end{align*}
Here, $f$ represents a query or computation applied to the data, and $\Delta f$ quantifies the maximum potential change in the function’s output across two neighboring databases. The sensitivity of $f$ is a key concept; a lower sensitivity indicates minor changes between outputs for neighboring databases, simplifying the task of masking these differences with noise. Consequently, lesser noise is required to achieve privacy. Notably, if the sensitivity is zero, $f$ effectively behaves as a constant function, and no noise is necessary to preserve privacy.

Numerical queries, which output a real number, can be made differentially private by adding calibrated noise to their true output values. For a given function $f$:
\begin{lemma*}[The Gaussian mechanism]
The Gaussian mechanism, defined as $\M(f,x,\varepsilon) = f(x) + Z$ where $Z \sim \mathcal{N}\left(0,\sqrt{2 \ln(1.25/\delta)} \cdot \Delta f / \varepsilon \right)$ is $(\varepsilon,\delta)$-differentially private. 
\end{lemma*}

In practice, the magnitude of noise introduced to preserve privacy is inversely related to the $\varepsilon$ parameter. Lower $\varepsilon$ values are associated to the addition of more noise, which in turn enhances the privacy guarantees. This inverse relationship underscores a fundamental trade-off in differential privacy: increasing privacy strength typically results in a reduction of the utility of the output due to the greater noise level. {\em This paper will focus on understanding how this reduction in utility may be disproportional distributed among different individuals}.

\noindent\textbf{Post-processing invariance.} Differential privacy satisfies several important properties \cite{dwork2014algorithmic}. In particular DP is resistant to post-processing manipulations. Informally, this property states that any data-independent post-processing step applied solely to the output of a differentially private mechanism does not compromise its privacy guarantees. More formally:

\begin{theorem*}[\citet{dwork2014algorithmic}]
Let $\M$ be a randomized algorithm that is $(\varepsilon,\delta)$-differentially private. Consider $f$, an arbitrary randomized function from the range of $\mathcal{M}$ to $\mathbb{R}$. The composite function $f \circ \mathcal{M}$ retains the $(\varepsilon, \delta)$-differential privacy properties of $\mathcal{M}$.
\end{theorem*}

\section{Model: Settings and Goals}\label{sec:model}
We consider the problem of \emph{differentially private graph data release}.  Formally, let $G = (V,E,\w)$ be a weighted graph with vertex set $V$, edge set $E$, and weights $\w: E \rightarrow \mathbb{R}_{\geq 0}$. 
For each edge $e \in E$, $w(e)$ is used to denote the its weight, here used to  represent the ``time'' or ``cost'' it takes to traverse it. 
Without loss of generality, we consider connected graphs $G$ in which any two nodes are reachable from each other. Importantly, in this work we consider weights $\w$ that are functions of sensitive user data and whose values must be protected. For instance, the weights might represent traffic congestion based on commuter locations or the strength of private social relationships in a network. We write $w(e) = f_e(x_1,\ldots,x_n)$ where $(x_1,\ldots,x_n)$ denotes sensitive information, such as geographic locations of users $1$ through $n$.

\paragraph{Differential privacy graph release model.} Consider a network administrator who wishes to release information about a weighted graph $G = (V, E, \w)$ to a third party. To preserve the privacy of underlying data, the administrator generates a graph $\widetilde{G} = (V, E, \tilde{\w})$ where the structure of nodes and edges remains unchanged, but the edge weights $\tilde{\w}$ are altered to ensure differential privacy. This modified graph, termed the privatized or publicized graph, retains the publicly available network topology of $G$ while safeguarding sensitive weight information through differential privacy techniques.

The administrator uses the \emph{Gaussian mechanism}, described in Section~\ref{sec:prelim}, to release weights $\tilde{\w}$; formally, for each $e \in E$, 
\begin{align}\label{eq:noisy_weight}
\tilde{w}(e) = \max \left(0, w(e) + Z(e) \right),
\end{align}
where $Z(e) \sim \mathcal{N}(0,\sigma^2)$ is a centered Gaussian random variable. The application of the $\max$ function ensures that all reported weights remain non-negative, adhering to the post-processing immunity of differential privacy, as outlined in the Preliminaries\footnote{This step retains differential privacy, per the post-processing guarantees discussed earlier}. 
When the sensitivity of function $f_e(\cdot)$ in users' data is bounded by $\Delta f$ for all $e \in E$, the released graph guarantees the $(\varepsilon,\delta)$-differential privacy of the edge weights for any $(\varepsilon,\delta)$ satisfying $\sigma = \sqrt{2 \ln(1.25/\delta)} \cdot \Delta f / \varepsilon$. The higher the value of $\sigma$, the \emph{better} the privacy guarantee. In this paper, we will focus on $\sigma$ as our main parameter controlling the level of noise and privacy, and refer the reader to this model section to relate the choice of $\sigma$ to a formal differential privacy guarantee.

\begin{remark}[Motivating Example]
The study of the shortest path problem provides a compelling context for our study. A notable real-world application is the private release of traffic data on road networks. Services like Google Maps leverage crowd-sourcing to gather live location data from thousands of users, enabling the system to assess traffic conditions, predict commute times, and suggest optimal routes in real-time~\citep{google_maps}. Numerous other organizations also collect and disseminate extensive user data to third parties, aiming to enhance understanding of traffic patterns and congestion levels. However, the use of such sensitive data raises significant privacy concerns~\citep{google_privacy}, necessitating robust privacy-preserving mechanisms. Differential privacy is, in this setting, a widely adopted tool for private graph data release.
\end{remark}

\paragraph{Impact of differential privacy on bias and fairness.} As the introduction of noise for privacy and the subsequent post-processing step in $\widetilde{G}$, which ensures non-negative edge weights, can introduce inaccuracies and biases in statistical and optimization tasks performed on the publicized, privatized graph. In this paper we aim to {\bf (1)} characterize such bias both theoretically and experimentally, and {\bf (2)} to understand \emph{unfairness} in how different segments of the network may  be \emph{disparately} affected by this bias. Our analysis focuses primarily on the disparities in how users, experiencing varying commute times through the network, are impacted by these modifications.

In most of the paper, we fix the task of interest to be a \emph{shortest-path computation} task. Let $\cP_{ij}$ be the set of paths between any two vertices $i,~j \in V$. The \emph{length} of a path $P \in \cP_{ij}$ is given by $w_G(P) = \sum_{e \in P} w(e)$. The \emph{shortest path} between nodes $x$ and $y$ is given by 
\[
P^*_{ij} = \arg\min_{P \in \cP_{ij}} w_G(P) = \arg\min_{P \in \cP_{ij}} \sum_{e \in P} w(e).
\]
Our goal is to evaluate the extent to which differential privacy mechanisms, when applied to graph $G$ to produce graph $\widetilde{G}$, impact this computation. In the privatized graph $\widetilde{G}$, the \emph{perceived} shortest path is computed as:
\[
\tilde{P}_{ij} = \arg\min_{P \in \cP_{ij}} w_{\widetilde{G}}(P) = \arg\min_{P \in \cP_{ij}} \sum_{e \in P} \tilde{w}(e).
\]
We note that $\widetilde{G}$ serves as a basis for the shortest path computations and route recommendation, the actual cost incurred by a user that decides to take path $\tilde{P}_{ij}$ corresponds to the weights from the \emph{original} graph $G$. Therefore, our evaluation metric is based on $w_G(\tilde{P}_{ij}) = \sum_{e \in \tilde{P}_{ij}} w(e)$, as highlighted in Figure~\ref{fig:eval_model}, and the \emph{realized bias}\footnote{We use the term ``realized bias'' here to highlight that there is, indeed, bias; Section~\ref{sec:results} shows that the error we make is always non-negative and not centered around 0.} or \emph{error} of the shortest path computation is given by 
\[ 
B_{ij}(\tilde P_{ij}) = \sum_{e \in \tilde{P}_{ij}} w(e) - \sum_{e \in P_{ij}^*} w(e).
\]
Given the stochastic nature of $\tilde{\w}$, the \emph{perceived} shortest path $\widetilde{P}_{ij}$ is subject to variability. Therefore, it is useful to also define the \emph{(expected) bias} of the shortest path computation as follows: 
\begin{align}\label{eq:exp_bias}
   \mathbb{E}[B_{ij}] = \mathbb{E}_{\tilde{\w}} \left[ \sum_{e \in \tilde{P}_{ij}} w(e) - \sum_{e \in P_{ij}^*} w(e)   \right].
\end{align}
\begin{figure}[!ht]
      \centering
      \includegraphics[width=0.8\textwidth]{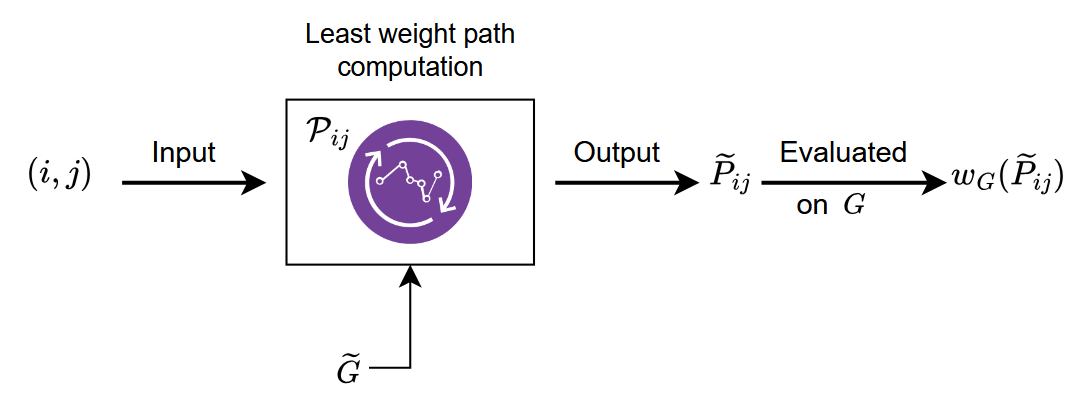}
      \caption{Evaluation Framework: Given any node pair ($i,j$) and privatized graph $\widetilde G$, a user computes the shortest path between ($i,j$) on the set $\cP_{ij}$. The computation returns path $\tilde P_{ij}$ as the perceived shortest path on $\widetilde G$ which the user commits to. Her decision is then evaluated on the original graph $G$ incurring a cost of $w_G(\tilde P_{ij})$ and realizing bias $B_{ij} = w_G(\tilde P_{ij})-w_G(P_{ij}^*)$.}
      \label{fig:eval_model}
\end{figure}

In the numerical section, we will often work with relative errors or bias, defined as 
\begin{align}\label{eq:rel_bias}
R_{ij} = \frac{\mathbb{E}[B_{ij}]}{ \sum_{e \in P_{ij}^*} w(e)},
\end{align}
and representing the \emph{percentage} change in the length of the recommended path (in expectation) compared to the true shortest path. Figure~\ref{fig:eval_model} provides a summary of the evaluation framework.

\paragraph{Summary of model and interactions.} 
We conclude this section with an overview of the interactions in our model referring back to Figure~\ref{fig:privacy_model}. A network administrator with access to the true graph $G$ computes a differentially private version $\widetilde{G}$ of said graph though addition of noise to the edge weights. The network administrator then shares the privatized graph $\widetilde{G}$ with a downstream user, that runs an optimization task on $\widetilde{G}$ which, in this case, is a shortest path computation.

\section{Bias: A Theoretical Perspective}\label{sec:theory}

This section presents the main theoretical insights of our work. Our primary contribution is characterizing the bias of the shortest path computation due to privacy noise and understanding how it drives unfair outcomes across different types of source-destination pairs on graphs. We introduce our first result in Claim \ref{clm:bias_positive} which provides insights about the sign or direction of the bias. 

\begin{claim}\label{clm:bias_positive}
The realized bias of the shortest path computation due to privacy noise is always greater than or equal to zero.
\end{claim}

\begin{proof}
Suppose, some path $P \in \cP_{ij}$ is the new \emph{perceived} shortest path on privatized graph $\widetilde{G}$ instead of the true shortest path $P_{ij}^*$ on $G$. In this case,
the realized bias $B_{ij}(P)$ is given by:
\[
      B_{ij}(P) = \sum_{e \in P} w(e) - \sum_{e \in P_{ij}^*} w(e) = w_G(P) - w_{G}(P_{ij}^*).
\]
Now, since $P_{ij}^*$ is the \emph{true} shortest path on $G$, by definition, it must be that: 
\[
        w_G(P) \geq w_G(P_{ij}^*) \quad \forall~ P \in \cP_{ij},
\]
which directly implies that $B_{ij}(P) \geq 0$. Since the above holds for any general path $P \in \cP_{ij}$, this concludes the proof of the claim. 
\end{proof}
A direct consequence of the above claim is that the \emph{expected bias} and \emph{expected relative bias} are non-negative. Note that all our numerical results in Section \ref{sec:exp} plot empirical probabilities for incurring different levels of expected relative bias. 


When it comes to fairness impacts of privacy, there are two main competing effects that drive which groups of node pairs will \emph{unfairly} face more disruptions (on average) due to privacy: 
\begin{enumerate}
    \item The first of those is the \emph{effective relative noise effect} which is explored in Section \ref{sec:2_pathcase}: when the number of path alternatives is fixed, we show that node pairs which are farther apart have a lower likelihood of being affected by privacy noise. 
    \item On the other hand, we also demonstrate the \emph{path cardinality effect} in Section \ref{sec:general}, i.e., the higher the number of different paths available to travel between the source and destination, the higher is the likelihood of shifting to a \emph{worse} path due to privacy noise and incurring a large bias. This effect favours node pairs which are closer because they usually have a smaller number of alternate path options. 
\end{enumerate}
The trade-off between these two effects explains most of our observations in the numerical experiments section. We also provide a dual interpretation of our main theorem in Section \ref{sec:general} which helps us to derive high probability bounds on the realized bias of any shortest path computation.\\


Before we present our main results, we need to introduce some additional notation for ease of exposition. From now on, we drop the subscript $``ij"$ whenever it is clear from context to simplify notations.
\begin{defn}\label{defn:set}
For any two paths $P_1$ and $P_2$ in $\cP_{ij}$, we define $S_{P_1, P_2} \subset E$ as follows: 
\[
    S_{P_1, P_2} := \{ e \in E:~ e \in (P_1 \setminus P_2) \cup (P_2 \setminus P_1)\}, 
\]
i.e., $S_{P_1, P_2}$ is the set of those edges which belong in exactly one of the two paths $P_1$ and $P_2$. 
\end{defn}
\noindent
Note that when paths $P_1$ and $P_2$ have no overlapping edges, $|S_{P_1, P_2}| = n_{P_1} + n_{P_2}$ where $n_{P_1}$ and $n_{P_2}$ denote the number of edges in paths $P_1$ and $P_2$ respectively. In general, $|S_{P_1, P_2}| \leq n_{P_1} + n_{P_2}$.

\subsection{Effective Relative Noise Effect}\label{sec:2_pathcase} 

In this segment, we are interested in understanding the disparate impacts that privacy noise has on node pairs which are close by versus node pairs which are far apart, when the number of alternate path options is kept fixed for each pair. We measure the impact of noise by estimating the probability that for any two given paths, the \emph{worse} one is perceived to be \emph{better} when computations are done using privatized graph $\widetilde G$. Higher the value of this probability, higher is the impact of noise. We make the following conjecture:

\begin{conj}\label{conj1}
Node pairs which are closer incur, on average, larger levels of relative noise and hence are more impacted by privacy as opposed to node pairs which are farther apart. 
\end{conj}

In order to gain intuition about why the above conjecture may be true, we will start by presenting the following technical result. Let $P^*$ be the true shortest path between nodes $i$ and $j$ and $P' \neq P^*$ be any other alternate path. Define the \textbf{gap} $\alpha_{P',P^*}$ as $\alpha_{P',P^*} = w_G(P') - w_G(P^*)$. We assume that $\alpha_{P',P^*} > 0$ which means that $P^*$ is strictly better than $P'$. Then, 

\begin{lemma}\label{lem:prob_2path}
The probability that path $P'$ is perceived to be \emph{shorter} than the true best path $P^*$ on a privatized graph $\widetilde G$, i.e., $\pr \left[ w_{\widetilde G}(P') < w_{\widetilde G }(P^*) \right]$, is given by:
\[
      q = \Phi^{c}\left( \frac{\alpha_{P',P^*}}{\sigma \sqrt{|S_{P', P^*}|}} \right),
\]
where $\Phi^c(\cdot)$ is the complementary CDF of a standard normal random variable. We call ``$q$" the \textbf{path deviation probability}.
\end{lemma}
\begin{proof}[Proof Sketch]
Recall that $Z(e)$ is the amount of noise added to edge $e \in E$. We know that $Z(e)$'s are i.i.d. normal mean-zero random variables with variance $\sigma^2$. The proof idea is to express the event of choosing the wrong shortest path equivalently as an event when a certain linear inequality condition on $Z(e)$'s is satisfied. Then we can exploit the normality and independence properties of $Z(e)$'s to reason about the probability. The full proof can be found in Appendix \ref{sec:proofs}. 
\end{proof}

\paragraph{Intuition about Conjecture \ref{conj1}:} We can obtain valuable insights about our earlier conjecture from Lemma \ref{lem:prob_2path}. Suppose for a given pair of nodes, there are exactly $2$ paths which have $|S|$ distinct edges between them and they differ in weight by amount $\alpha$. This implies that the gap $\alpha$ is contributed by exactly $|S|$ edges on which the effective privacy noise has standard deviation $\sigma \sqrt{|S|}$. Therefore, the ratio $\frac{ \sigma\sqrt{|S|} }{\alpha}$ represents the \emph{effective relative noise} (effective noise relative to the weight gap between paths). Now, suppose we scale the number of edges by a factor of $M > 1$ to represent a node pair which are farther apart than the first pair. Assuming that all edge weights are i.i.d. samples from some distribution $\mathcal{D}$ and this new pair of nodes also have exactly $2$ paths, the path gap between them should also scale by $M$ in expectation. In this case, the \emph{effective relative noise} is $\frac{1}{\sqrt{M}}\cdot \frac{\sigma \sqrt{|S|}}{\alpha}$. Because of the additional $\frac{1}{\sqrt{M}}$ factor, the effective relative noise is \emph{smaller} on average for the pair of nodes farther apart. Therefore by Lemma \ref{lem:prob_2path}, node pairs which are farther apart have on average, a lower likelihood of picking the worse path and hence are less affected by privacy noise.

\paragraph{Other observations from Lemma \ref{lem:prob_2path}:}
Recall that the standard deviation of the privacy noise $\sigma$ depends on the privacy parameter $\varepsilon$ and the sensitivity of the weight function $\Delta f$. The dependence is of the following form: $\sigma \propto \frac{\Delta f}{\varepsilon}$. This implies that at higher levels of privacy (smaller $\varepsilon$), the probability $q$ would be larger. This is intuitive: stronger privacy requires more perturbation to the edge weights and therefore there is a higher chance that the order is flipped, i.e., a previously longer path is perceived to be shorter. We can argue similarly for the case where the sensitivity of $f(\cdot)$ is high. Higher sensitivity of $f(\cdot)$ implies we need more noise to achieve the same level of privacy. This leads to higher $q$. We plot these dependencies in Figure \ref{fig:2_path}. 


We have already explored at depth how $q$ depends in average on the effective relative noise (Conjecture \ref{conj1}). $q$ also depends on the local network topology of paths $P'$ and $P^*$ as we illustrate with the following example. Let there be two users traveling between two different node pairs, each of them has two path choices, one which is the true best and another which is strictly worse. For ease of comparison, we assume that for both node pairs, the worse path is off the respective true best by the same amount $\alpha$. Now, suppose that user $1$ faces a scenario where both of her paths have a large degree of overlap, leading to a smaller $|S|$, while for user $2$, the paths are largely distinct. In this case, user $2$ has a higher chance of deviating to the \textit{worse} path, simply because noise on shared edges affects both paths equally. This example demonstrates that despite 
the path gap being identical, unfairness can also arise due to network topology wherein privacy has a much more adverse effect on some users compared to others.  

\begin{figure}[!ht]
  \centering
  \raisebox{20pt}{\parbox[b]{.11\textwidth}{}}%
  \subfloat[][Variation with gap $\alpha$]{\includegraphics[width=.5\textwidth]{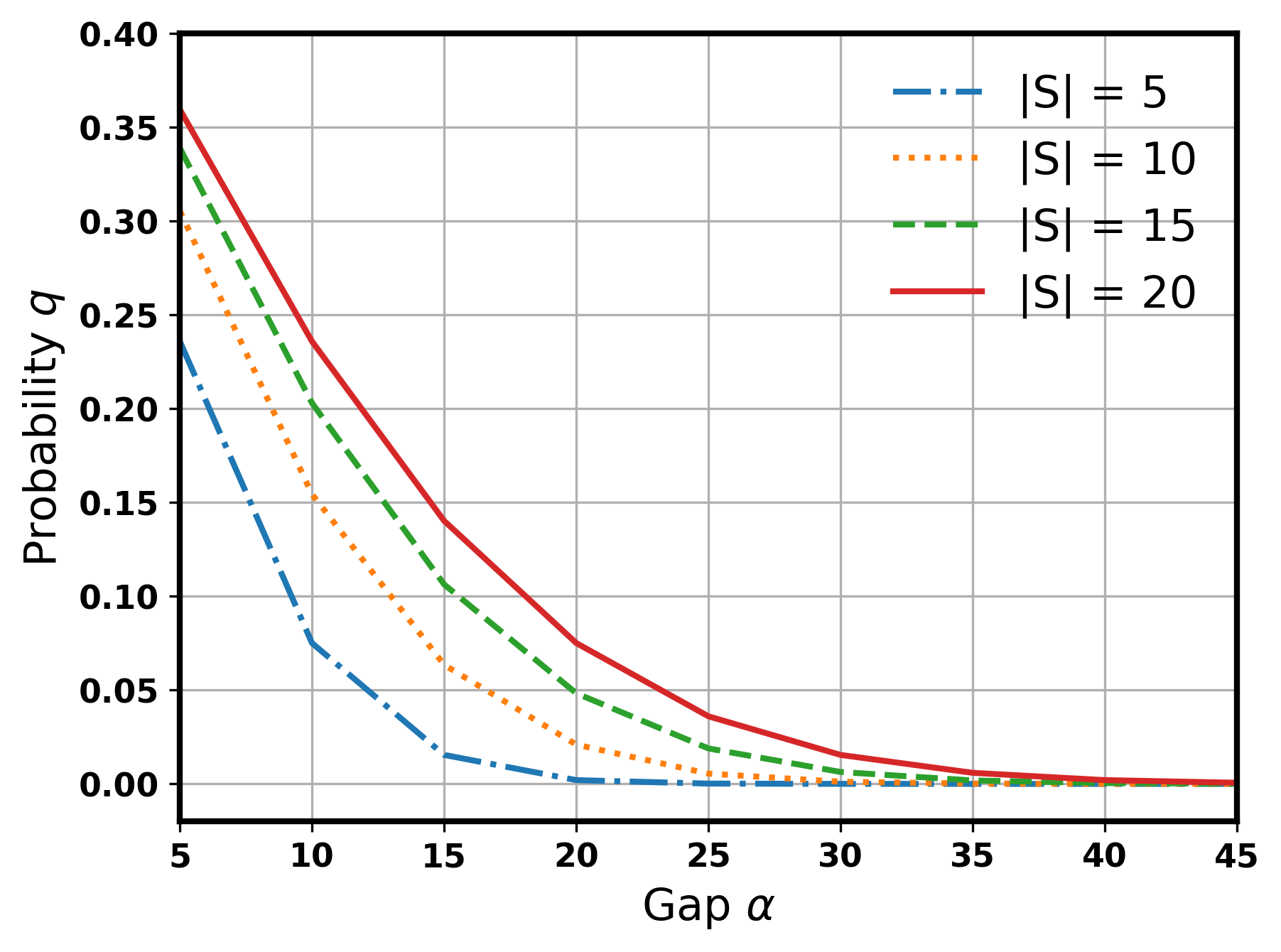}}\hfill
  \subfloat[][Variation with sensitivity $\Delta f$]{\includegraphics[width=.5\textwidth]{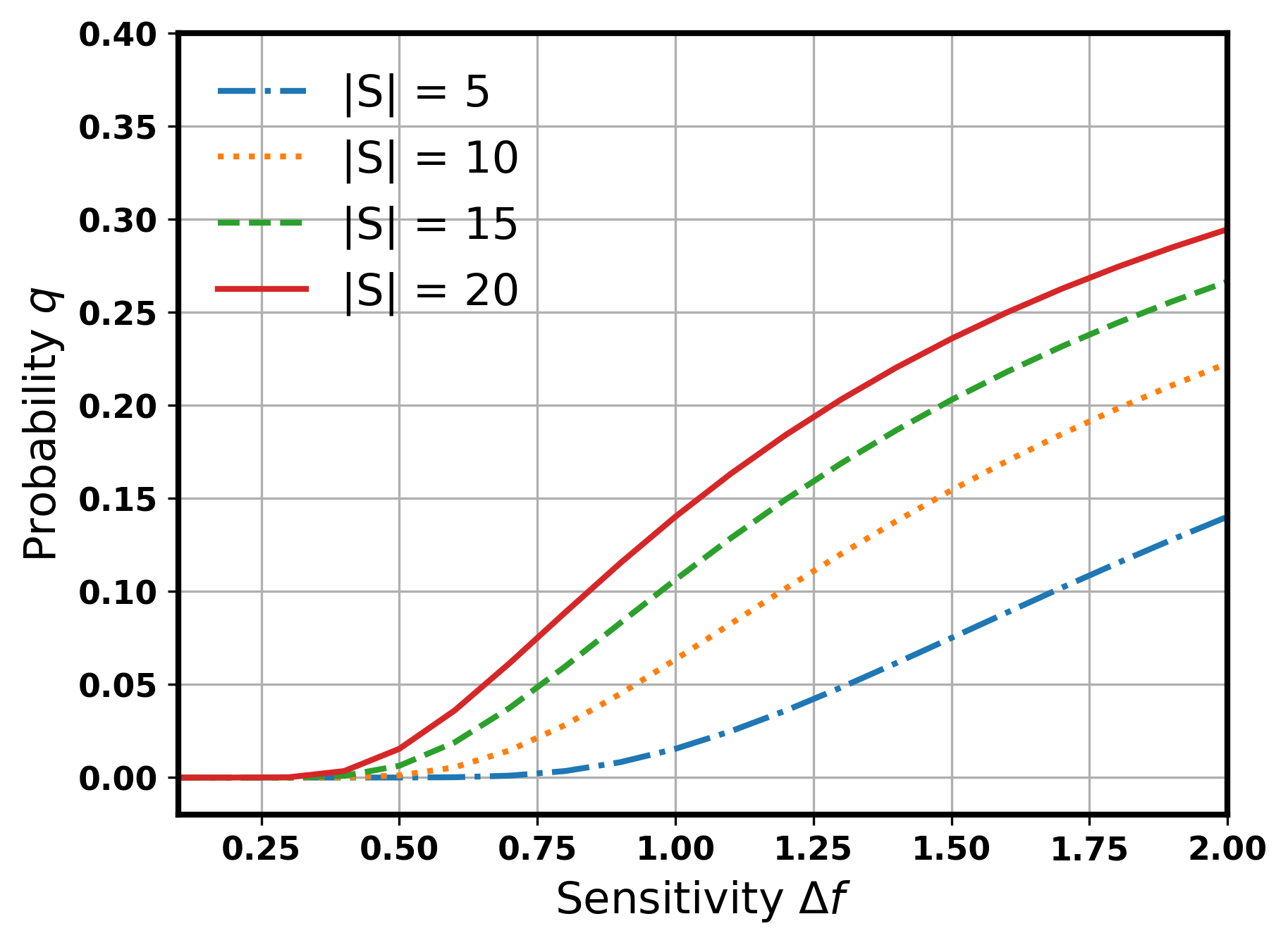}}\par
  \caption{Variation of probability $q$ as a function of gap $\alpha_{P',P^*}$ in (a) and sensitivity $\Delta f$ in (b) for different values of $|S_{P',P^*}|$. We set $(\varepsilon, \delta) = (1, 0.01)$. Additionally, for (a), we fix $\Delta f = 1$ and for (b), we fix $\alpha_{P',P^*} = 15$.}
  \label{fig:2_path}
\end{figure}

\subsection{Path Cardinality Effect}\label{sec:general}
In this segment, we 
are interested in understanding the disparate impacts that privacy noise has on node pairs which have many alternate path choices as opposed to node pairs which have fewer paths. We call this effect the \emph{path cardinality} effect. In this case, we measure the impact of noise by estimating the probability of realizing bias at least as large as $\beta$, given some $\beta > 0$. Again, a higher probability indicates a higher impact of noise. We now make the following conjecture:

\begin{conj}\label{conj2}
Node pairs which have a large path cardinality are, on average, more impacted by privacy noise as opposed to node pairs which have fewer alternate path options. 
\end{conj}


In order to gain insight into the above conjecture, we will present our main technical result in Theorem \ref{thm:q_beta}. Before stating the theorem, we need to introduce the following definition and set notations:  

\begin{defn}\label{betaworse_paths}($\beta$-worse paths)
Any path $P \in \cP_{ij}$ is said to be $\beta$-worse, if: 
\[
     w_G(P) \geq w_G(P^*) + \beta,
\]
where $P^*$ is the least weight path between nodes $i$ and $j$ on graph $G$.
\end{defn}
\noindent
Therefore, given $\beta > 0$, we can partition set $\cP_{ij}$ into two sets $\cP_{ij}^{\geq \beta}$ and $\cP_{ij}^{< \beta}$: 
\[
    \cP_{ij}^{\geq \beta} := \{ P \in \cP_{ij}:~ w_G(P) \geq w_G(P^*) + \beta\}
\]
\[
    \cP_{ij}^{< \beta} := \{ P \in \cP_{ij}:~ w_G(P) < w_G(P^*) + \beta\}
\]
We are now ready to present our theorem: 
\begin{theorem}\label{thm:q_beta}
Let $q_{\beta}$ be the probability that the realized bias of shortest path computation using a privatized graph $\widetilde G$ is at least $\beta$. Then $q_{\beta}$ is upper bounded as follows:
\[
      q_{\beta} \leq \sum_{P \in \cP_{ij}^{\geq \beta}}  \Phi^c \left(\frac{\alpha_{P,P^*}}{\sigma \sqrt{|S_{P,P^*}|}}  \right) \leq \left| \cP_{ij}^{\geq \beta} \right| \cdot \Phi^c\left( \frac{\beta}{\sigma \sqrt{S_{max}}} \right),
\]
where $S_{max} = \max_{P \in \cP_{ij}^{\geq \beta}} \left|S_{P,P^*} \right|$. 
\end{theorem}
\begin{proof}[Proof Sketch]
The proof idea is as follows: we can express $q_{\beta}$ as the probability of the event that there exists a path in $\cP_{ij}^{\geq \beta}$ which has the lowest weight on privatized graph $\widetilde G$. Since only one path can be the shortest path on any realization of $\widetilde G$, the above event decomposes into a union of disjoint sub-events (a specific path in $\cP_{ij}^{\geq \beta}$ is the new shortest path on $\widetilde G$). The technical parts of the proof deal with upper bounding the probability of each of these sub-events for which we use Lemma \ref{lem:prob_2path}. The detailed proof can be found in Appendix \ref{sec:proofs}.
\end{proof}

\paragraph{Observations from Theorem \ref{thm:q_beta}:} We can derive useful insights from the expression of the upper bound. It is immediate that it depends on the cardinality of the set $\cP_{ij}^{\geq \beta}$. I.e., the higher is the number of $\beta$-worse candidate paths, higher the probability that the shortest path changes to one such path which is exactly the intuition for Conjecture \ref{conj2}. The dependence on $\beta$ is actually two-fold: firstly, as $\beta$ increases, the term $\Phi^c \left( \frac{\beta}{\sigma \sqrt{S_{max}}} \right)$ decreases. Additionally, a higher $\beta$ decreases the cardinality of $\cP_{ij}^{\geq \beta}$. Essentially, this means that if $\beta$ is large, the probability that we end up shifting to a $\beta$-worse path decreases very quickly (refer to Figure \ref{fig:qbeta_vs_beta}). This idea will be explored in greater depth in Corollary \ref{corr:bound_beta}.   
\begin{figure}[!ht]
      \centering
      \includegraphics[width=0.6\textwidth]{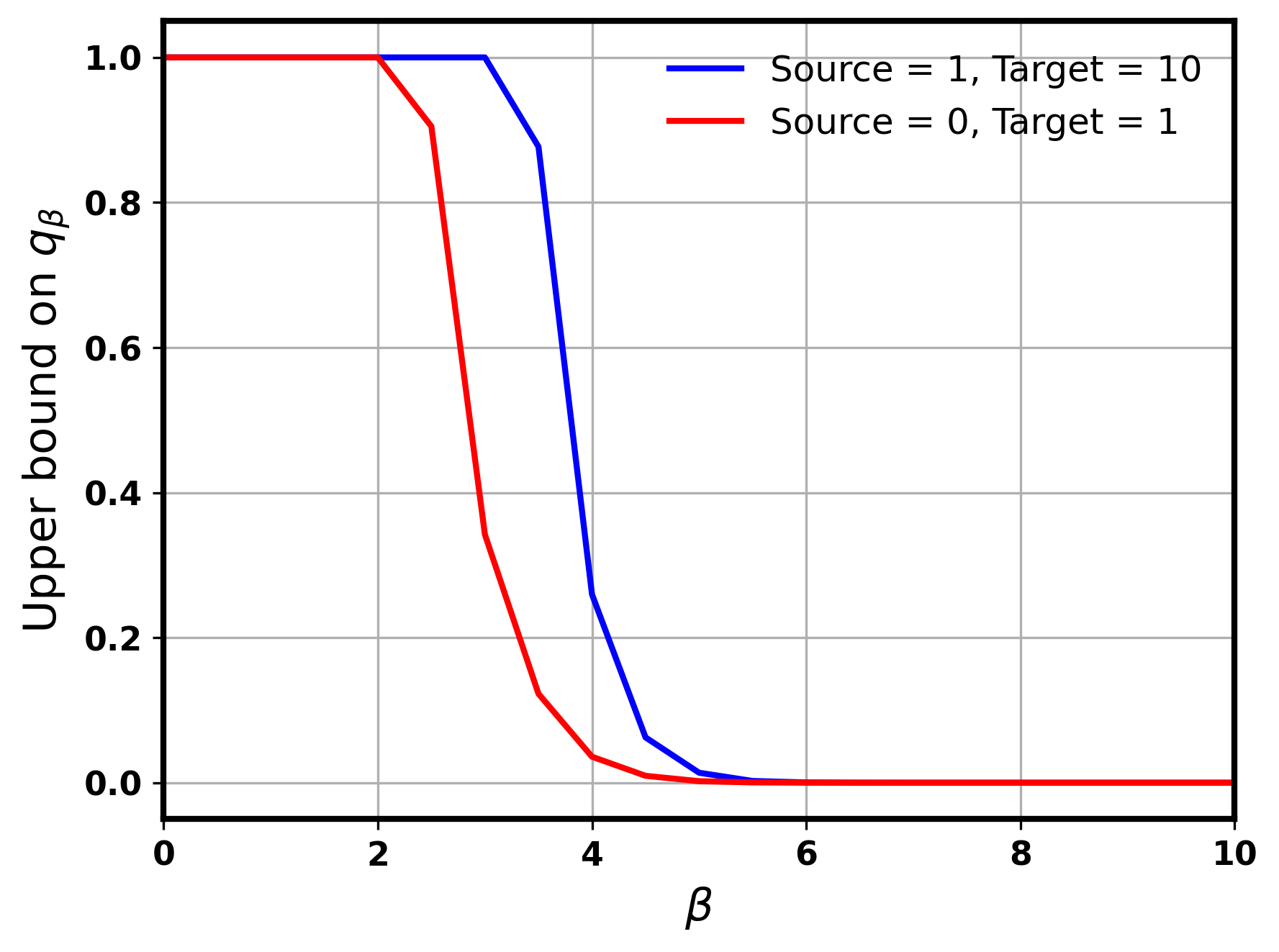}
      \caption{Evolution of the upper bound on $q_{\beta}$ as a function of $\beta$ for a wheel graph with $N = 21$. All ground truth edge weights drawn independently from $U[0, 1]$. We plot results for two types of source-destination pairs: the blue legend is for a pair of nodes which lie on diametrically opposite sides of the wheel graph, the red legend is for a pair of nodes consisting of the central node and a circumference node. The noise is sampled from a mean zero Gaussian distribution with standard deviation $\sigma = 0.3$. For very small values of $\beta$, the bound is vacuous. However, once the bound becomes non-trivial, it decreases rapidly and can be expected to approximate $q_{\beta}$ very accurately.}
      \label{fig:qbeta_vs_beta}
\end{figure}

\begin{remark}
Note that the upper bound is tight when $|\cP_{ij}^{< \beta}| = 1$ and $|\cP_{ij}^{\geq \beta}| = 1$. In this case, we recover the exact expression we derived in Lemma \ref{lem:prob_2path}, implying that our results are consistent.
\end{remark}
\noindent 

For general networks, it is not possible to improve on this bound without having additional information about the network topology. However, there are instances where $q_{\beta}$ may be computed exactly, for example, when all the paths in $\cP_{ij}$ are disjoint and have no overlapping edges. We direct the interested reader to Section \ref{sec:nonoverlap} of the Appendix where we derive the expression of $q_{\beta}$ exactly and demonstrate through numerical experiments how our bounds in Lemma \ref{thm:q_beta} compare with the exact expression (Refer to Figure \ref{fig:comp_nonoverlap}).  

\paragraph{Dual of Theorem \ref{thm:q_beta}:} 
We can also write Theorem \ref{thm:q_beta} in terms of high-probability bounds on the realized bias. Before stating the result, we formally define the notion of \textit{z-scores}:
\begin{defn}\label{defn:zscores}
For any $\eta \in [0, 1]$, we can define the \textit{z-score} corresponding to $\eta$ as follows: 
\[
           z_{\eta} = \Phi^{-1}\left( \eta\right),
\]
where $\Phi^{-1}(\cdot)$ represents the inverse of the standard normal CDF. Alternatively, $z_{\eta}$ is the value at which the standard normal CDF evaluates to $\eta$. 
\end{defn}

Our main result is then as follows:
\begin{corr}\label{corr:bound_beta}
Suppose, $B_{ij}$ is the realized bias while computing the shortest path between nodes $i$ and $j$ using a privatized graph $\widetilde G$. Then,
\[
       \pr\left[B_{ij} < \sqrt{2} \left(\sigma z^* \sqrt{S}\right) \right] \geq 1-\gamma, 
\]
where $z^* = z_{1-\frac{\gamma}{|\cP_{ij}|}}$ and $S$ denotes the maximum number of edges in any path in $\cP_{ij}$.
\end{corr}

\begin{proof}
The proof can be found in Appendix \ref{sec:proofs} and follows directly from Theorem~\ref{thm:q_beta}. 
\end{proof}
Theorem \ref{thm:q_beta} showed that as $\beta$ increases, the probability of incurring a bias at least as large as $\beta$ decreases sharply. This implies that the probability of incurring a large bias is very ``small". This is exactly what Corollary \ref{corr:bound_beta} claims. Thus, Theorem \ref{thm:q_beta} and Corollary \ref{corr:bound_beta} are duals of each other. 


\begin{figure}[!ht]
  \centering
  \raisebox{20pt}{\parbox[b]{.11\textwidth}{}}%
  \subfloat[][]{\includegraphics[width=.5\textwidth]{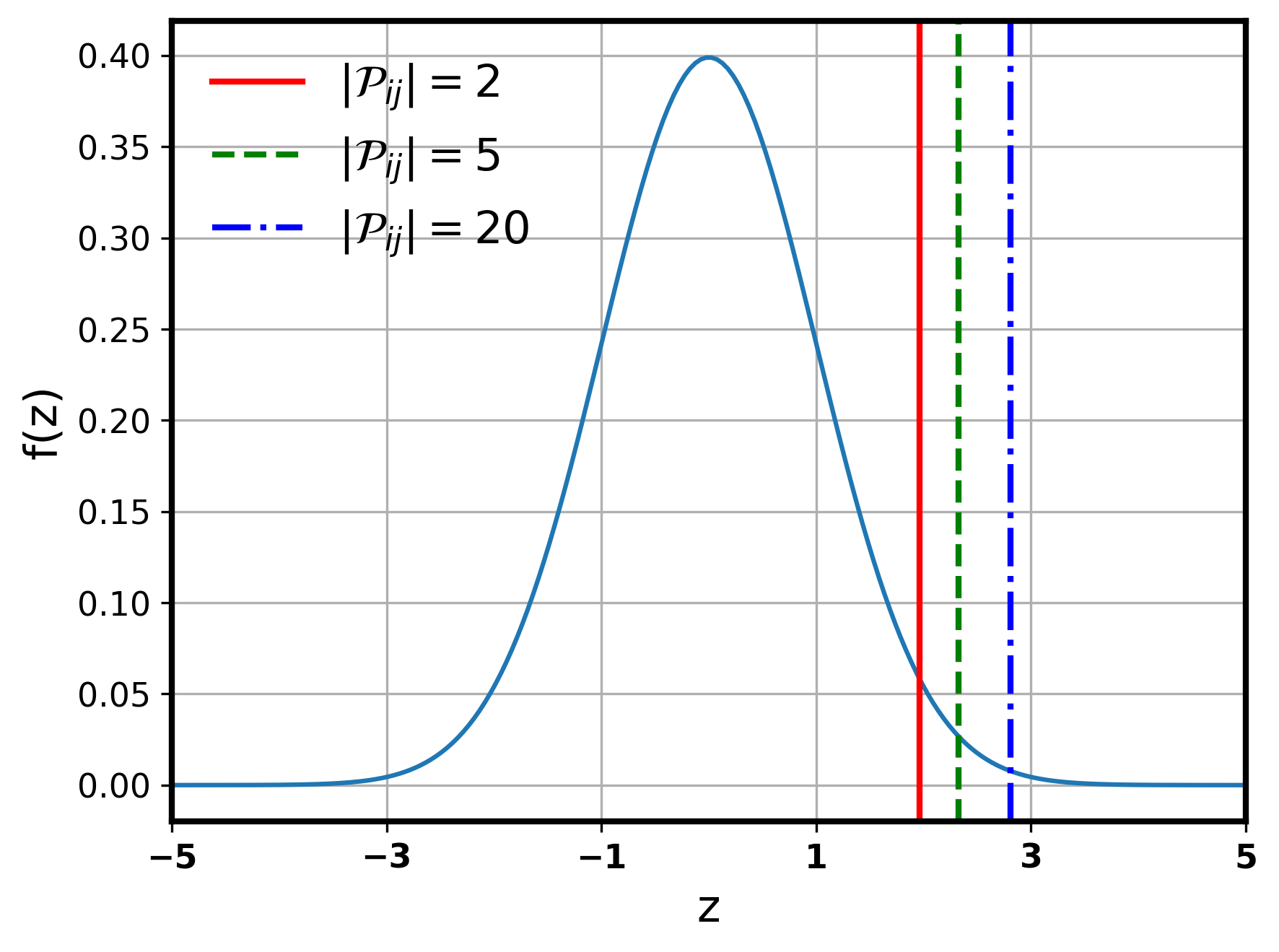}}\hfill
  \subfloat[][]{\includegraphics[width=.5\textwidth]{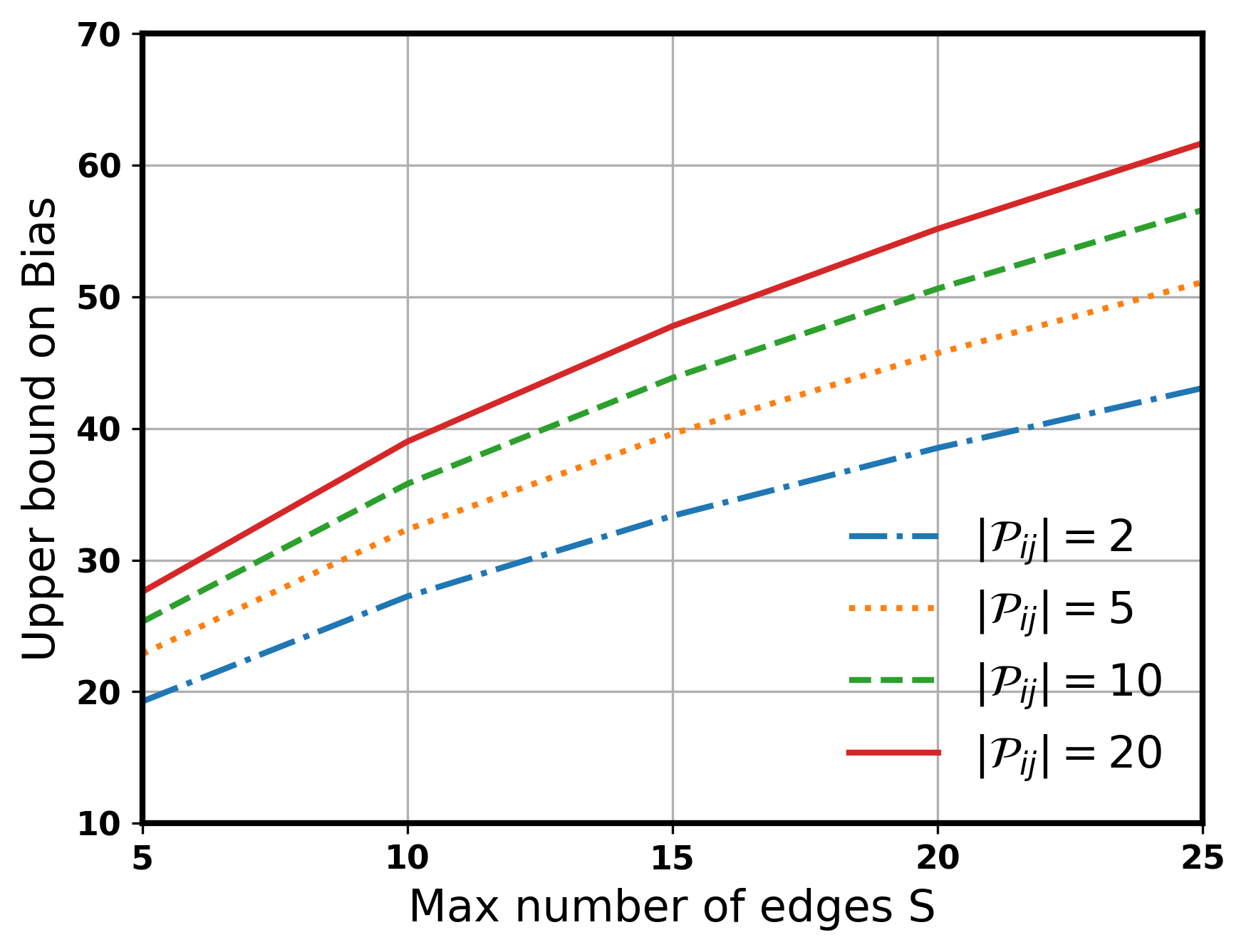}}\par
  \caption{In (a), we show how the z-scores change with the cardinality of $\cP_{ij}$. Higher values of $|\cP_{ij}|$ leads to higher z-scores. For all cases, we use $\gamma = 0.05$, i.e., we desire $95\%$ coverage. In (b), we illustrate how the bounds on bias $B_{ij}$ calculated in Corollary \ref{corr:bound_beta} vary with $S$ and $|\cP_{ij}|$. 
}
  \label{fig:delta_bound}
\end{figure}

\section{Experimental Characterization of Bias and Unfairness}\label{sec:exp}

In this section, we provide experimental results that extend and empirically validate the theoretical findings discussed above. The goal is to simulate the behavior of a DP release task on graphs that closely mimic networks in the real world focusing on the impacts of privacy on bias and fairness. 
To do so, we perform an extensive analysis of synthetic graphs, which enables us to ablate various graph parameters, including sparsity, structure, and form. 
Interestingly, in this process, we also discover that some graph classes may be more robust than others to disruptions under privacy. We present the experimental setup adopted next, in Section \ref{sec:experimental_setup}, and detail the analysis of results on 3 different classes of graphs in Section \ref{sec:results}.

\subsection{Experimental setup}
\label{sec:experimental_setup}

\paragraph{Synthetic graph generation:} 
The experiments investigate \textbf{three} different classes of graphs: 
{\bf i)} 2-dimensional grid graphs, 
{\bf ii)} wheel graphs, and 
{\bf iii)} scale-free graphs. 
While 2-D grids and wheel graphs closely emulate transportation networks in the real world (for example, Chicago and New York City have road networks that are laid out in a pattern of orthogonal grids, while road networks in cities like Paris and Rome are laid out in the shape of a wheel), scale-free graphs are often used to model other widely prevalent networks like social networks, the world wide web, friendship, etc. Thus, these graph classes cover a large variety of real-world networks. 

\paragraph{Parametrizations of each graph class:} We use the following sets of parameters to generate synthetic networks for each graph class: 
\begin{itemize}[leftmargin=*, parsep=0pt, itemsep=0pt, topsep=0pt]
    \item 2-D grid graphs: 
    A grid graph of size $N$ has $N^2$ nodes and $2N^2 + N$ edges. An illustration is provided in  Figure \ref{fig:schematic} (top left).
    \item Wheel graphs: 
    These graphs are described by the number of nodes $N$ and the ratio  $r$ of the spoke edge weights to the circumference edge weights ($r \geq 1$). The central node is by default, indexed $0$. A higher  $r$ indicates that the spoke edges have much higher weights compared to circumference edges. For example, in road networks, these edges may experience higher traffic and thus have higher ground truth weights.
    An illustration is provided in Figure \ref{fig:schematic} (top right). 
    \item Scale-free graphs: These graphs have a degree distribution following a power law and are parametrized by their size (number of nodes $ N $) and the exponent of the power law ($ \gamma $). A higher $ \gamma $ indicates very few high-degree nodes, characteristic of many real-world networks like social networks. Unlike 2-D grid and wheel graphs, scale-free graphs are random, meaning that even with the same parameters, different graph topologies may be generated in different instances. Figure \ref{fig:schematic} illustrates four scale-free graphs with varying power $\gamma $.    
\end{itemize}

  \begin{figure}[!tb]
    \centering
    \raisebox{35pt}{\parbox[b]{.05\textwidth}{}}%
    \subfloat[][2-D grid with $N = 5$]{\includegraphics[width=.35\textwidth]{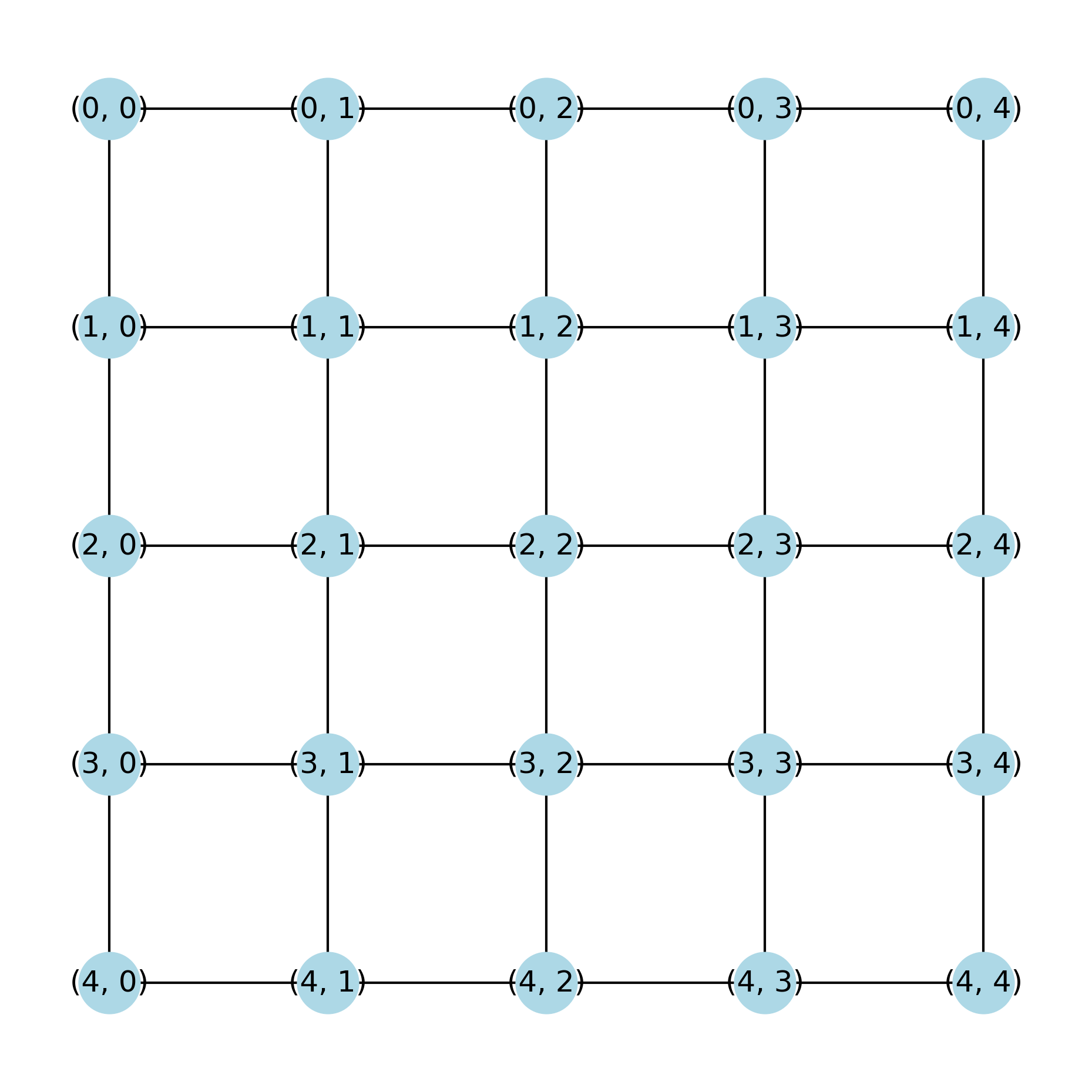}}\hfill
    \subfloat[][Wheel with $N = 10$]{\includegraphics[width=.35\textwidth]{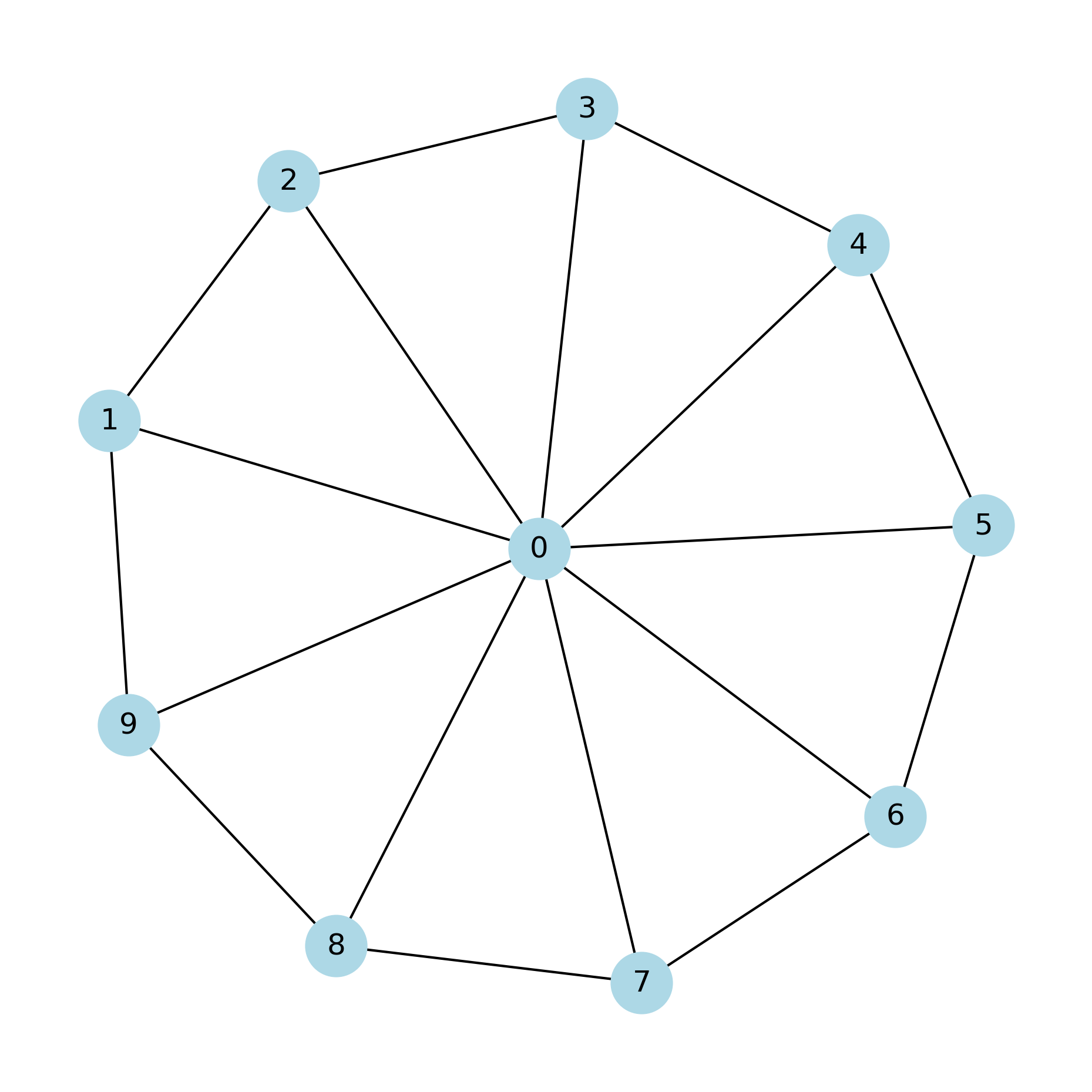}}\par
    \raisebox{35pt}{\parbox[b]{.05\textwidth}{}}%
    \subfloat[][Scale free, $\gamma = 1.1$]{\includegraphics[width=.23\textwidth]{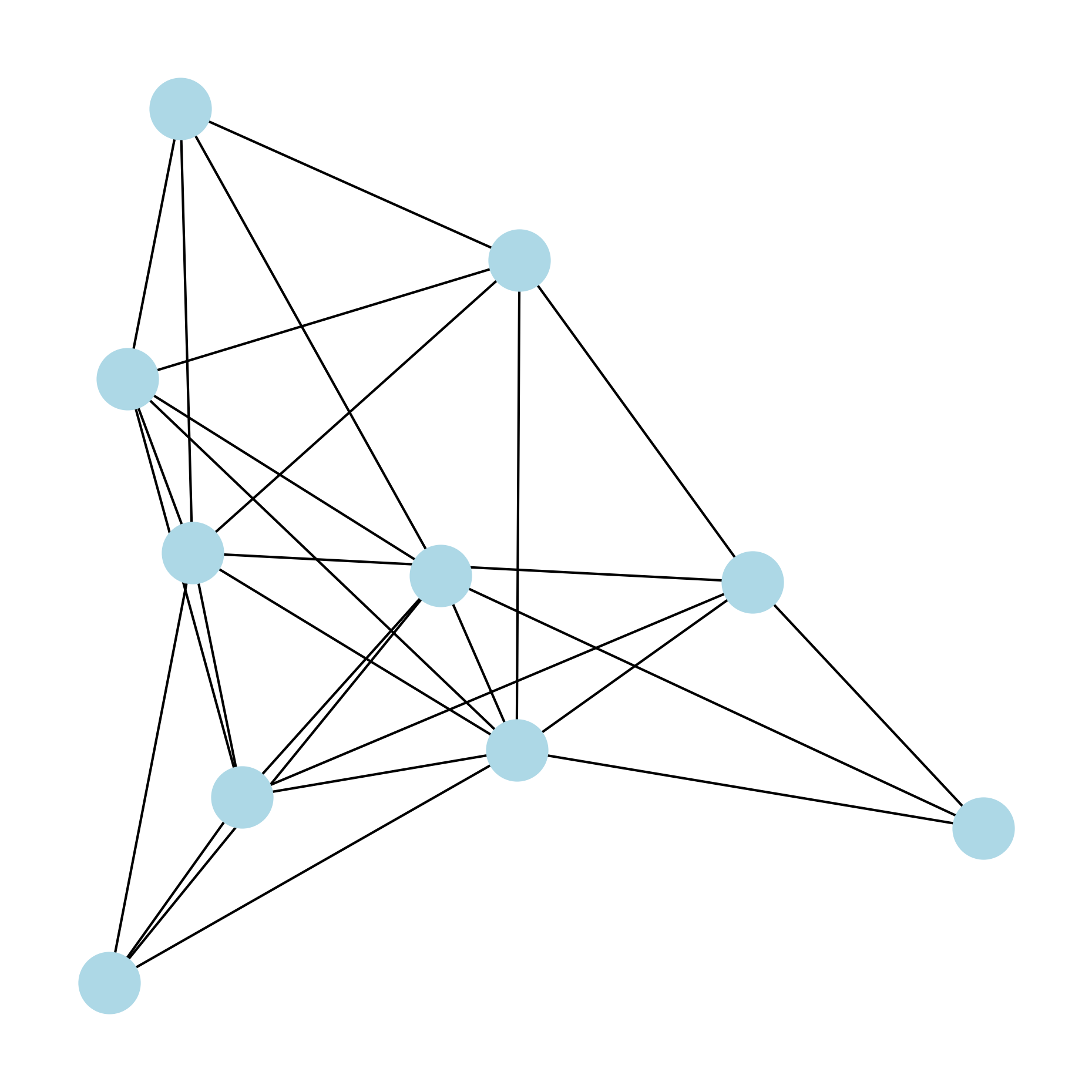}}\hfill
    \subfloat[][Scale free, $\gamma = 1.5$]{\includegraphics[width=.23\textwidth]{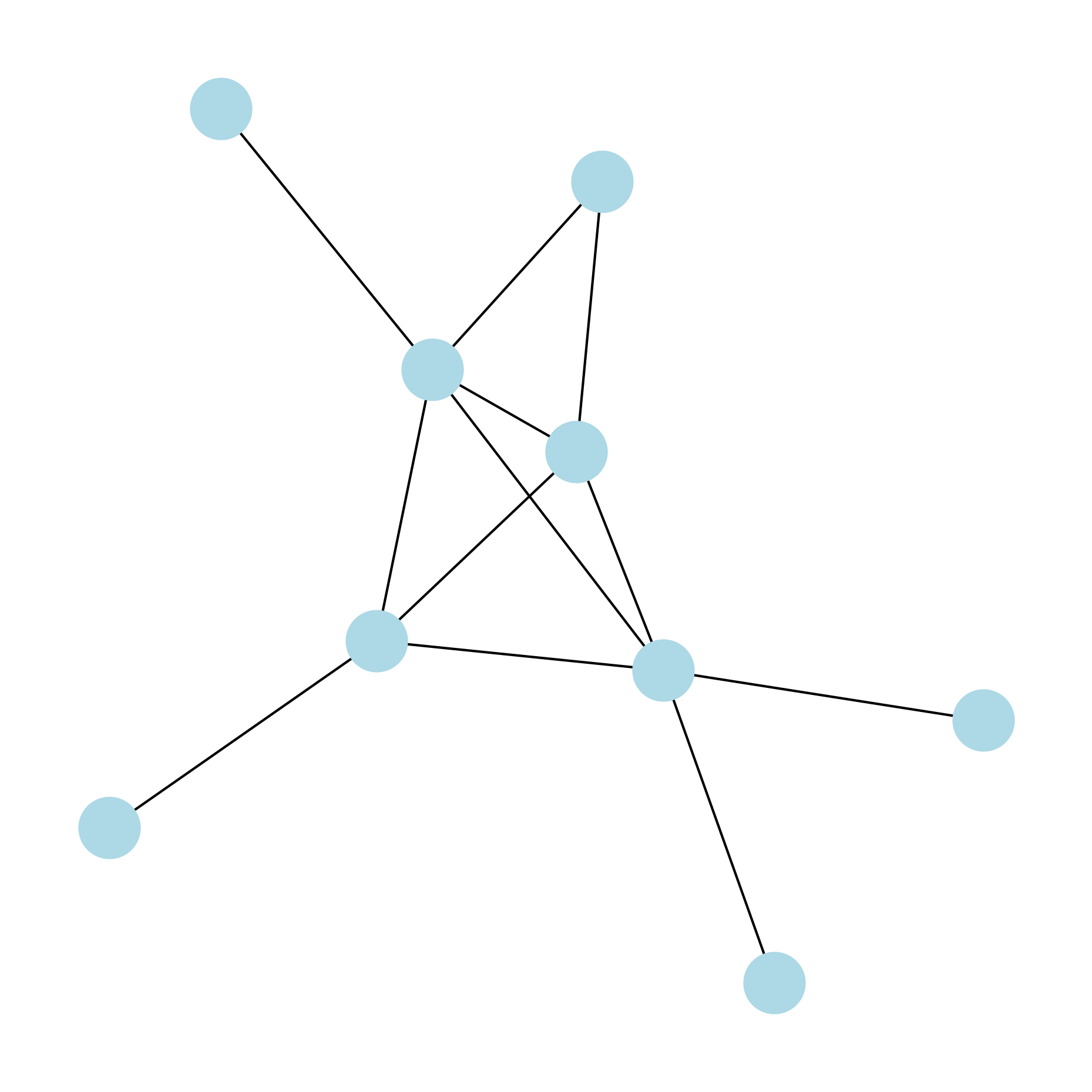}}\hfill
    \subfloat[][Scale free, $\gamma = 2$]{\includegraphics[width=.23\textwidth]{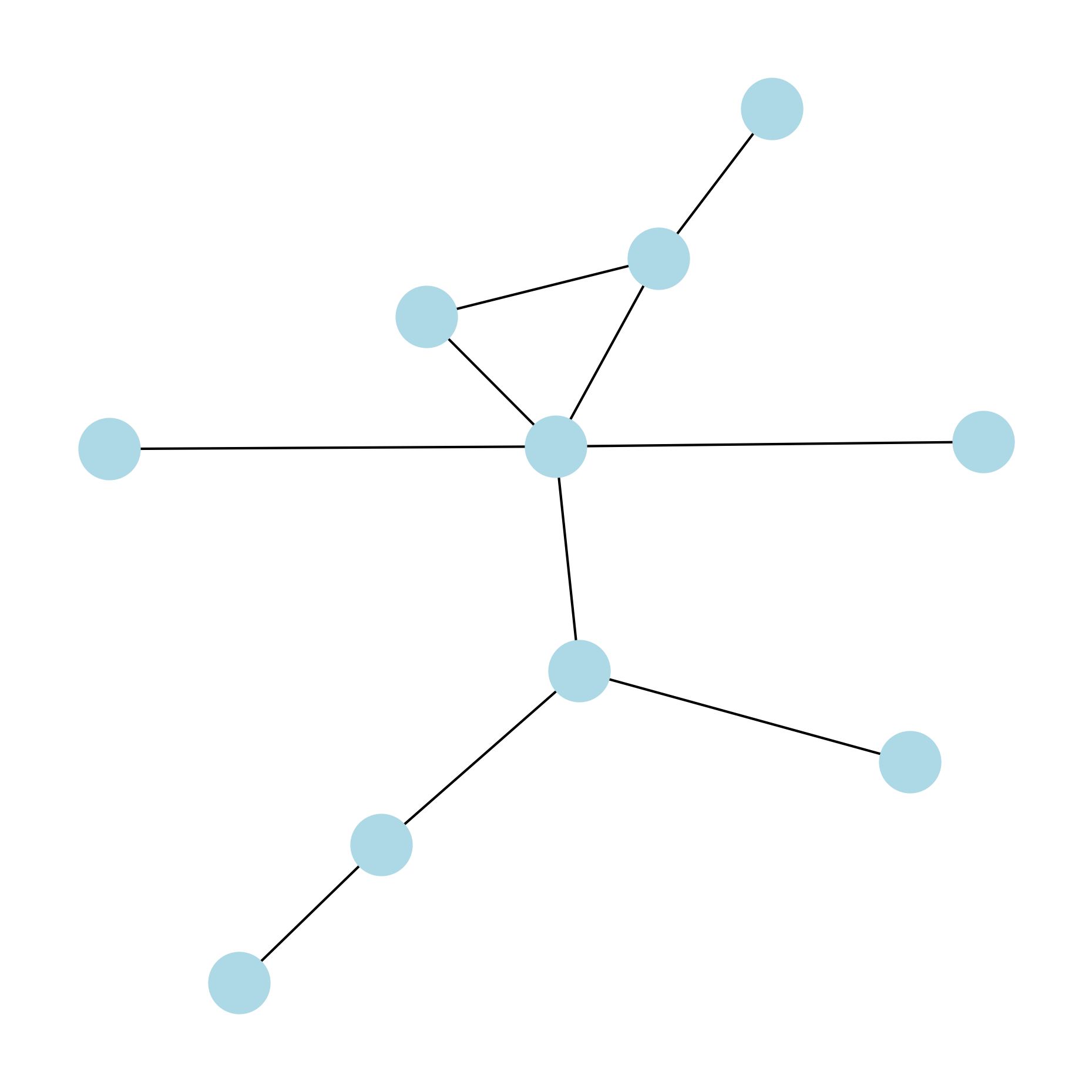}}\hfill
    \subfloat[][Scale free, $\gamma = 3$]{\includegraphics[width=.23\textwidth]{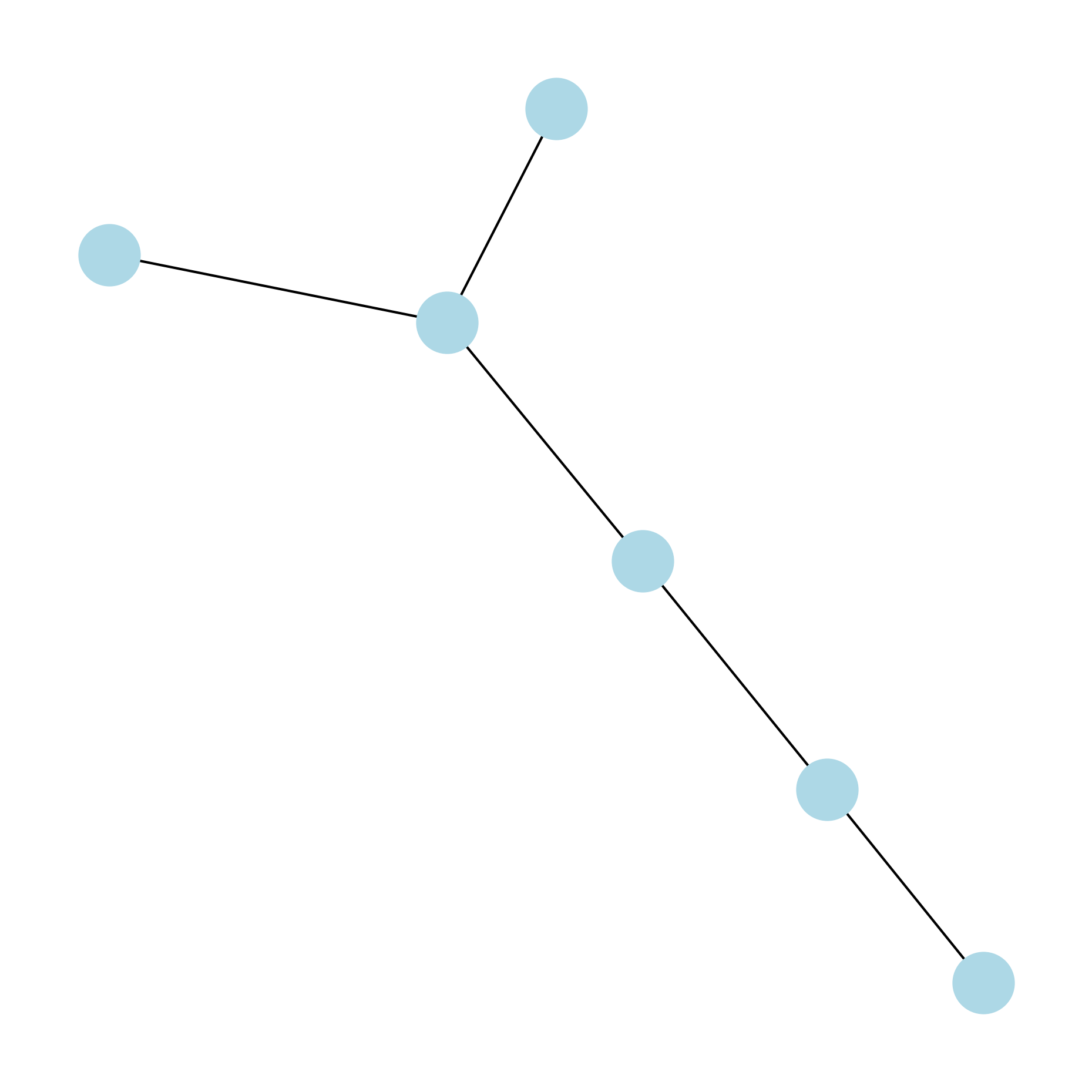}}\par
    \caption{Schematics of the $3$ graph classes introduced earlier. $2$-D grid and wheel graphs are self-explanatory. All scale free graphs are generated with $N = 10$ nodes. Since randomly generated scale free graphs can generate disconnected sub-graphs, the experiments select the largest connected component in each case (resulting in a mismatch in the number of nodes, despite the same starting $N$). Observe that when $\gamma$ is close to $1$, many nodes have very high degrees and the graph becomes more dense. However as $\gamma$ increases, the number of nodes with high degree decreases and the graph becomes more and more tree-like.}
    \label{fig:schematic}
    \end{figure}

\paragraph{Implementation of privacy model:} Given a ground truth graph, we generate $100$ private counterparts by perturbing each edge with additive noise from the standard Gaussian distribution with the desired variance. This, in turn, is a function of privacy parameters $\epsilon$ and $\delta$. As outlined earlier in Section \ref{sec:model}, a post-processing step $\max \left(0, w(e) + Z(e) \right)$ is applied to produce noisy edges $\tilde{w}(e)$ to ensure non-negativity. 
The reported results are averaged over all private realizations of a graph. 


\subsection{Results \& Insights}
\label{sec:results}

\paragraph{Metrics:} Given a graph $G$, we aim to empirically estimate the probability that a randomly chosen node-pair $i-j$ experiences a certain level of relative bias in its shortest path computation under privacy noise. We consider the following levels of relative bias: i) $0\%$ (indicating the shortest path remains unchanged), ii) $0-10\%$, iii) $10-20\%$, iv) $20-40\%$, v) $40-60\%$, vi) $60-100\%$, and vii) $>100\%$. We classify node-pairs by first computing the shortest path weight between all pairs of distinct nodes on $G$ and constructing the weight distribution of these paths. Each node-pair is then categorized based on the quartile of the weight distribution in which its true shortest path weight lies. We will refer to these categories as \Qone, \Qtwo, \Qthree, and \Qfour. \Qone, includes node-pairs whose shortest path weight lies in the first quartile (nodes are very close), while \Qfour\, includes those in the last quartile (nodes are far apart). This categorization allows us to investigate whether privacy noise impacts node pairs differently based on their distance. When presenting our observations, we often compare $\Qone{}$ and $\Qfour{}$ pairs because they represent the two extremes of the spectrum and are expected to have the maximum amount of disparity; however, we note that all trends are gradual as we go from \Qone{} to \Qfour{}. 

In the remainder of this section, we present extensive numerical results for the three classes of graphs described earlier across a wide range of parameter combinations. We rigorously analyze these results, highlighting scenarios where privacy introduces disparate impacts across different groups of node pairs and providing intuition for these effects. We also identify scenarios where certain groups are more robust to privacy noise, formally characterizing \emph{robustness} as the \emph{empirical likelihood of not being affected by privacy noise}.
   
\subsubsection{2D grid graphs.} 
The first class of graphs we explore is the $2-$D grid graph. For each ground truth graph instance, the edge weights are drawn independently from a \emph{Uniform}$[0, 1]$ distribution. The two main parameters of interest here are i) size of grid $N$ and ii) the variance (or standard deviation) of the privacy noise added. We generate results for $3$ different grid sizes $N = 10$ (with $100$ nodes), $N = 20$ (with $400$ nodes) and $N = 40$ (with $1600$ nodes). Similarly, we simulate for $4$ different levels of noise (we do not record standard deviation in absolute terms, rather we express it relative to the mean edge weight): $20~\%$, $50~\%$, $100~\%$ and $200~\%$. Refer to Figure \ref{fig:2Dgrid} for the results, in each row, the grid size remains fixed and the level of noise increases from $20~\%$ to $200~\%$ from left to right, while in each column, the grid sizes increase from $10$ to $40$ at a fixed level of noise. We make the following observations:

As the level of noise increases from left to right, node-pairs across all categories are more likely to incur a strictly positive relative bias. This follows directly from Lemma \ref{lem:prob_2path}: for any node pair $(i,j)$ and any path $P$, a higher noise level leads to a higher probability that $w_{\widetilde G}(P) < w_{\widetilde G}(P^*)$. Aggregating over all paths in $\cP_{ij}$, the overall probability of a strictly positive relative bias increases.

However, there is a clear disparity between the source-destination pairs in \Qone{} and those in \Qfour{}. At any noise level, \Qone{} pairs are much more likely to remain unaffected compared to \Qfour{} pairs. \Qfour{} pairs usually represent nodes that are very far apart. On 2-D grid graphs, pairs of nodes that are farther apart have a larger set of alternative paths (higher $\left| \cP_{ij} \right|$) and a higher number of edges on these paths (higher $S_{max}$), thus facing a higher risk of being affected by privacy noise. Here, the \emph{path cardinality effect} explained in Section \ref{sec:2_pathcase} overtakes the \emph{effective relative noise} effect, in favor of shorter paths. 

The above disparity in empirical probability estimates is particularly amplified at low noise levels. When the noise level is low, \Qfour{} pairs still have a higher chance of being affected due to more edges (higher $S_{max}$). However, at high levels of noise, the path weights are so distorted that the ordering of paths on privatized graph $\widetilde G$ does not reveal any information about the true ordering. This greatly increases the likelihood of picking the wrong path (almost) across all categories of node pairs, reducing the disparity as we move rightwards in the figure. This follows from Lemma \ref{lem:prob_2path} which shows that a higher $\sigma$ increases the probability of picking wrong paths.

These trends are consistent across graph sizes $N$. However, as the grid size increases, the bar plots become increasingly right-heavy. This indicates that for the same noise level, a larger graph is more likely to induce higher magnitudes of relative bias across all categories of node pairs. This is again a consequence of the \emph{path cardinality effect} which is amplified on large graphs.

\begin{figure}[!ht]
  \centering
  \raisebox{35pt}{\parbox[b]{.03\textwidth}{}}%
  \subfloat[][$N = 10$, \textsf{Std }$20\%$]{\includegraphics[width=.24\textwidth]{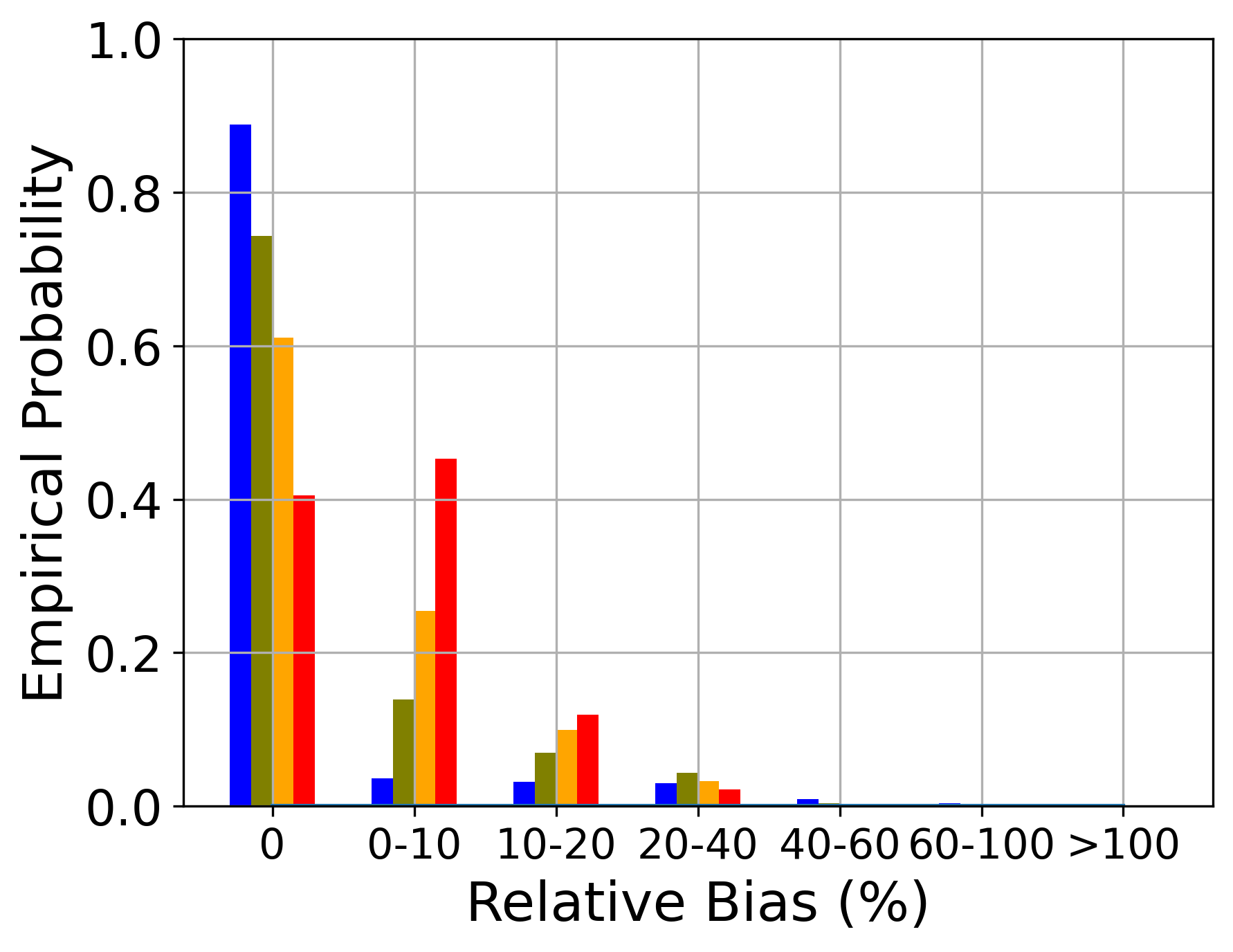}}\hfill
  \subfloat[][$N = 10$, \textsf{Std }$50\%$]{\includegraphics[width=.24\textwidth]{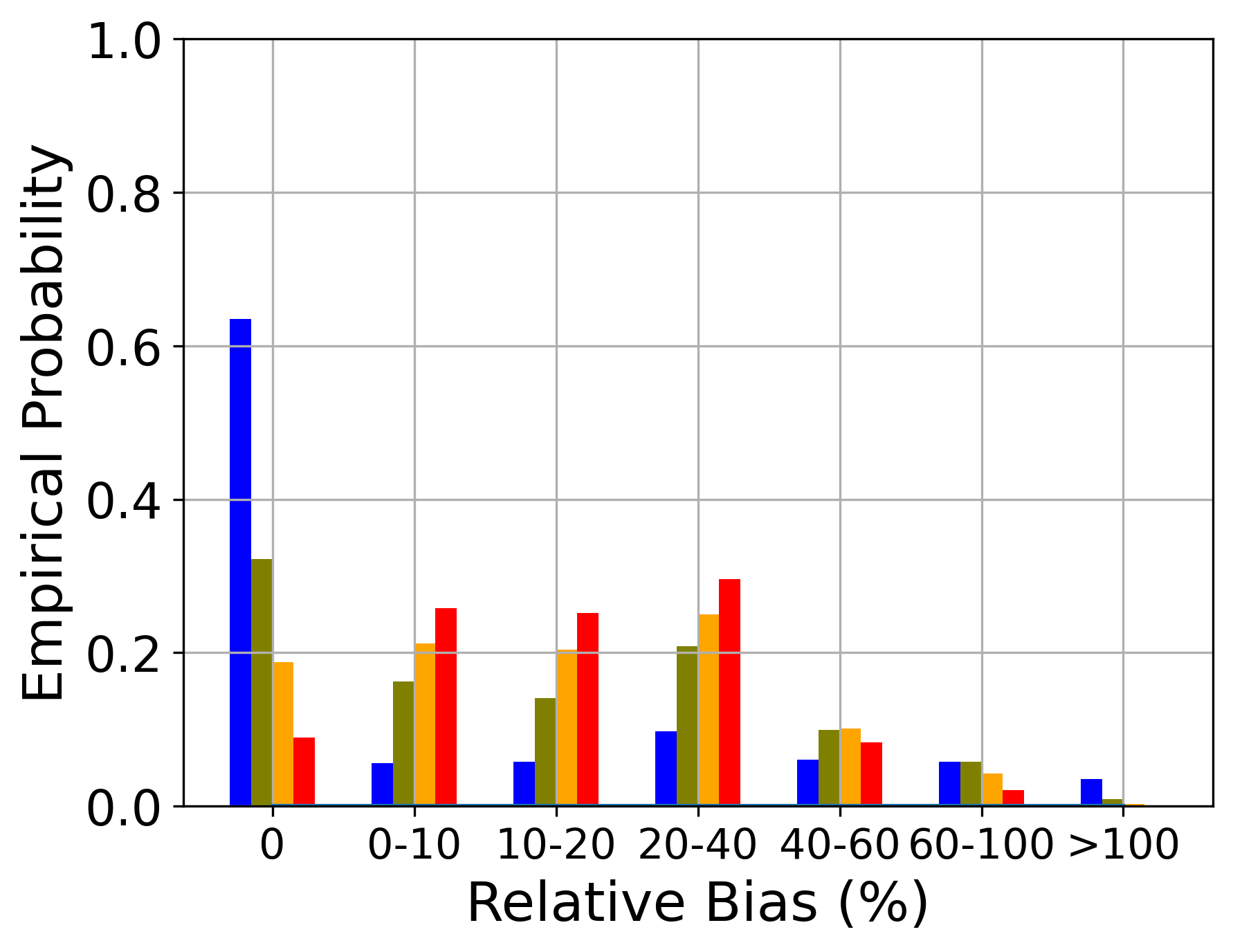}}\hfill
  \subfloat[][$N = 10$, \textsf{Std }$100\%$]{\includegraphics[width=.24\textwidth]{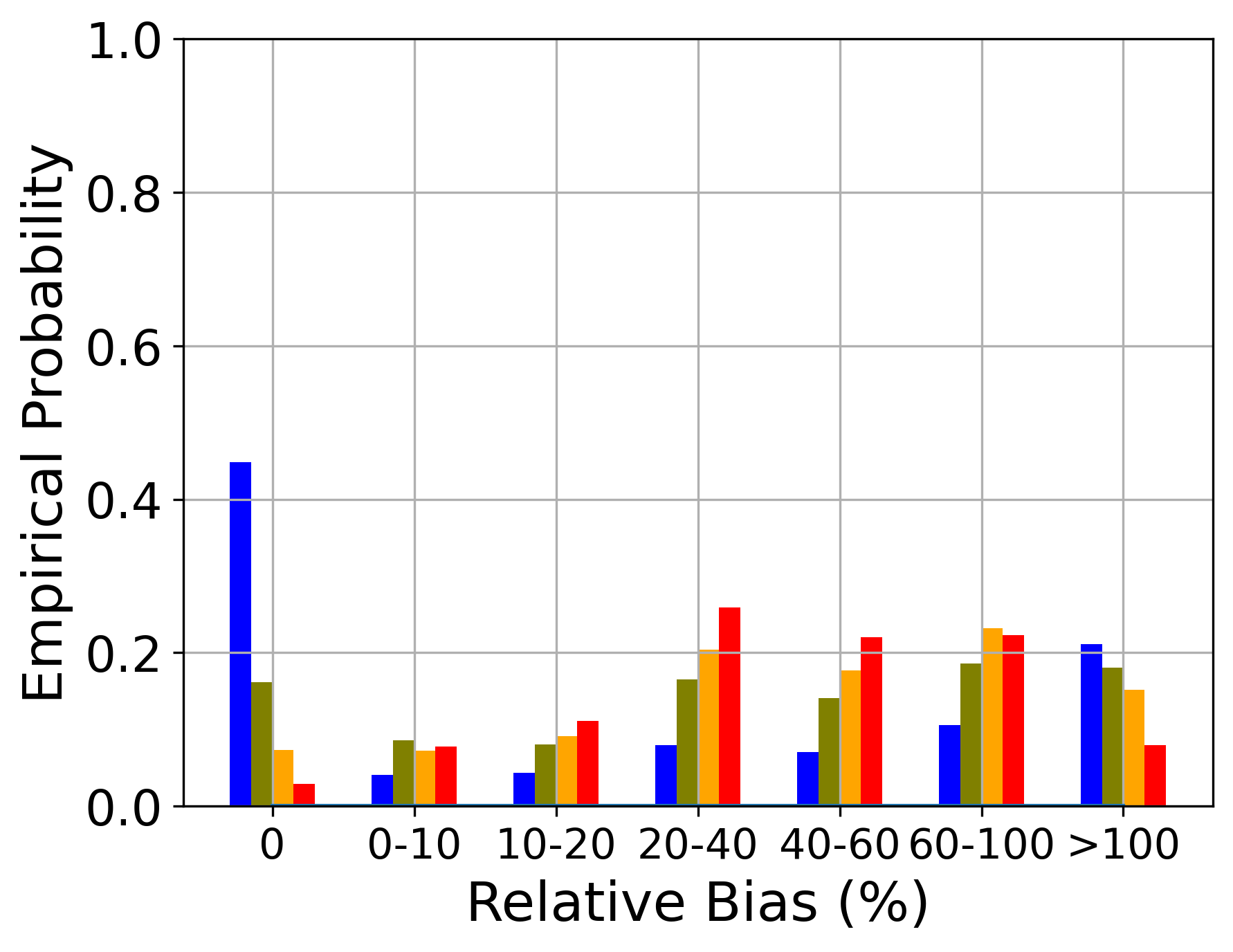}}\hfill
  \subfloat[][$N = 10$, \textsf{Std }$200\%$]{\includegraphics[width=.24\textwidth]{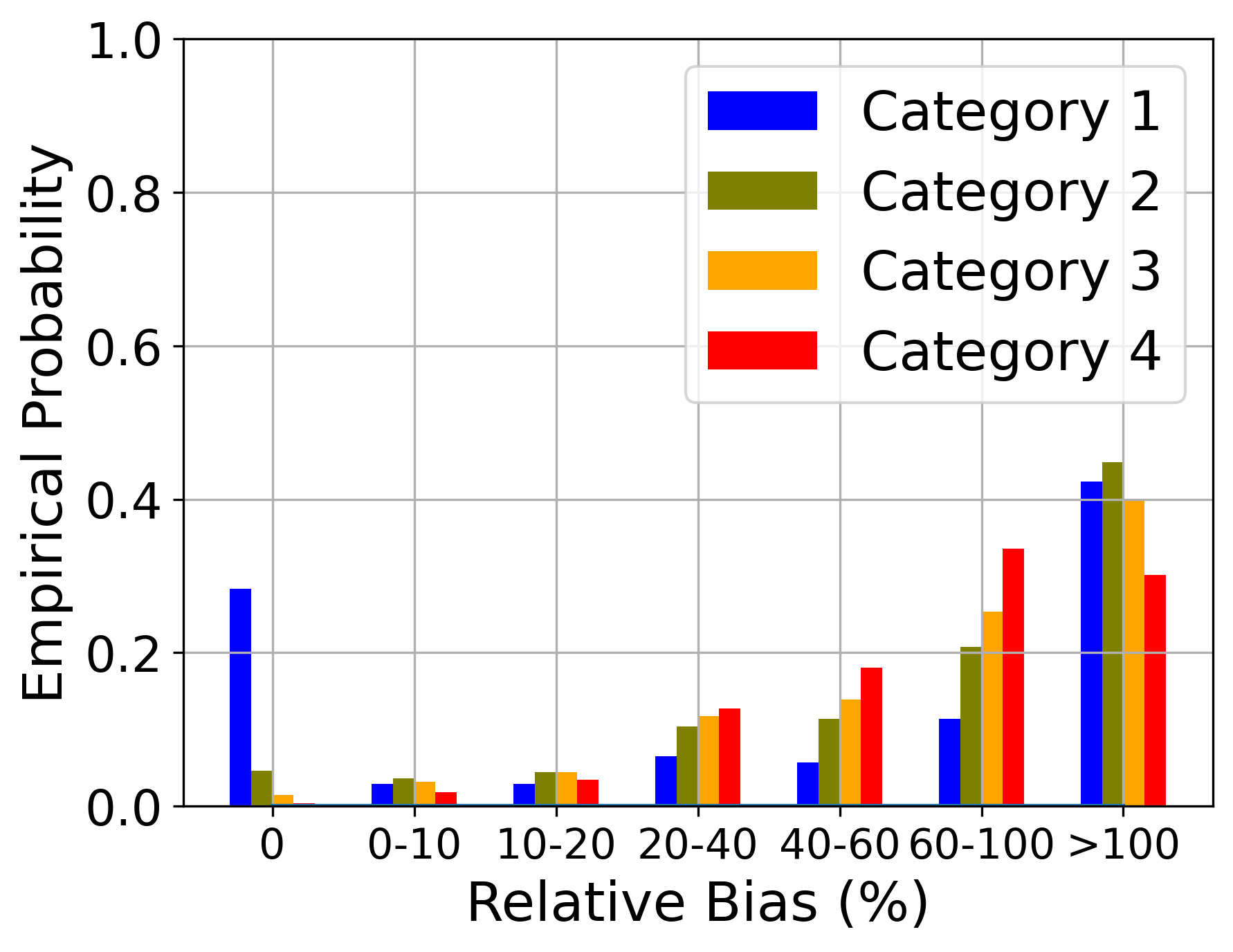}}\par
  \raisebox{35pt}{\parbox[b]{.03\textwidth}{}}
  \subfloat[][$N = 20$, \textsf{Std }$20\%$]{\includegraphics[width=.24\textwidth]{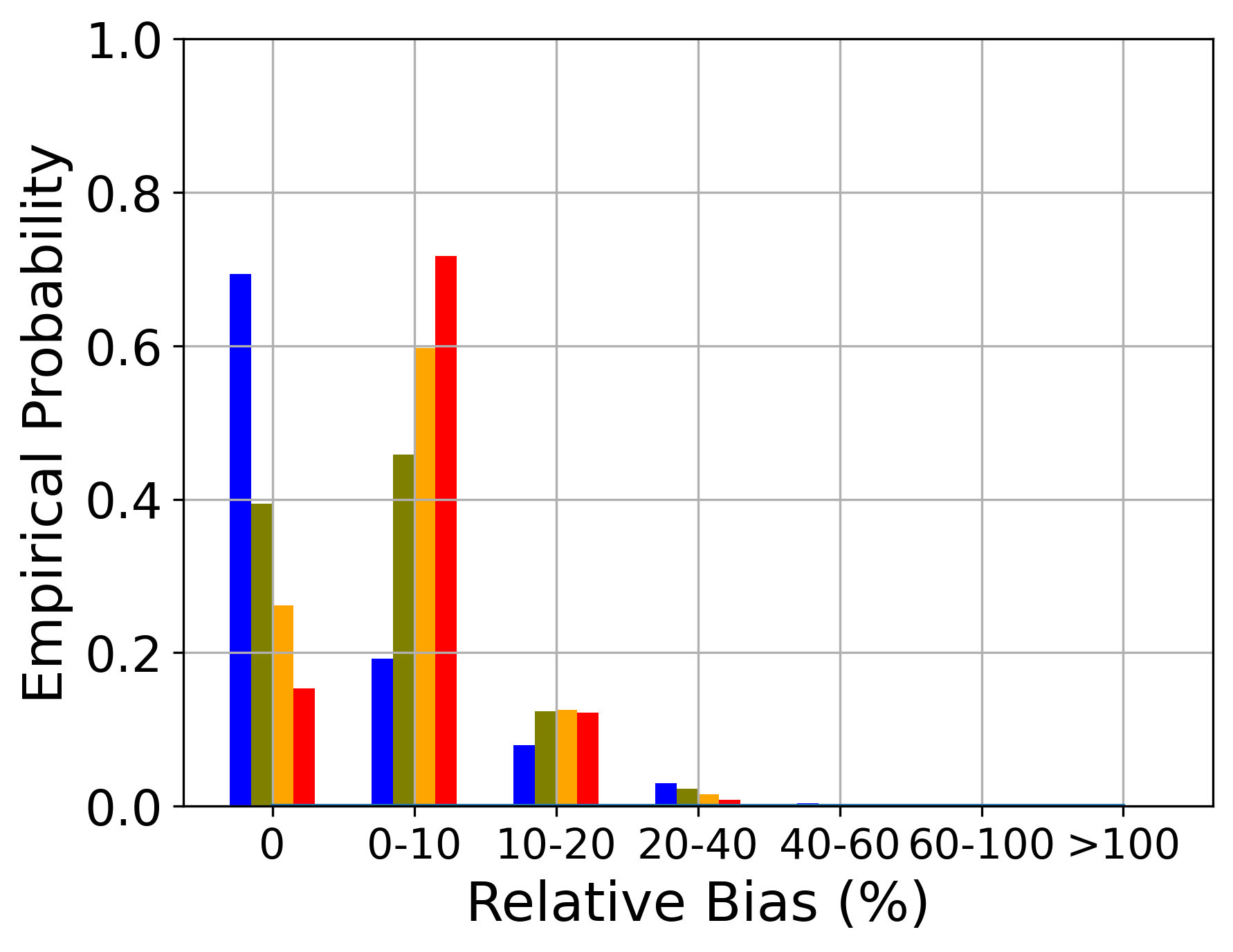}}\hfill
  \subfloat[][$N = 20$, \textsf{Std }$50\%$]{\includegraphics[width=.24\textwidth]{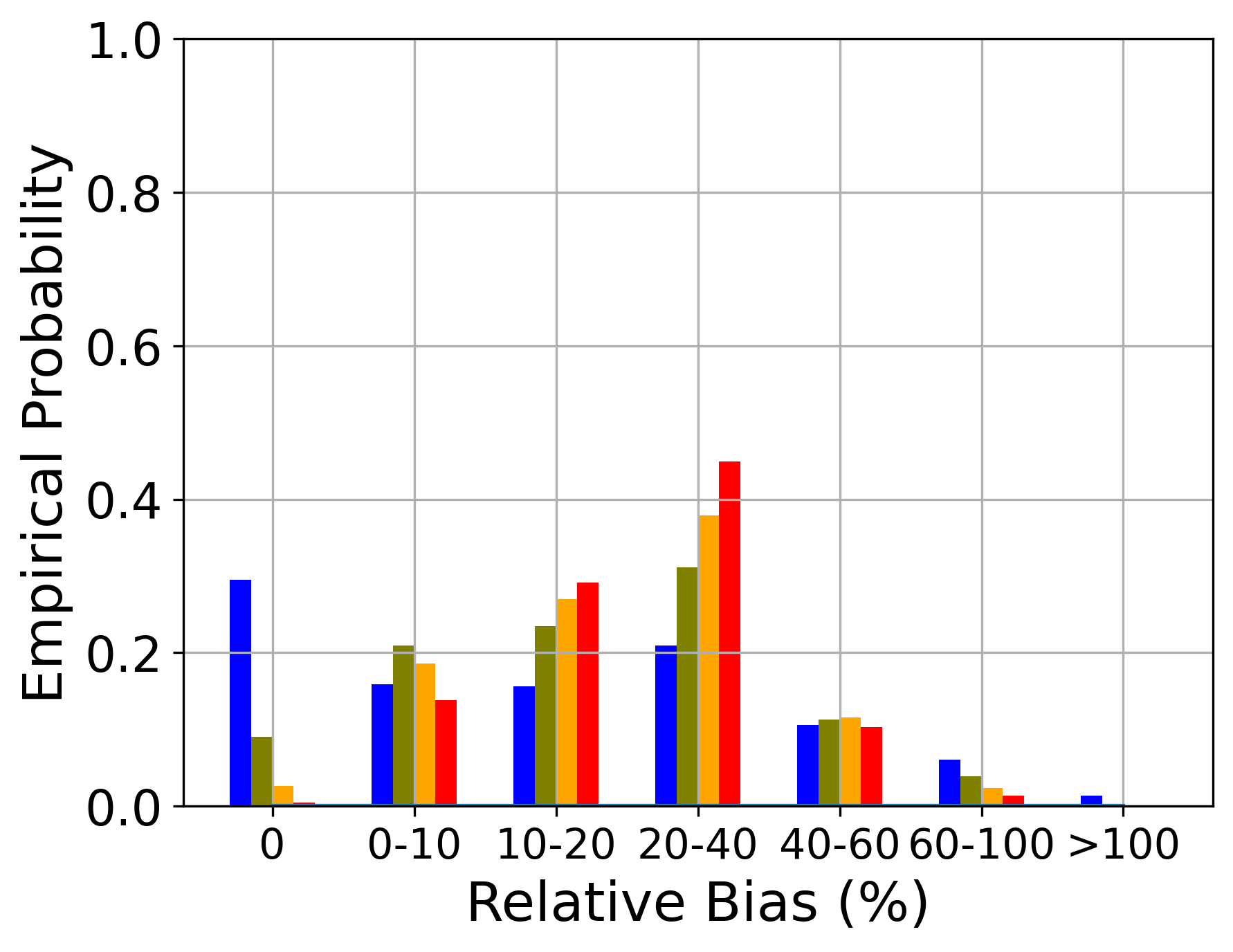}}\hfill
  \subfloat[][$N = 20$, \textsf{Std }$100\%$]{\includegraphics[width=.24\textwidth]{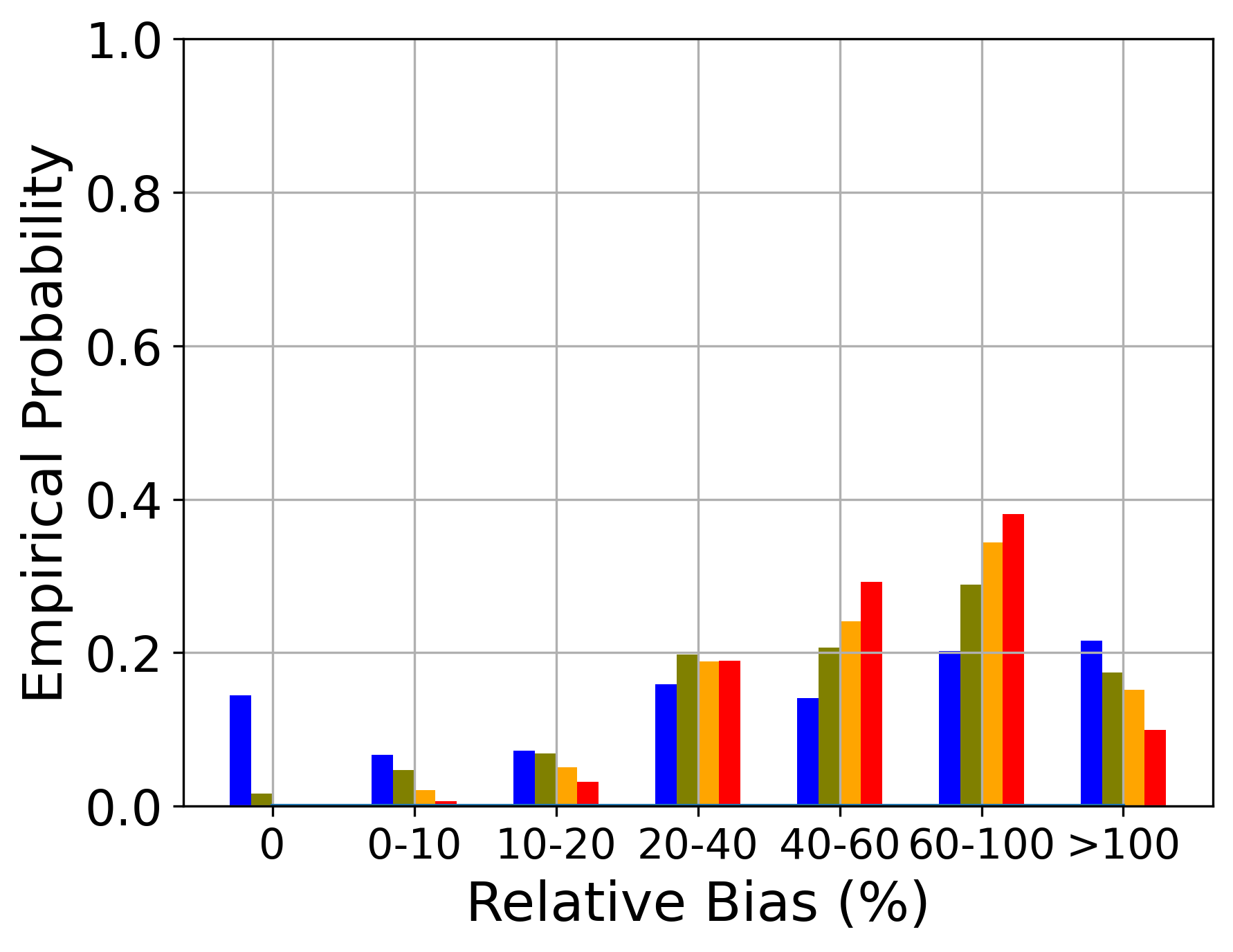}}\hfill
  \subfloat[][$N = 20$, \textsf{Std }$200\%$]{\includegraphics[width=.24\textwidth]{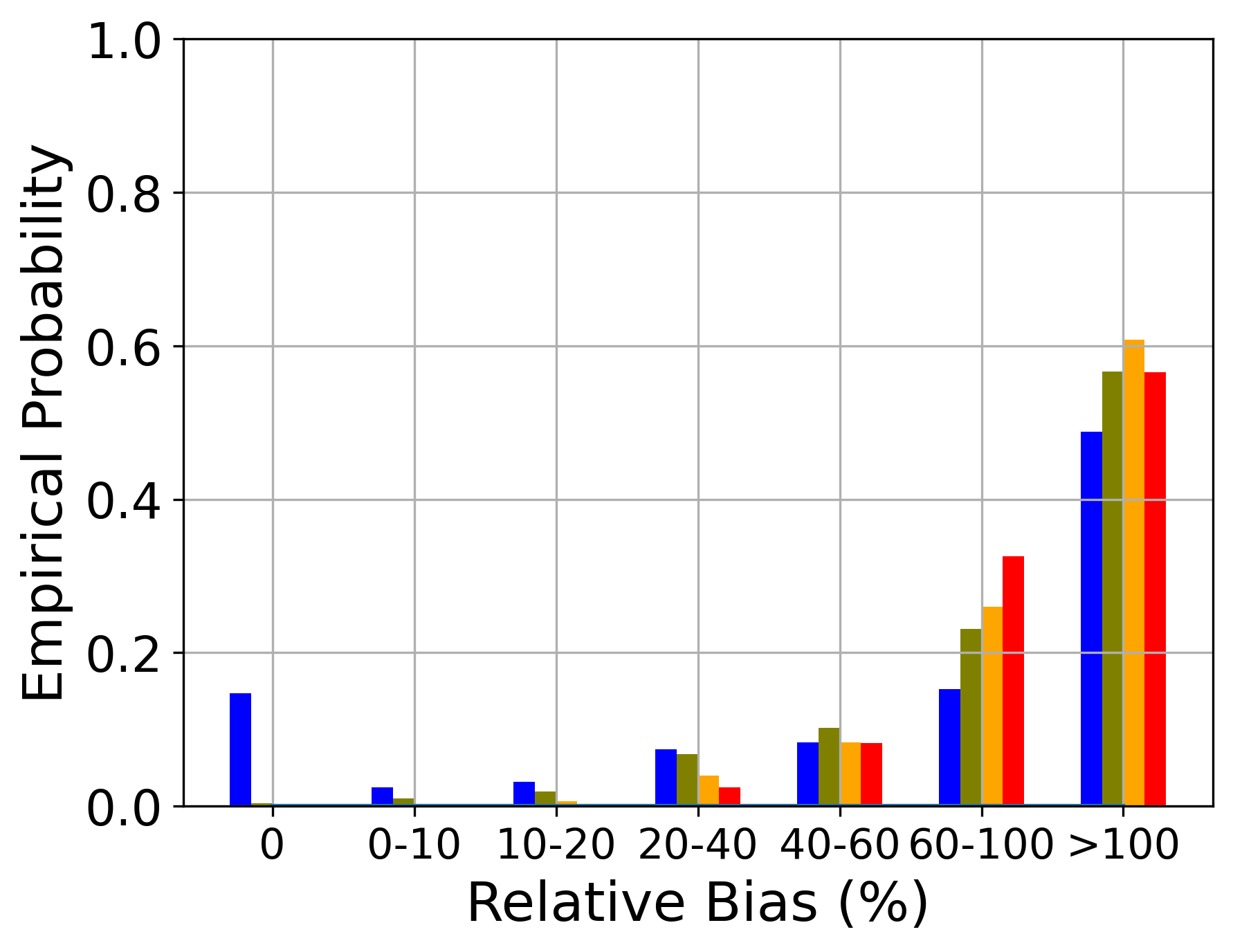}}\par
  \raisebox{35pt}{\parbox[b]{.03\textwidth}{}}
  \subfloat[][$N = 40$, \textsf{Std }$20\%$]{\includegraphics[width=.24\textwidth]{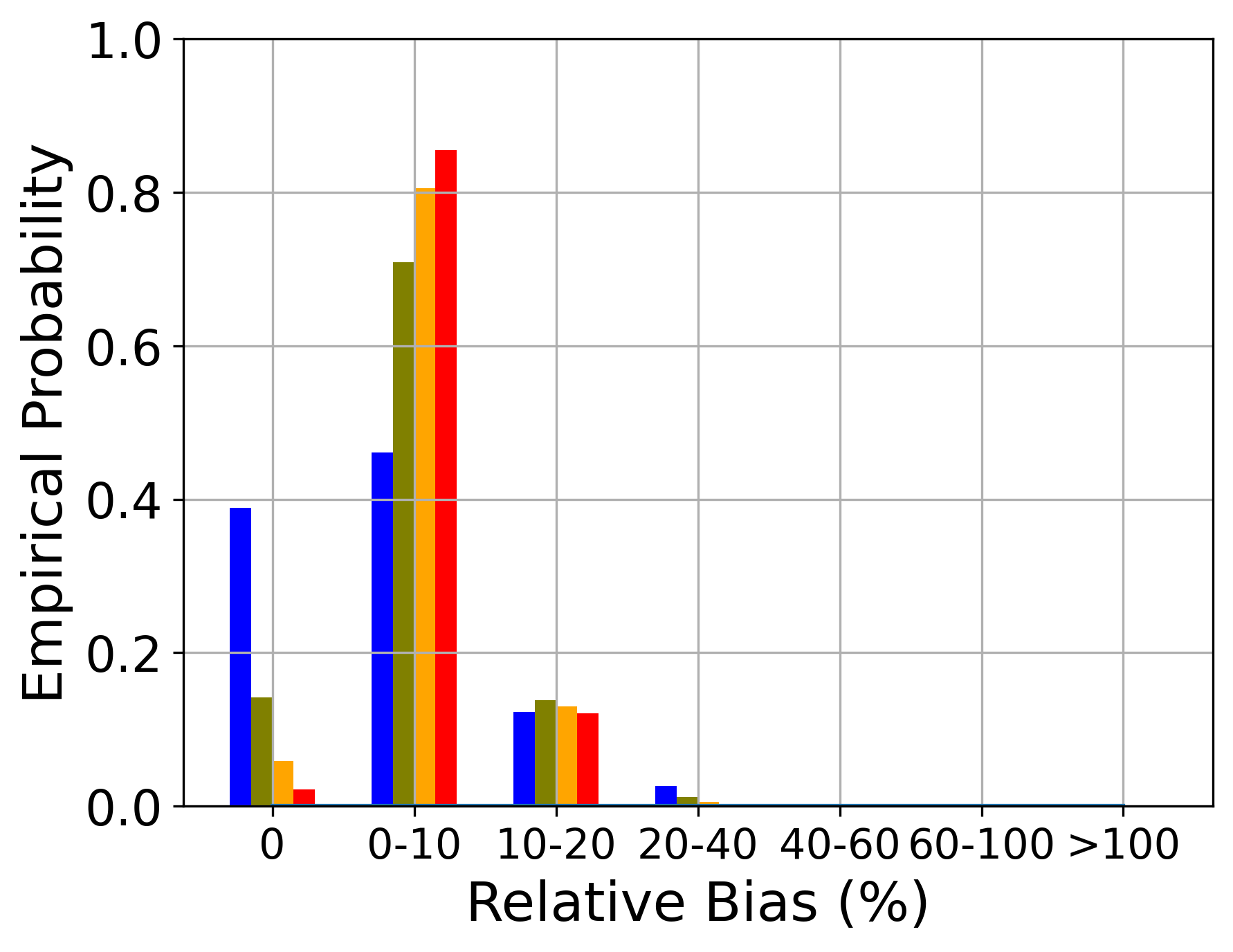}}\hfill
  \subfloat[][$N = 40$, \textsf{Std }$50\%$]{\includegraphics[width=.24\textwidth]{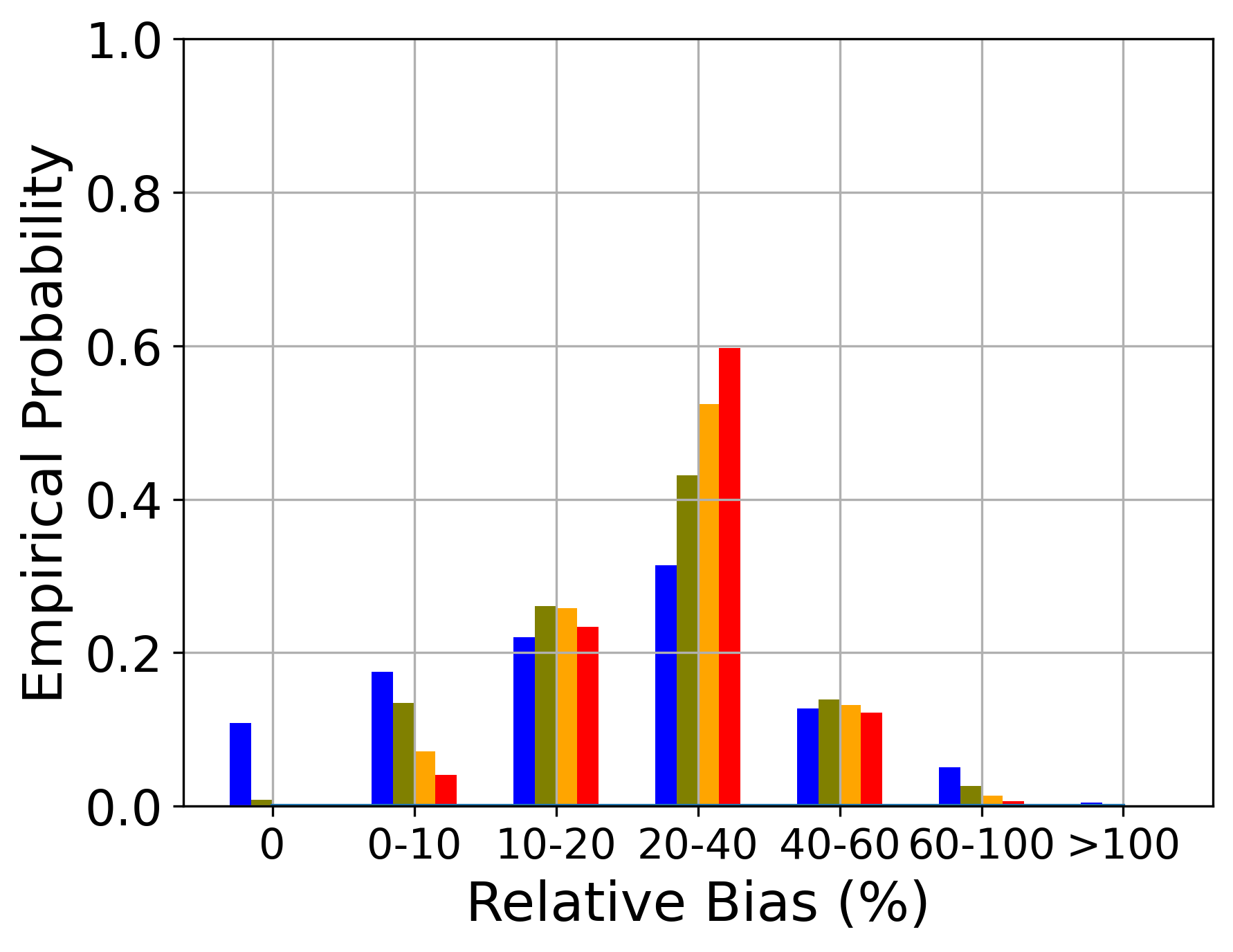}}\hfill
  \subfloat[][$N = 40$, \textsf{Std }$100\%$]{\includegraphics[width=.24\textwidth]{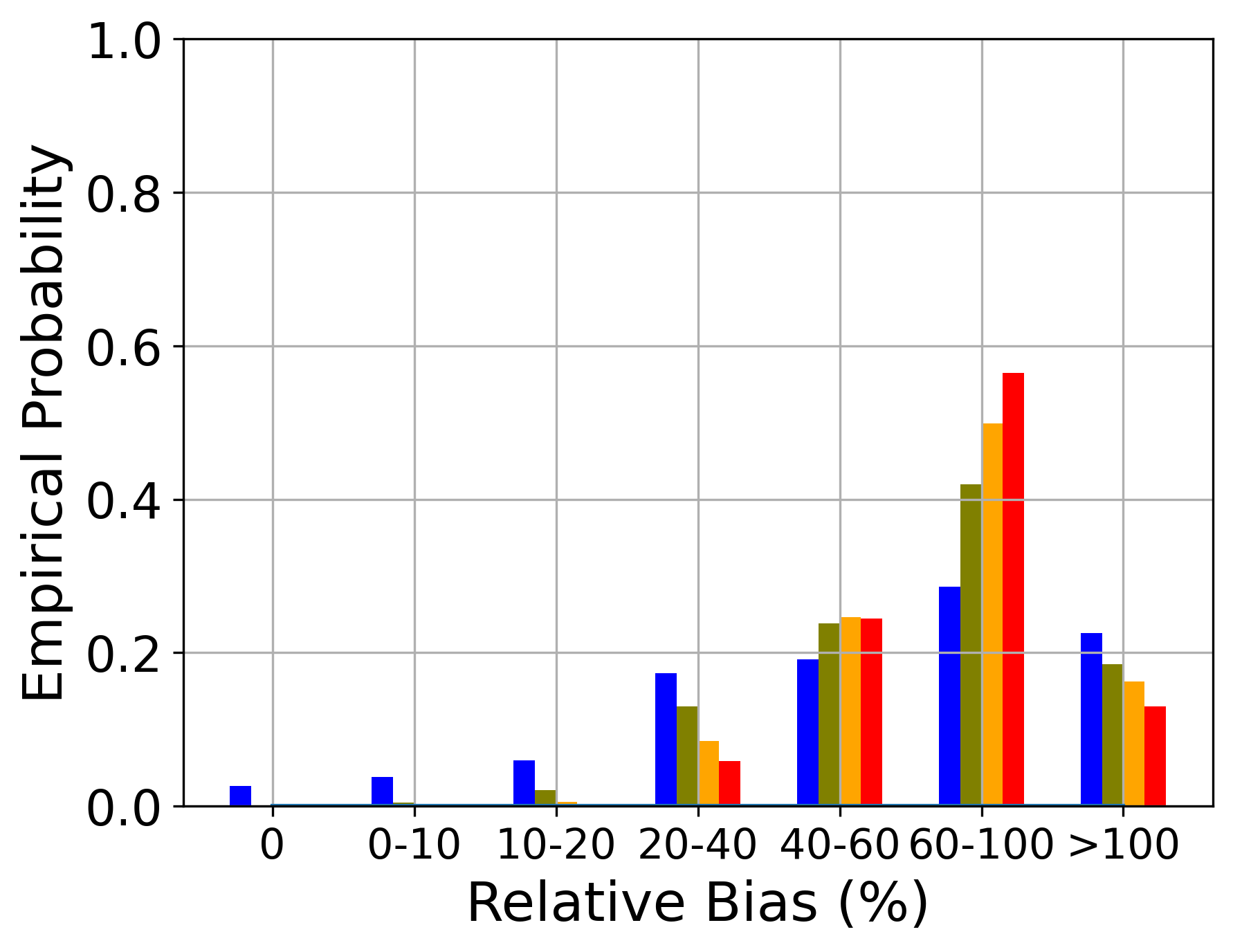}}\hfill
  \subfloat[][$N = 40$, \textsf{Std }$200\%$]{\includegraphics[width=.24\textwidth]{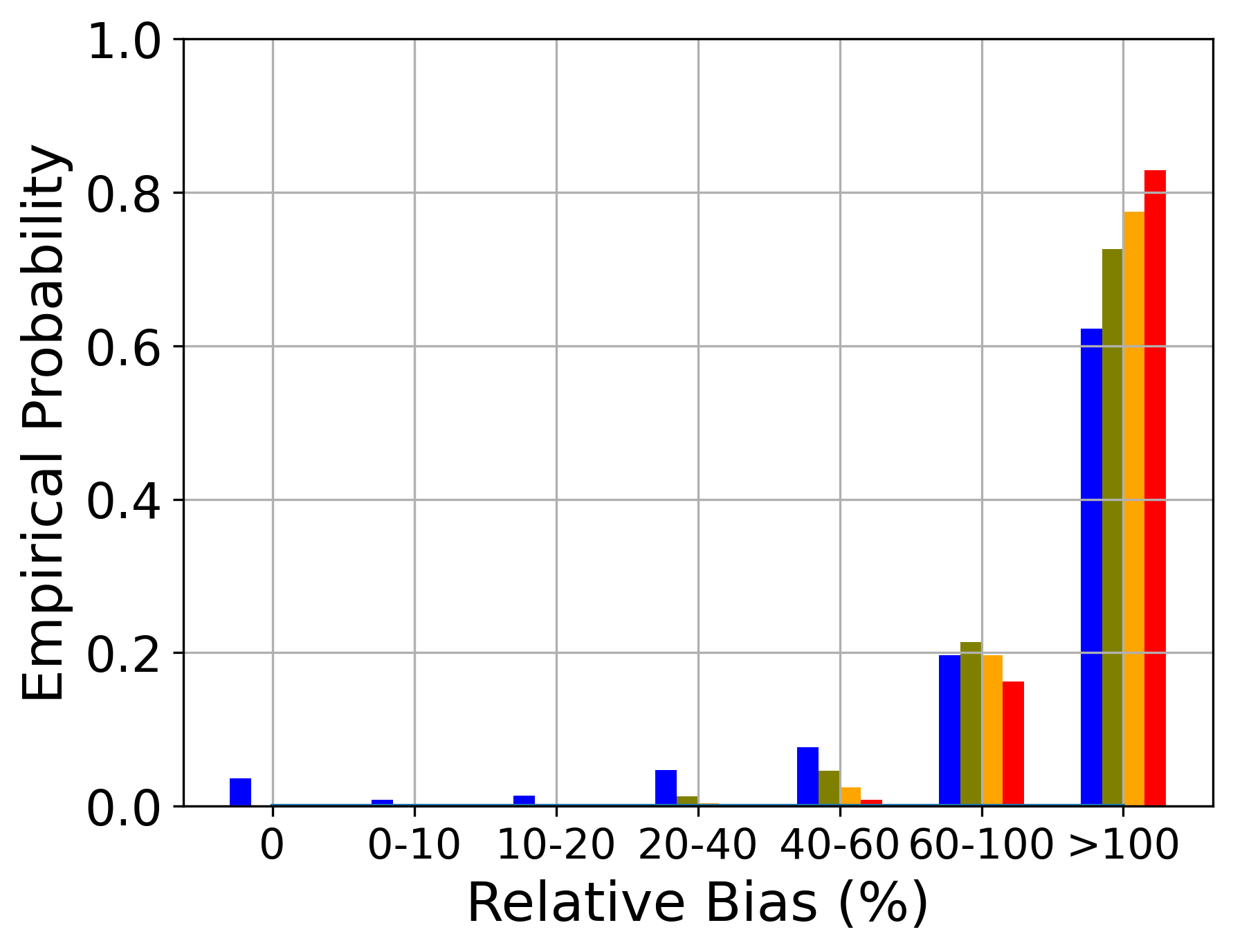}}
  \caption{2D grid graphs: Empirical probability estimates of incurring different levels of relative bias on shortest path computation across different categories of node pairs. We plot results for different graph sizes ($10 \times 10$, $20 \times 20$ and $40 \times 40$) and different levels of noise (standard deviation of noise is $20\%$, $50\%$, $100\%$, $200\%$ of mean edge weight). }
  \label{fig:2Dgrid}
\end{figure}

\smallskip\noindent{\bf Sparsity analysis.}
To further shed light on the disparity effects discussed above, 
we explore another variant of $2-$D grid graphs parametrized by a \emph{sparsity factor} (\textsf{Sp}), which is the percentage of edges with a ground-truth weight of zero\footnote{Note that even at a sparsity of $0$, there may be a significant amount of edges with a low ground truth weight, albeit not zero. For example, about $5$ percent of edges are expected to have weight $<0.05$ under a uniform distribution.}. Real-life transportation networks often have a significant proportion of edges with zero or near-zero traffic. We refer to these edges as \emph{sparse}. 
Observe that \emph{sparse} edges (those that have a $0$ ground true weight) are unique because they only contribute positive bias to the path weight, as compared to non-sparse edges. Our goal is to investigate if the sparsity factor of a graph determines how privacy affects shortest paths on it. We present results on a graph of size $N = 20$ for $4$ different sparsity factors $0~\%$, $25~\%$, $50~\%$ and $75~\%$ and at two different noise levels $20~\%$ and $50~\%$ (Figure \ref{fig:2Dgrid_sparsity}). Here, sparsity introduces two interesting effects that are in tension with each other: 
\begin{itemize}
\item \emph{1) Impact on the number of bad paths:} As the sparsity factor increases, most paths have low total weight. In turn, there are fewer bad paths whose weight is significantly larger than that of the best path, and it becomes \emph{less likely across all categories of node pairs to switch to a worse-off path}; for example, in the extreme case where the sparsity factor is $1$, all paths have weight $0$ and are equivalent. Further, \emph{longer paths are disproportionately affected and more likely to switch to a worse path than shorter paths}: this is because node pairs which are farther apart are more likely to have a short alternative due to sparsity.

\item \emph{2) Impact on path weight estimation bias:}  Equation~\eqref{eq:noisy_weight} highlights that the noisy weights, if negative, are rounded up to $0$. In particular, this introduces positive bias on edge weights. 

However, this bias affects edges disproportionately. In particular, edges whose weights are closer to $0$ experience more positive bias (as these edges have a high probability of needing to be rounded up after noise addition), while edges with higher weights experience bias closer to $0$ on average (after adding noise, these weights are almost never negative and do not need rounding up). This means that paths with fewer edges are disproportionately more likely to be overestimated compared to paths with more edges. 

This effect makes it i) \emph{more likely across all categories of node pairs to switch to a worse-off path} and ii) implies that \emph{shorter pairs are more likely to be affected} since they tend to be overestimated. Further, these bias effects increase as more noise is added, as it becomes more and more likely that enough noise is added that rounding to $0$ becomes necessary and the bias becomes positive.
\end{itemize}

Figure \ref{fig:2Dgrid_sparsity} shows the tension between these two effects. With a \textsf{Sp} factor of $0.75$, the first effect seems to take over and reduce the likelihood of bias; this is expected, as at this point, most paths are very short across all categories of node pairs and there are few opportunities to change to a significantly worse path. Subfigures (d) and (e) particularly highlight the second effect, with \Qone{} node pairs having a more extreme distribution of relative bias. The second effect is also visible in how outcomes for all categories of node pairs are worse with a noise level of $50~\%$ as opposed to $20~\%$ for all levels of sparsity. 

The interaction between the two effects can be complex. We note that for a noise level of $20~\%$, the first effect seems to dominate, leading to less overall relative bias, and this bias seems to affect \Qfour{} node pairs more than \Qone{} pairs. However, as the noise level increases to $50~\%$, the second effect starts becoming important, leading to potentially complex trends. At very high levels of sparsity (Sp $0.75$), the first effect seems to take over with \Qone{} node pairs becoming extremely robust to privacy noise and \Qfour{} pairs being more affected.

\begin{figure}[!ht]
  \centering
  \raisebox{35pt}{\parbox[b]{.05\textwidth}{}}%
  \subfloat[][\textsf{Std }$20\%$, \textsf{Sp }$0.25$]{\includegraphics[width=.32\textwidth]{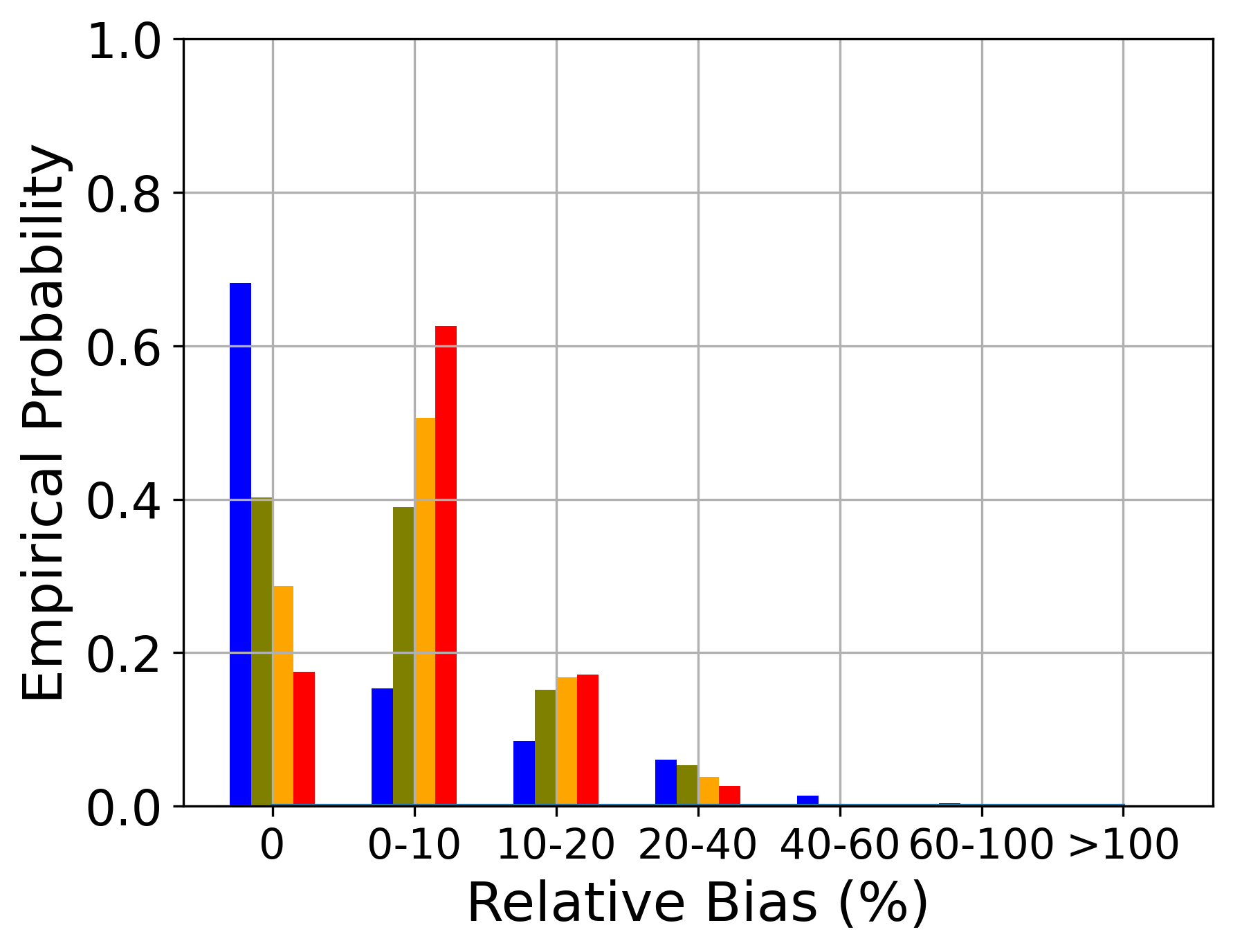}}\hfill
  \subfloat[][\textsf{Std }$20\%$, \textsf{Sp }$0.50$]{\includegraphics[width=.32\textwidth]{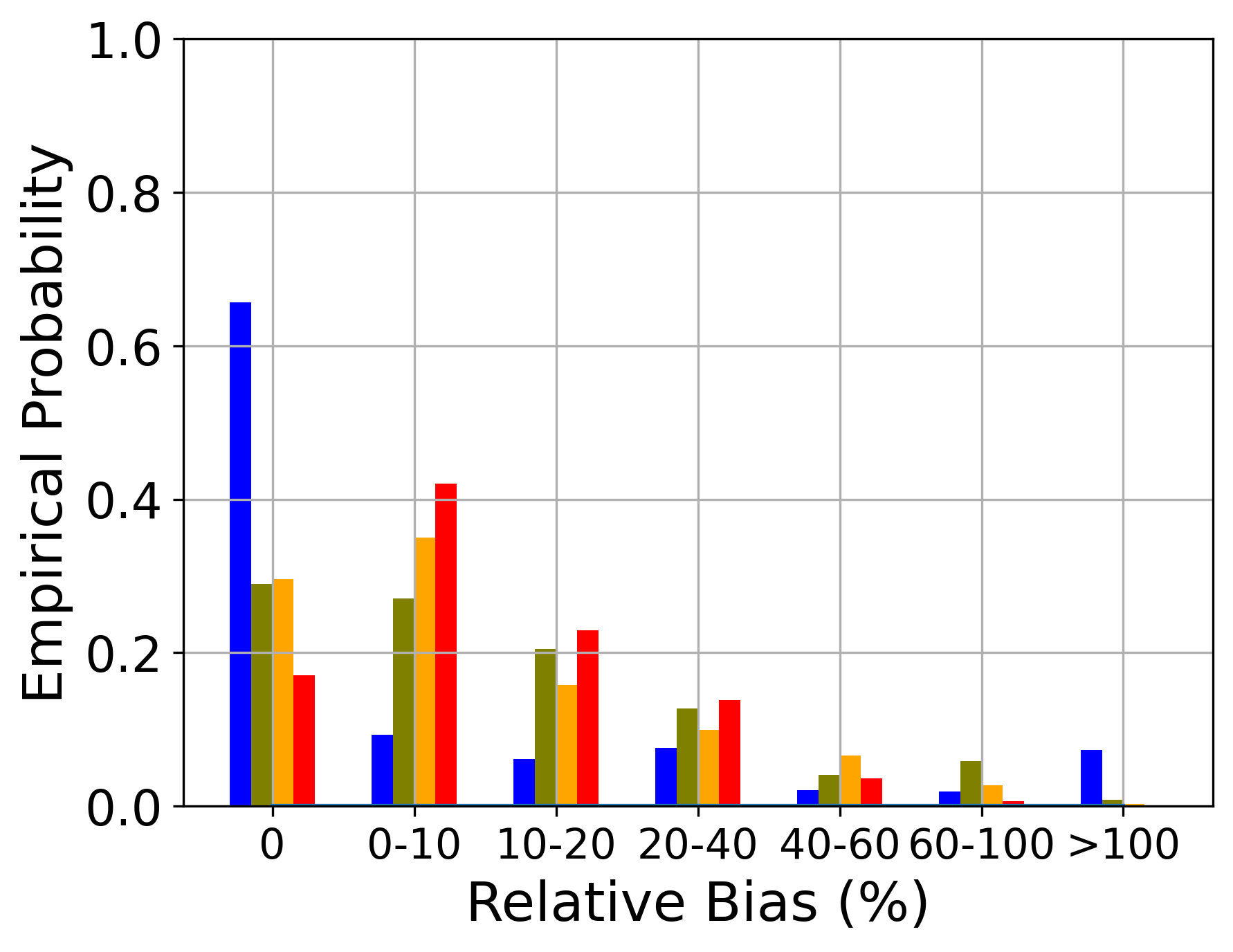}}\hfill
  \subfloat[][\textsf{Std }$20\%$, \textsf{Sp }$0.75$]{\includegraphics[width=.32\textwidth]{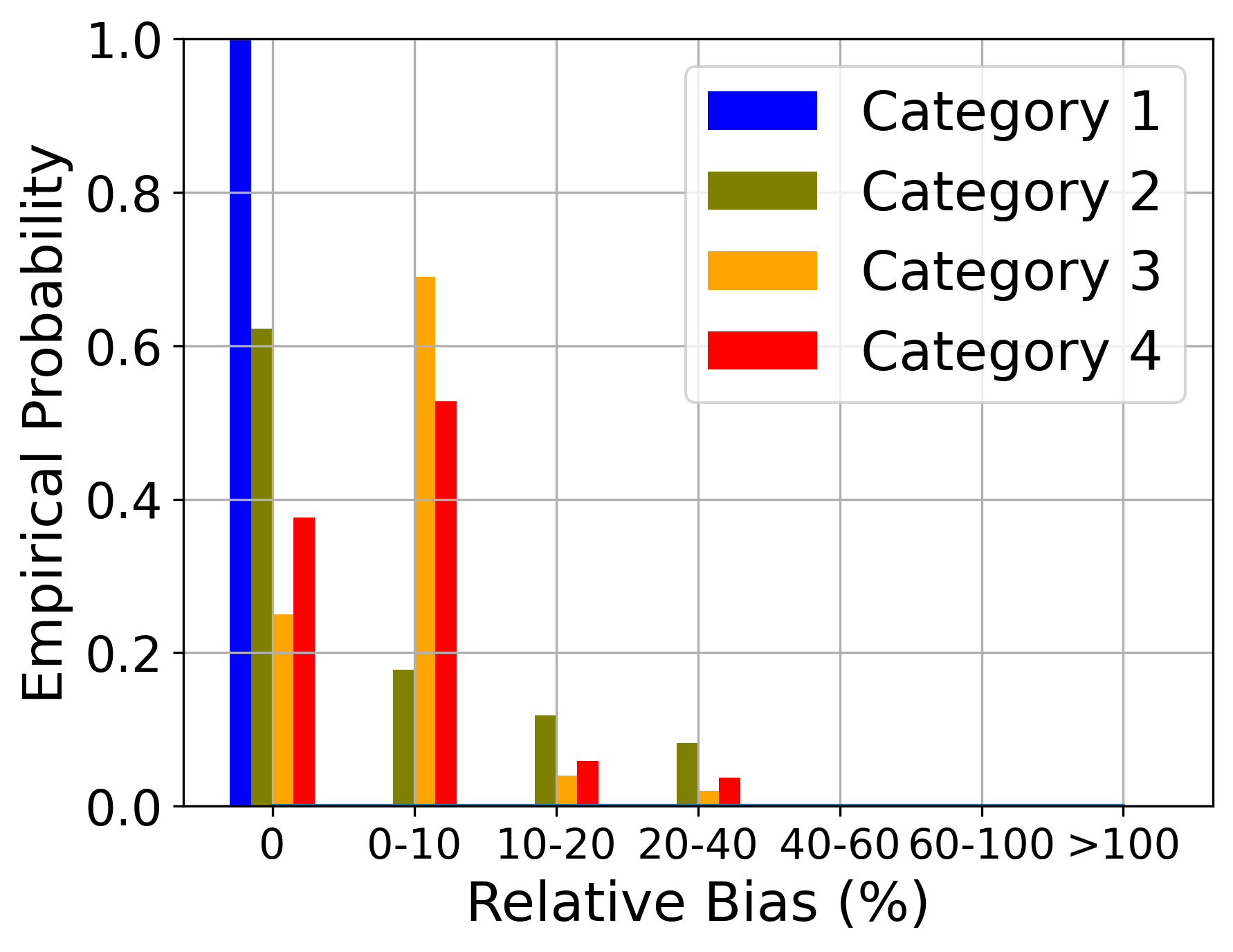}}\par
  \raisebox{35pt}{\parbox[b]{.05\textwidth}{}}
  \subfloat[][\textsf{Std }$50\%$, \textsf{Sp }$0.25$]{\includegraphics[width=.32\textwidth]{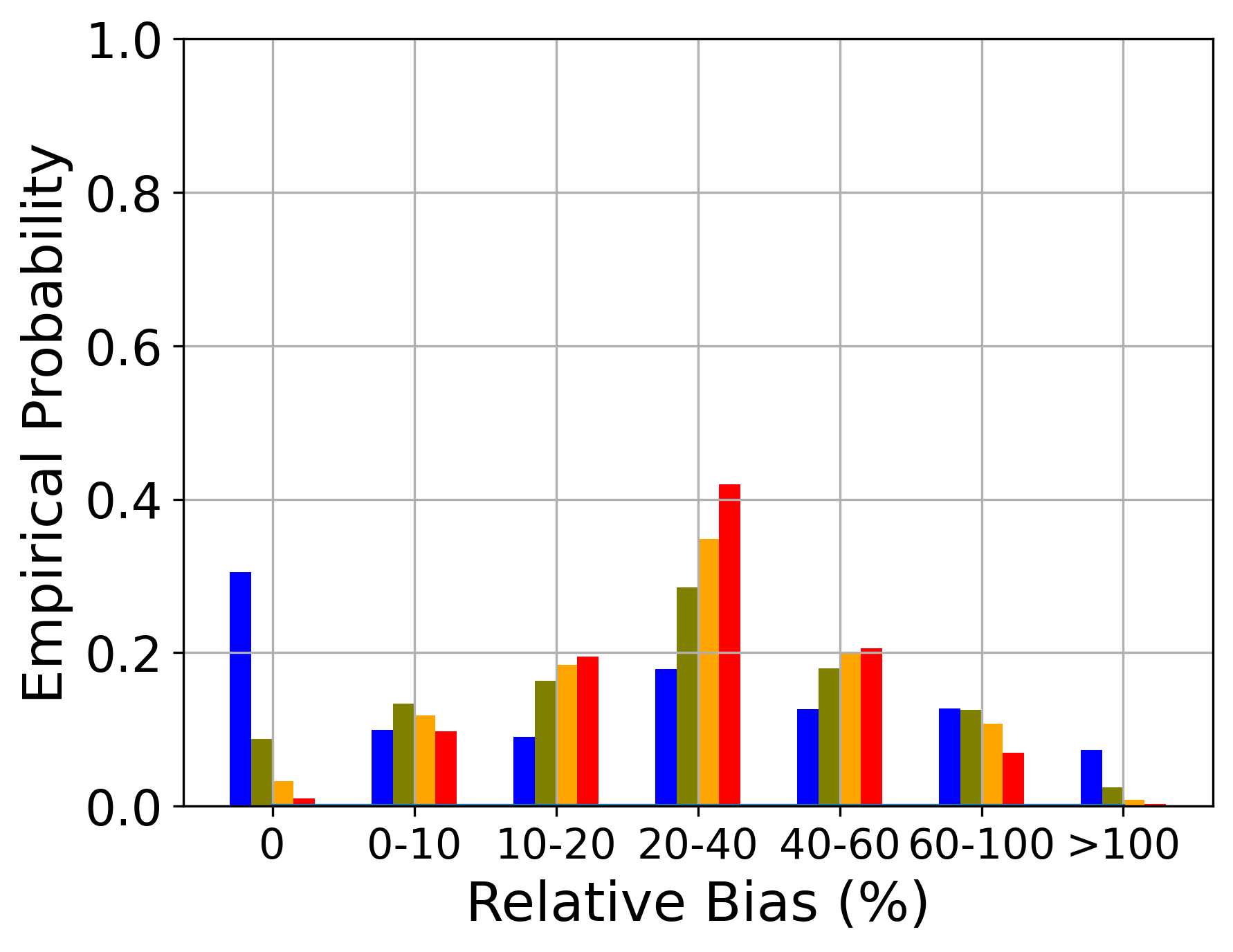}}\hfill
  \subfloat[][\textsf{Std }$50\%$, \textsf{Sp }$0.5$]{\includegraphics[width=.32\textwidth]{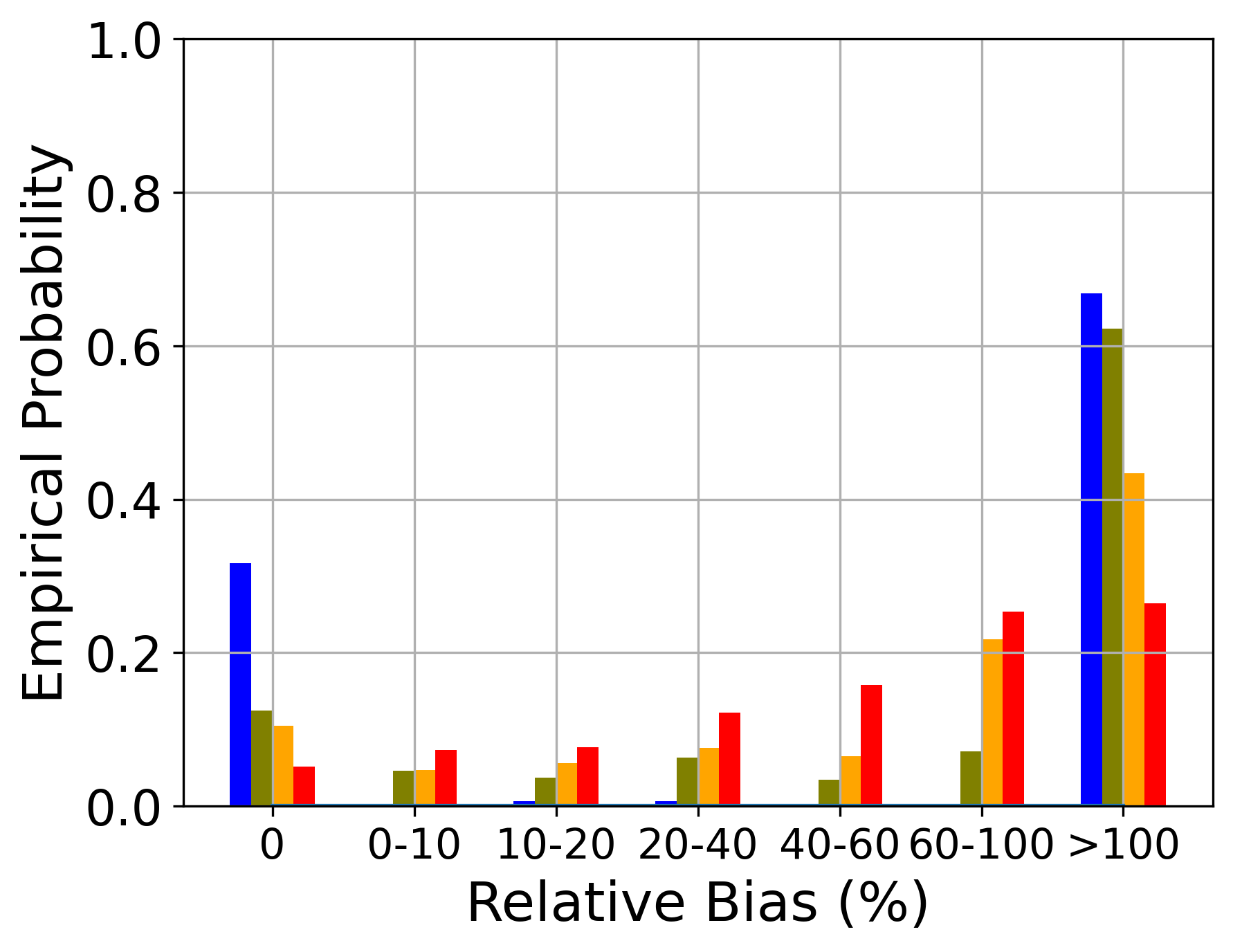}}\hfill
  \subfloat[][\textsf{Std }$50\%$, \textsf{Sp }$0.75$]{\includegraphics[width=.32\textwidth]{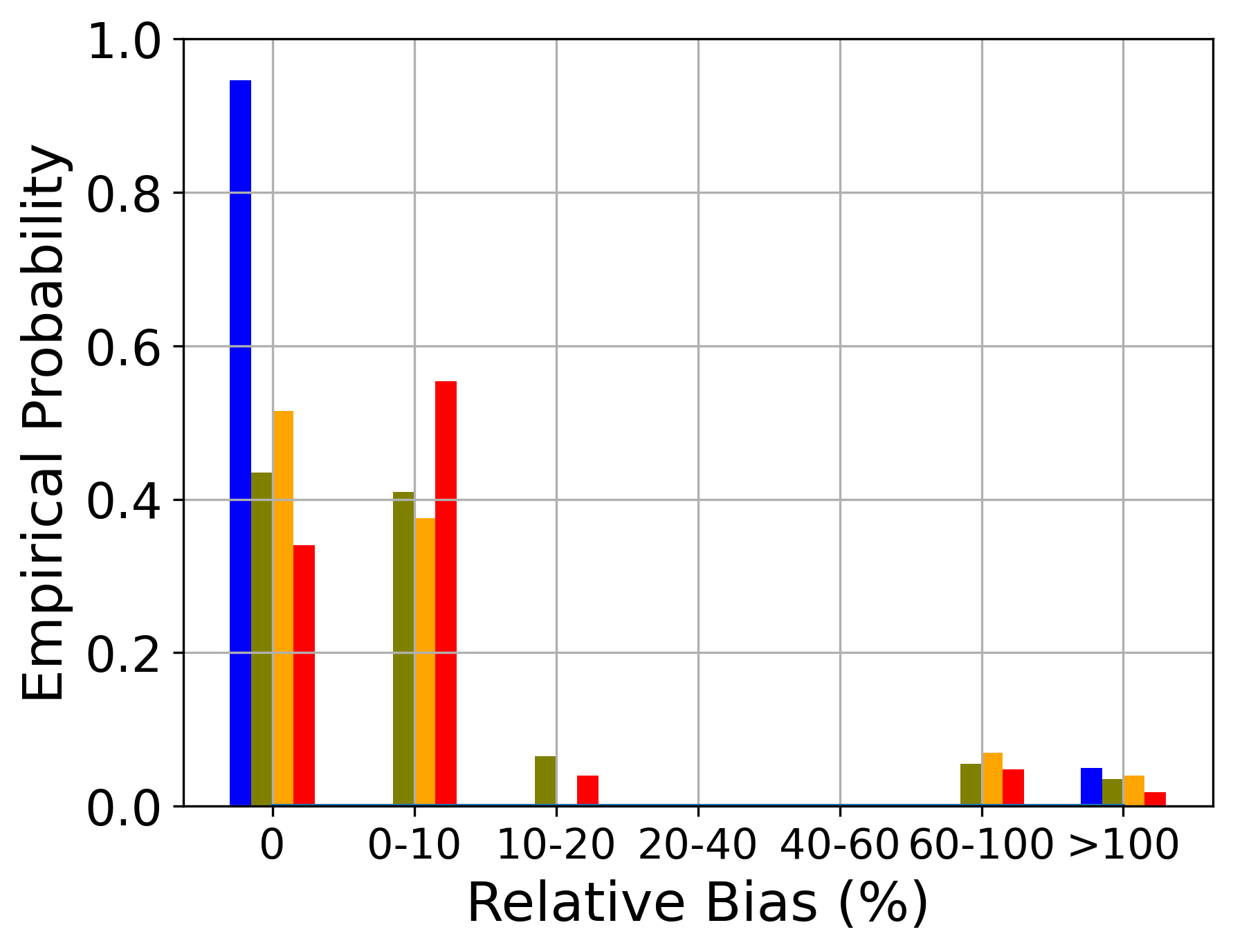}}
  \caption{2D grids graphs: Effects on privacy noise on path change statistics when graphs are sparse in a specific way: many edges on the graph have close to zero traffic and hence have $0$ ground truth weight. In each row (from left to right), we plot results for different levels of sparsity ($0\%$, $25\%$, $50\%$, $75\%$) at a fixed noise variance. In each column, we see the effect of varying noise variance at a fixed sparsity ratio.}
  \label{fig:2Dgrid_sparsity}
\end{figure}

\subsubsection{Wheel graphs.} 

Next, we examine wheel graphs. These graphs have two types of edges: {\bf i)} circumference edges and {\bf ii)} spoke edges. 
All circumference edges have their ground truth weights drawn independently from a \emph{Uniform}$[0, 1]$ distribution. Since spoke edges are expected to accommodate larger flows, their ground truth weights are drawn independently from \emph{Uniform}$[0, r]$ where $r \geq 1$. Thus, $r$ represents the ratio of mean edge weights for the two groups of edges. For numerical experiments, our parameters of interest are the following: i) size of the graph $N$ and ii) ratio $r$. However, wheel graphs have circular symmetry which means that $N$ does not affect the outcomes independently. So, we fix $N = 101$ for all experiments and only vary $r$ from the following set: $\{1, 20, 50, 100\}$. Additionally, like all previous experiments, we also consider different levels of privacy noise: $20~\%$, $50~\%$ and $100~\%$. Refer to Figure \ref{fig:wheel} for a graphical representation of all results, based on which we make the following observations:  

Similar to the observations for 2-D grid graphs, as the levels of noise increase, node pairs of all categories are more likely to be affected. Once again, \Qone{} pairs are significantly more robust against privacy noise compared to \Qfour{} pairs, for the same reasons as highlighted earlier.
    
The most striking observation is that the ratio $r$ greatly influences the degree to which bias is realized. As $r$ increases, all node pairs become more and more robust to privacy noise. This is a direct consequence of the topology of a wheel graph. Note that there are only two kinds of source-destination pairs: 
{\bf i)} between a central node and an outer node, and 
{\bf ii)} between two outer nodes. In both cases, with high $r$, there is only one candidate path that is the most viable shortest path. For case {\bf i)}, it involves identifying the spoke edge with the least weight, traversing it to reach the corresponding outer node, and then traveling along the circumference to reach the destination. For case {\bf ii)}, the only feasible least-cost path is to travel along the low-weight paths on the circumference (any trip to the center involves traversing a high-weight spoke edge and is sub-optimal). This result follows from Theorem \ref{thm:q_beta}: in this case, the large gap $\beta$ between the best path and all other paths drives the probability $q_{\beta}$ to very low levels, leading to a high degree of robustness.


\begin{figure}[!ht]
  \centering
  \raisebox{35pt}{\parbox[b]{.05\textwidth}{}}%
  \subfloat[][\textsf{Std }$20\%$, $r=1$]{\includegraphics[width=.30\textwidth]{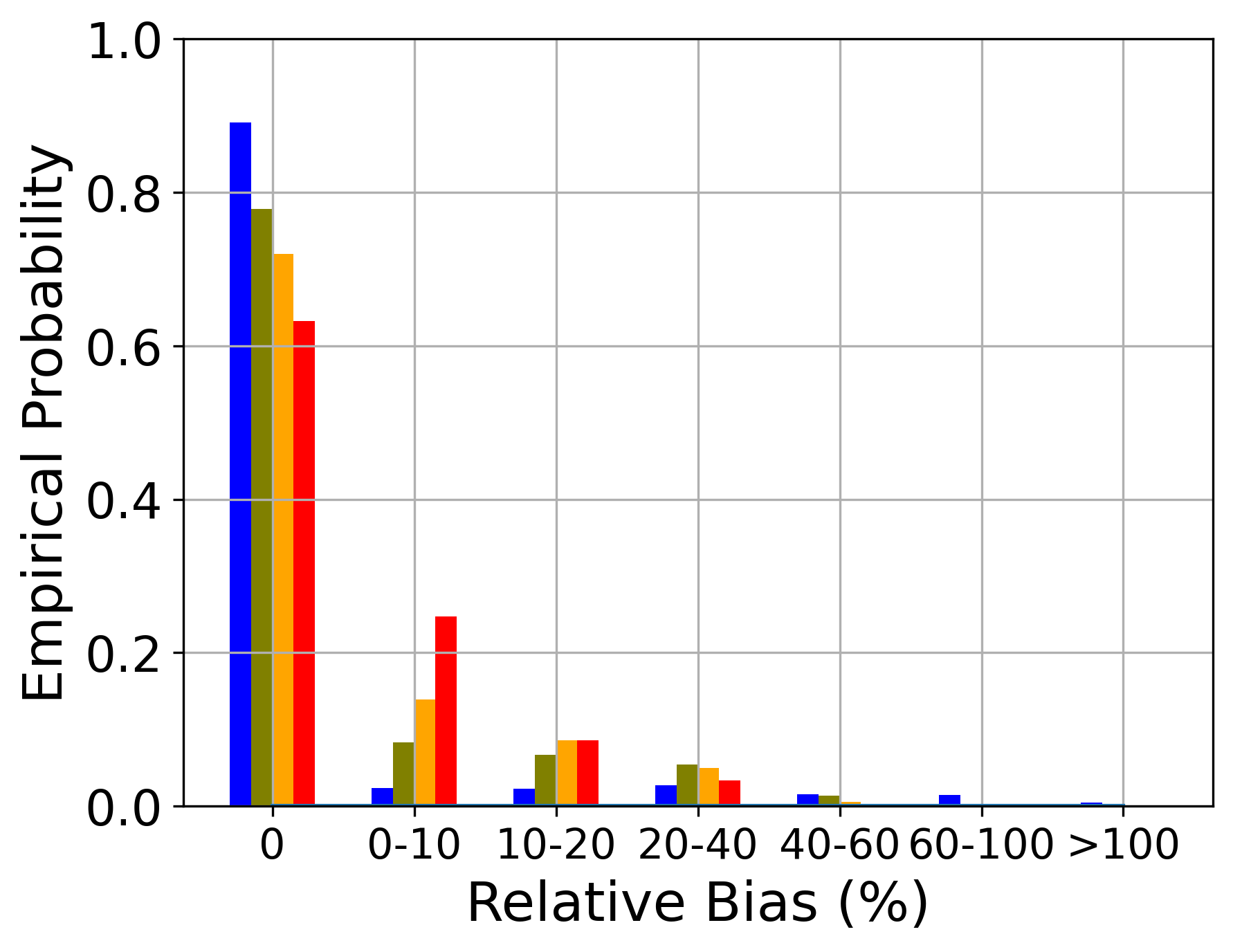}}\hfill
  \subfloat[][\textsf{Std }$50\%$, $r=1$]{\includegraphics[width=.30\textwidth]{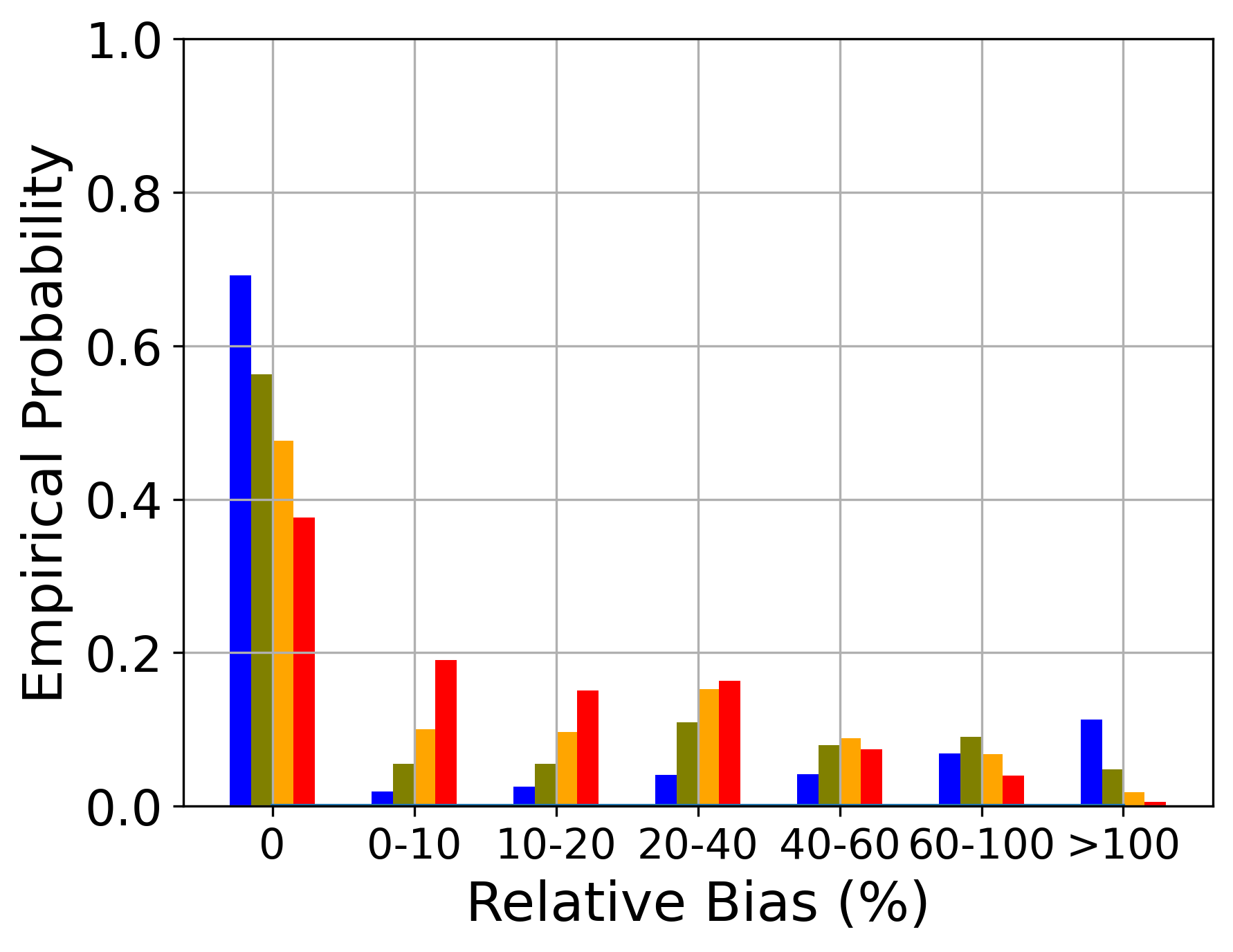}}\hfill
  \subfloat[][\textsf{Std }$100\%$, $r=1$]{\includegraphics[width=.30\textwidth]{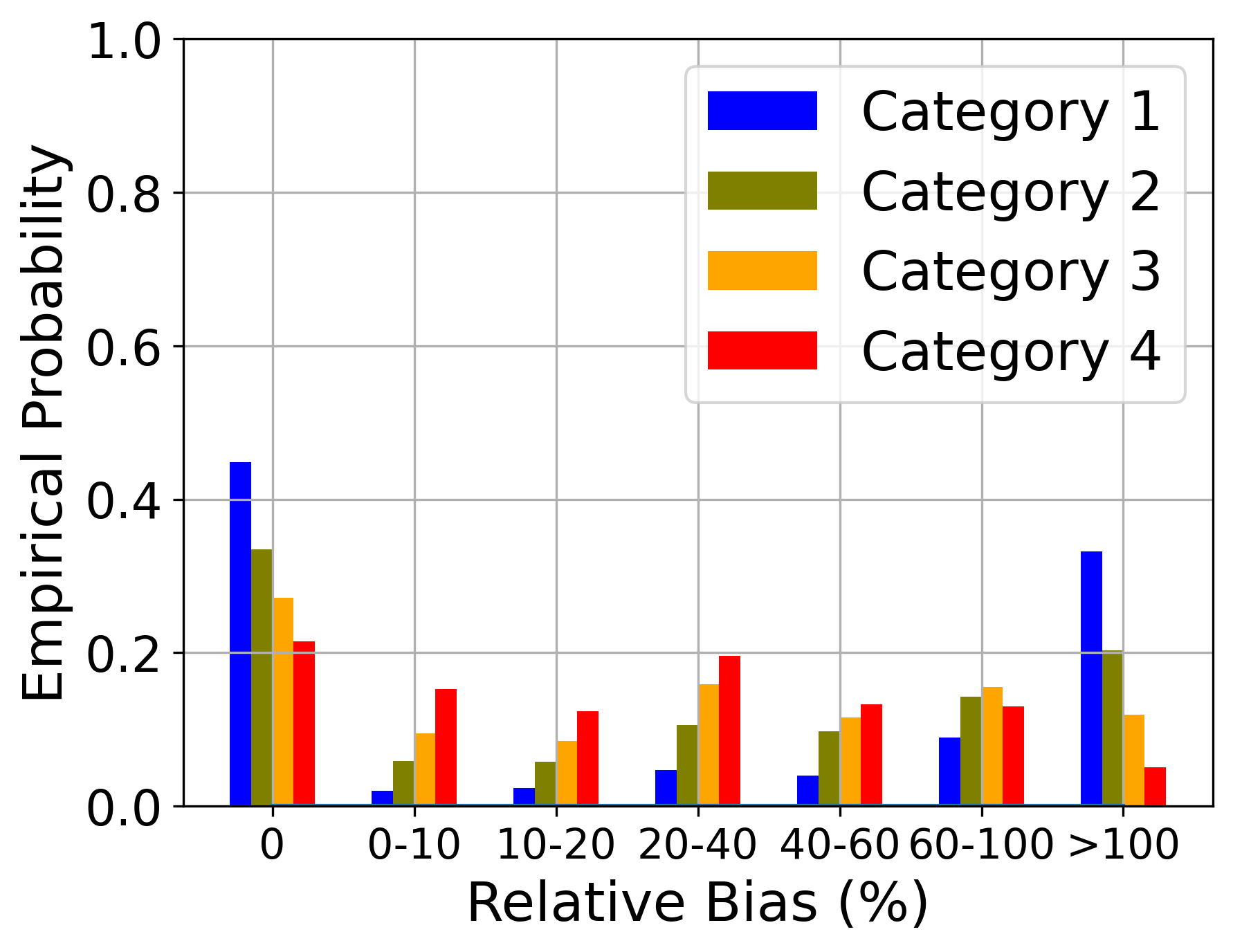}}\par
  \raisebox{35pt}{\parbox[b]{.05\textwidth}{}}
  \subfloat[][\textsf{Std }$20\%$, $r=20$]{\includegraphics[width=.30\textwidth]{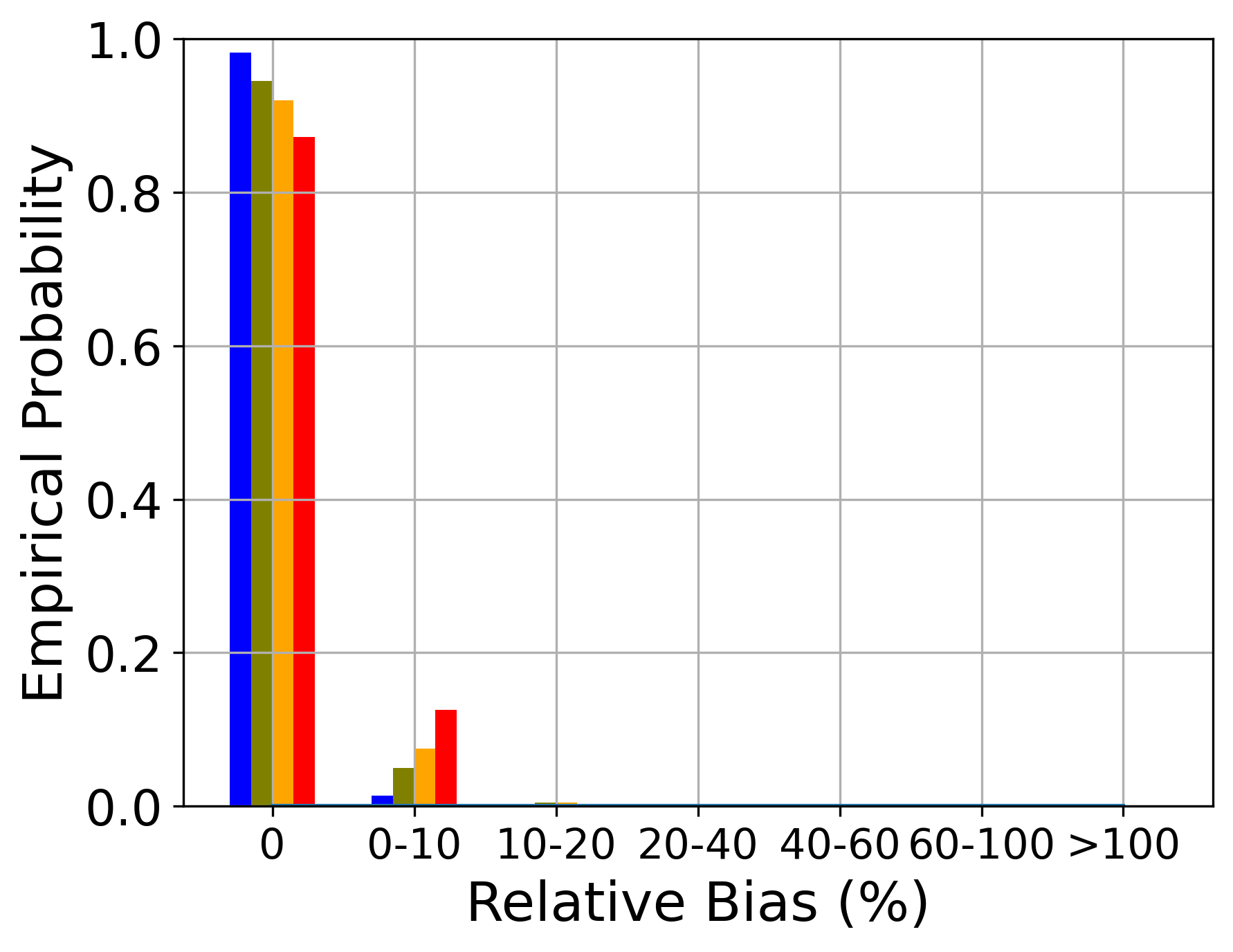}}\hfill
  \subfloat[][\textsf{Std }$50\%$, $r=20$]{\includegraphics[width=.30\textwidth]{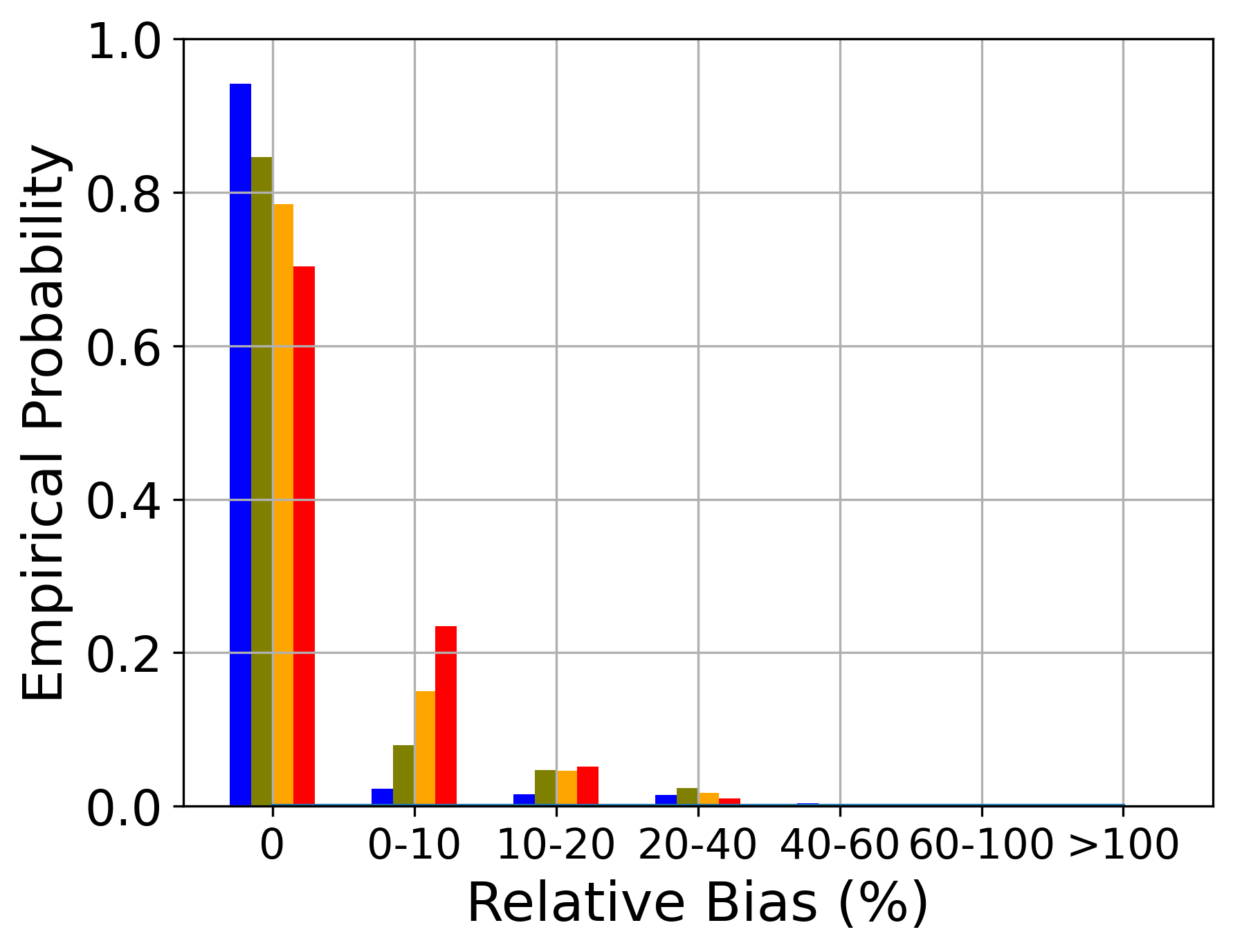}}\hfill
  \subfloat[][\textsf{Std }$100\%$, $r=20$]{\includegraphics[width=.30\textwidth]{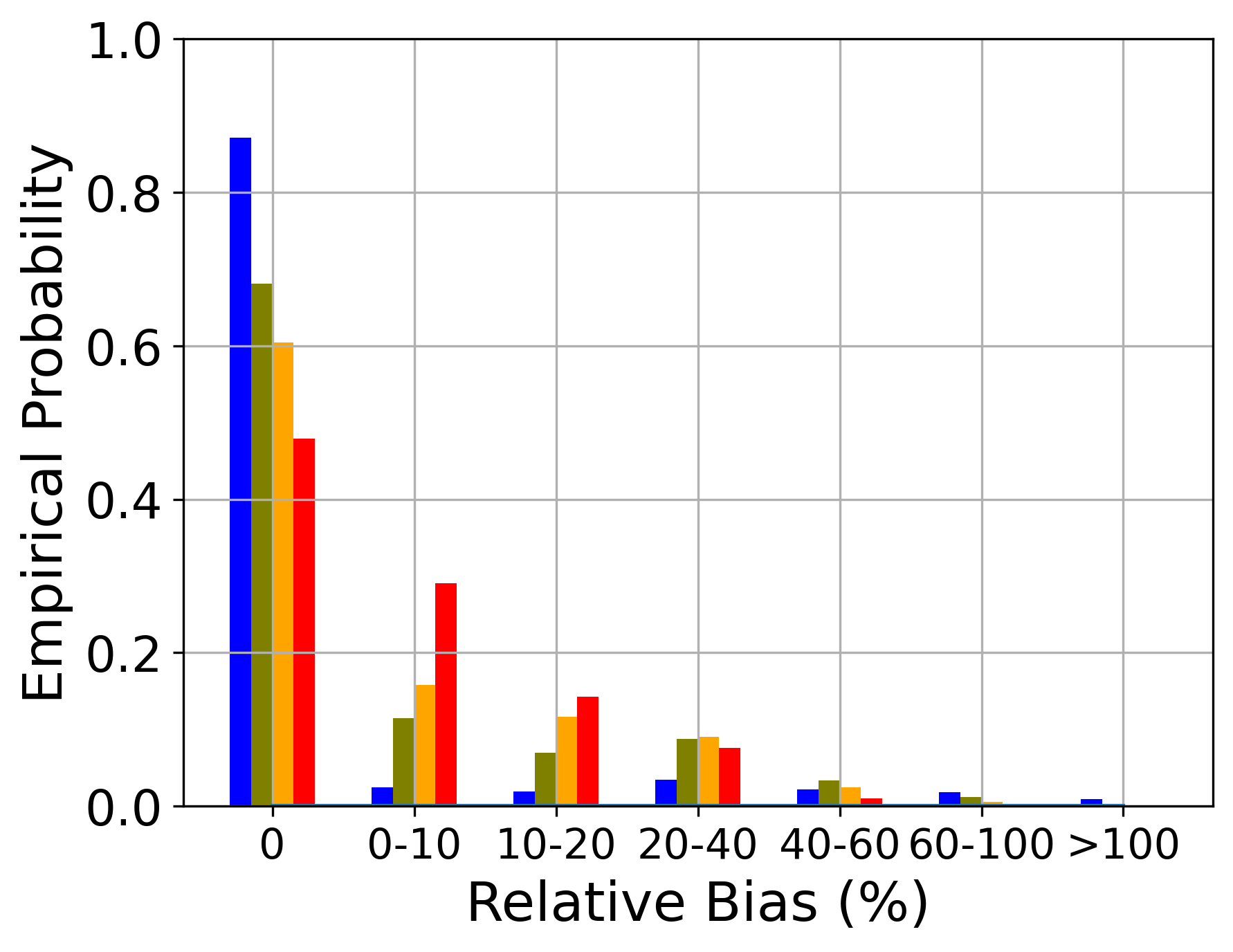}}\par
  \raisebox{35pt}{\parbox[b]{.05\textwidth}{}}
  \subfloat[][\textsf{Std }$20\%$, $r=50$]{\includegraphics[width=.30\textwidth]{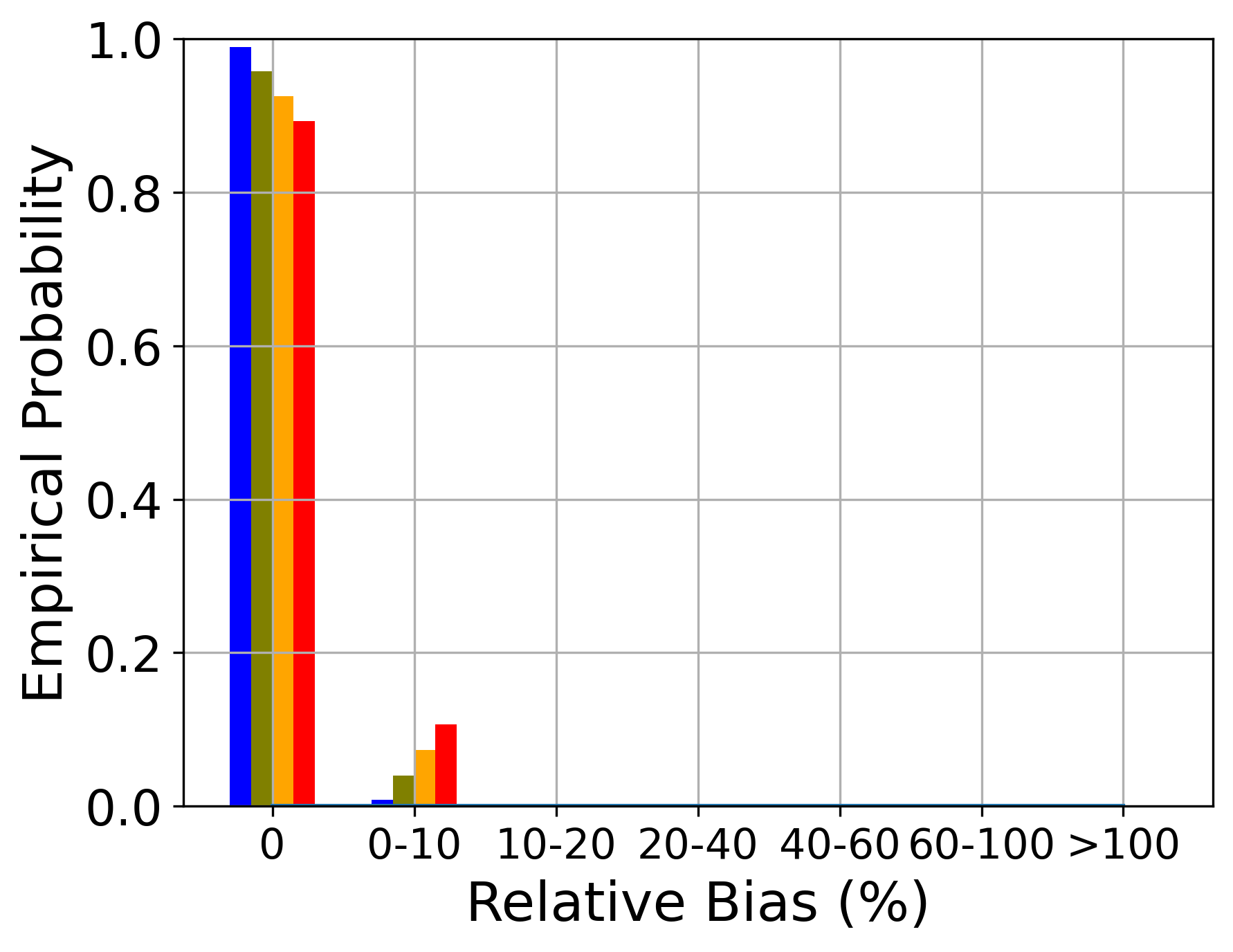}}\hfill
  \subfloat[][\textsf{Std }$50\%$, $r=50$]{\includegraphics[width=.30\textwidth]{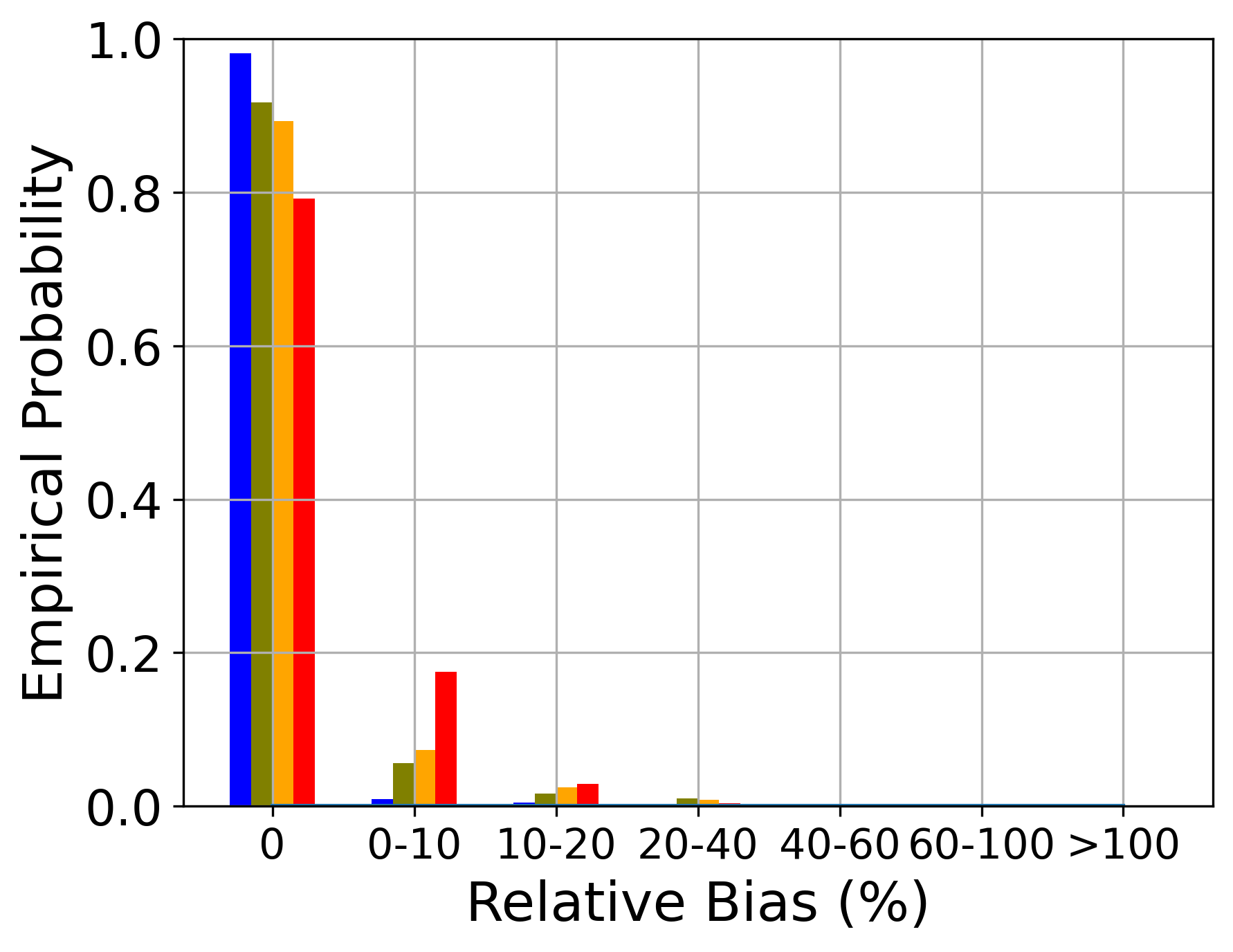}}\hfill
  \subfloat[][\textsf{Std }$100\%$, $r=50$]{\includegraphics[width=.30\textwidth]{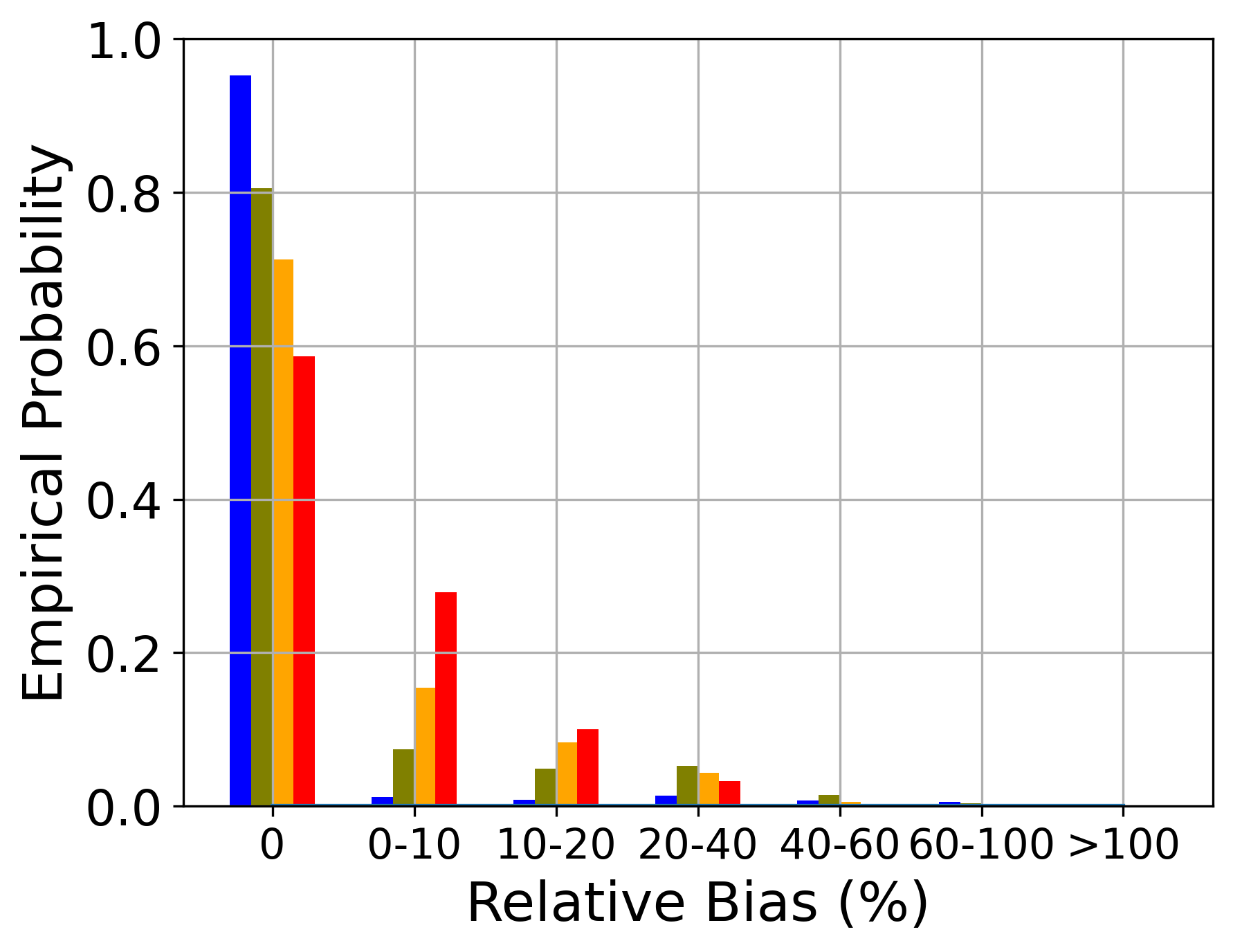}}\par
  \raisebox{35pt}{\parbox[b]{.05\textwidth}{}}
  \subfloat[][\textsf{Std }$20\%$, $r=100$]{\includegraphics[width=.30\textwidth]{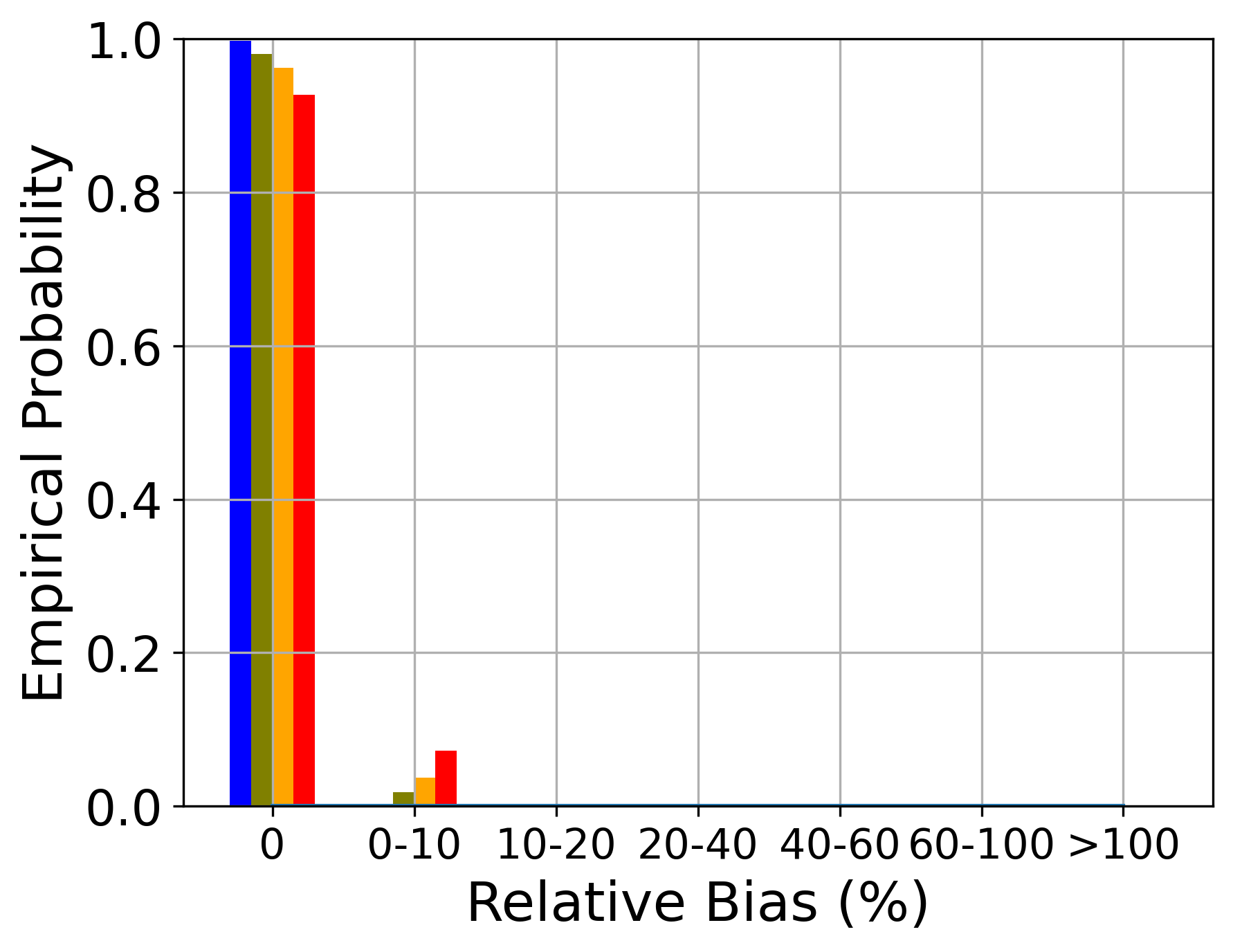}}\hfill
  \subfloat[][\textsf{Std }$50\%$, $r=100$]{\includegraphics[width=.30\textwidth]{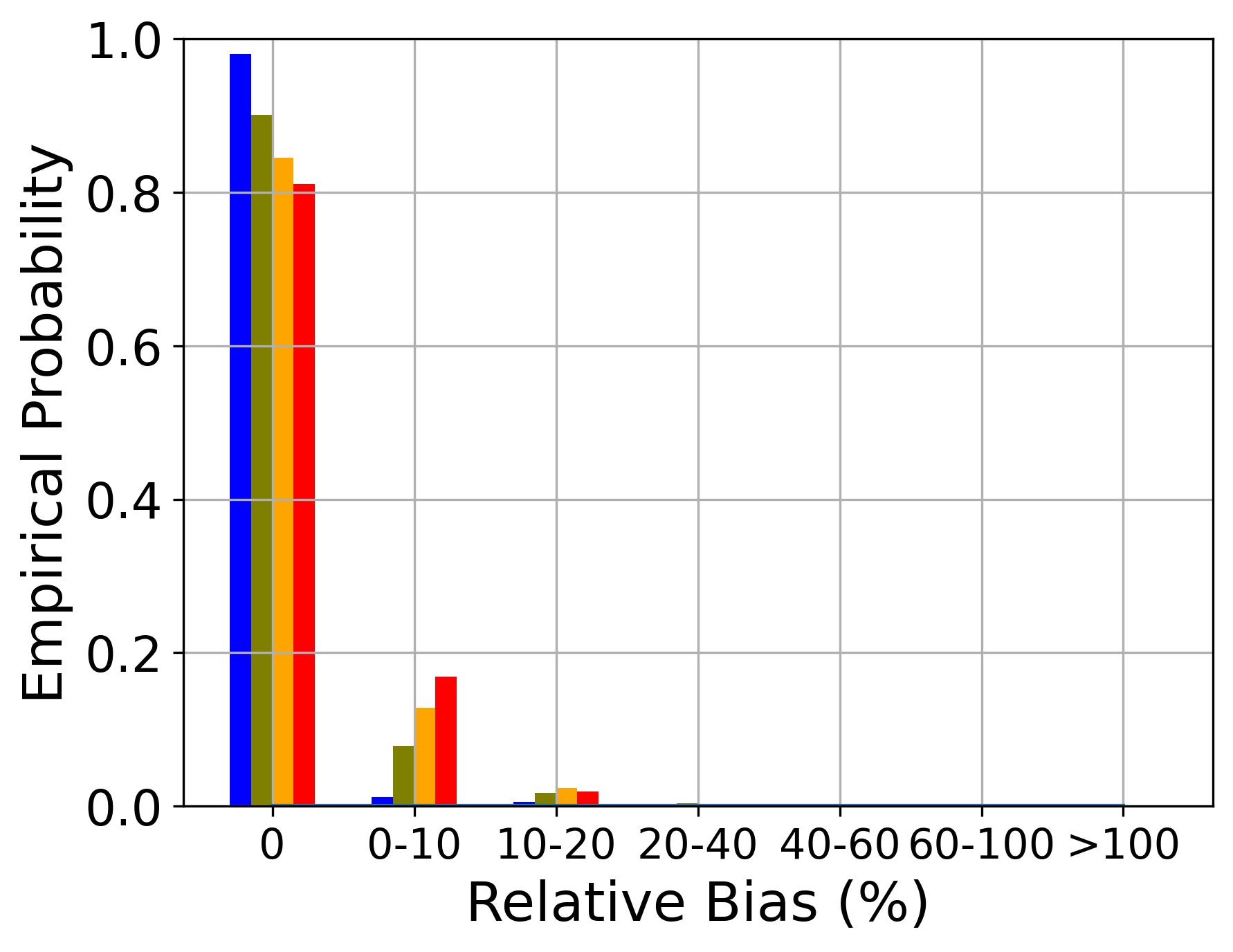}}\hfill
  \subfloat[][\textsf{Std }$100\%$, $r=100$]{\includegraphics[width=.30\textwidth]{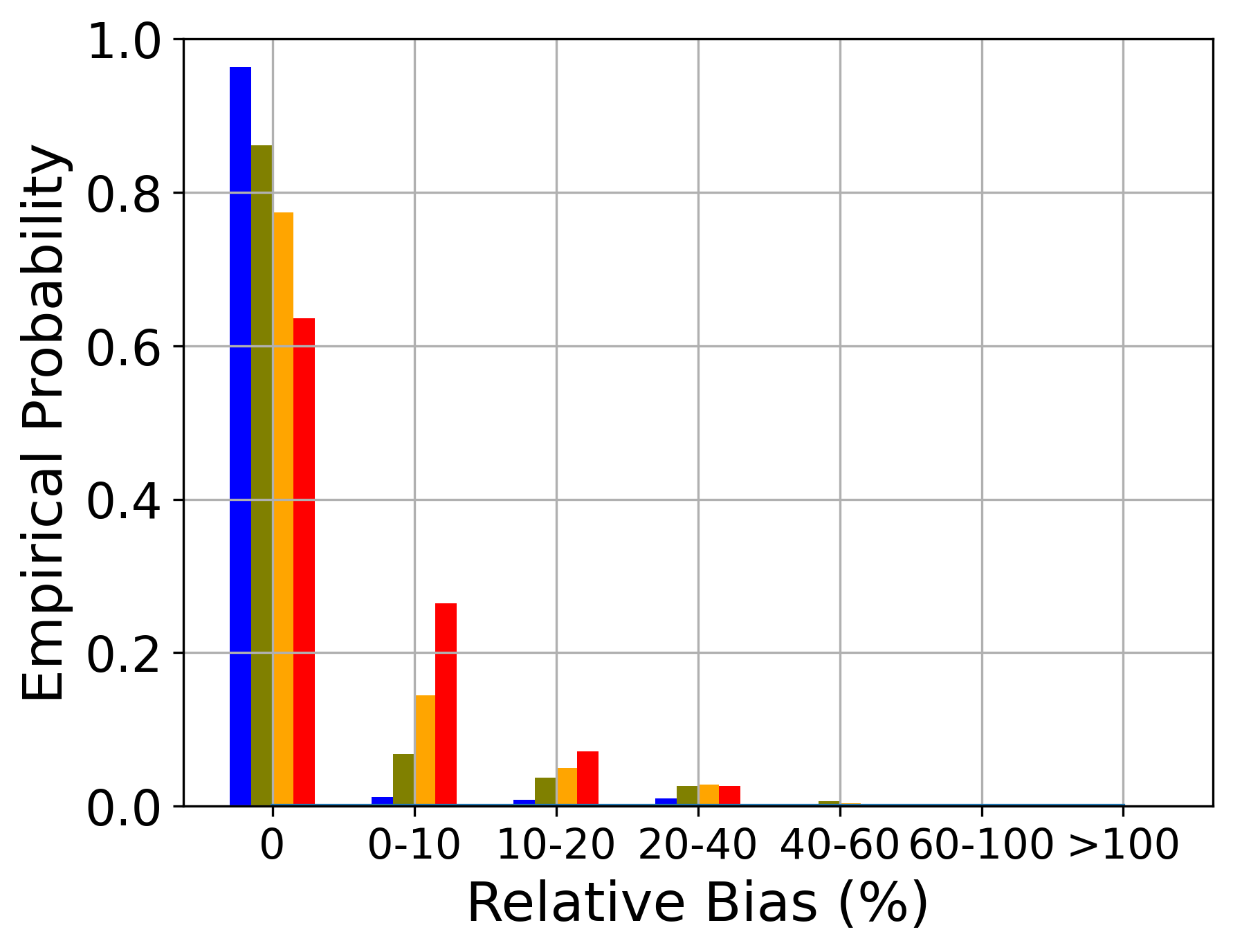}}
  \caption{Statistics for wheel graphs with $N = 101$ nodes. In each row (from left to right), we generate results for 3 different levels of noise: i) $20\%$; ii) $50\%$; and iii) $100\%$. On the other hand, in each column (from top to bottom), we plot results for different values of $r$: i) $r=1$; ii) $r = 20$; iii) $r=50$; and iv) $r=100$.}
  \label{fig:wheel}
\end{figure}

\subsubsection{Scale-free graphs.}

We conclude our experimental section with a study of scale-free graphs. The primary parameter of interest for scale-free graphs is the power $\gamma$ of their underlying degree distribution. Note that scale-free graphs can often have multiple disconnected components (including many singleton nodes of degree zero). However, for our simulation, we always pick its largest connected component. All ground truth edge weights are drawn independently from \emph{Uniform}$[0, 1]$. In Figure \ref{fig:scalefree100}, we present results for graphs generated using a starting value of $N = 100$ at different values of $\gamma \in \{1.5, 2, 2.5, 3\}$ and different levels of privacy noise. The main observations are as follows:

\begin{figure}[!h]
  \centering
  \raisebox{35pt}{\parbox[b]{.05\textwidth}{}}%
  \subfloat[][\textsf{Std }$20\%$, $\gamma=1.5$]{\includegraphics[width=.23\textwidth]{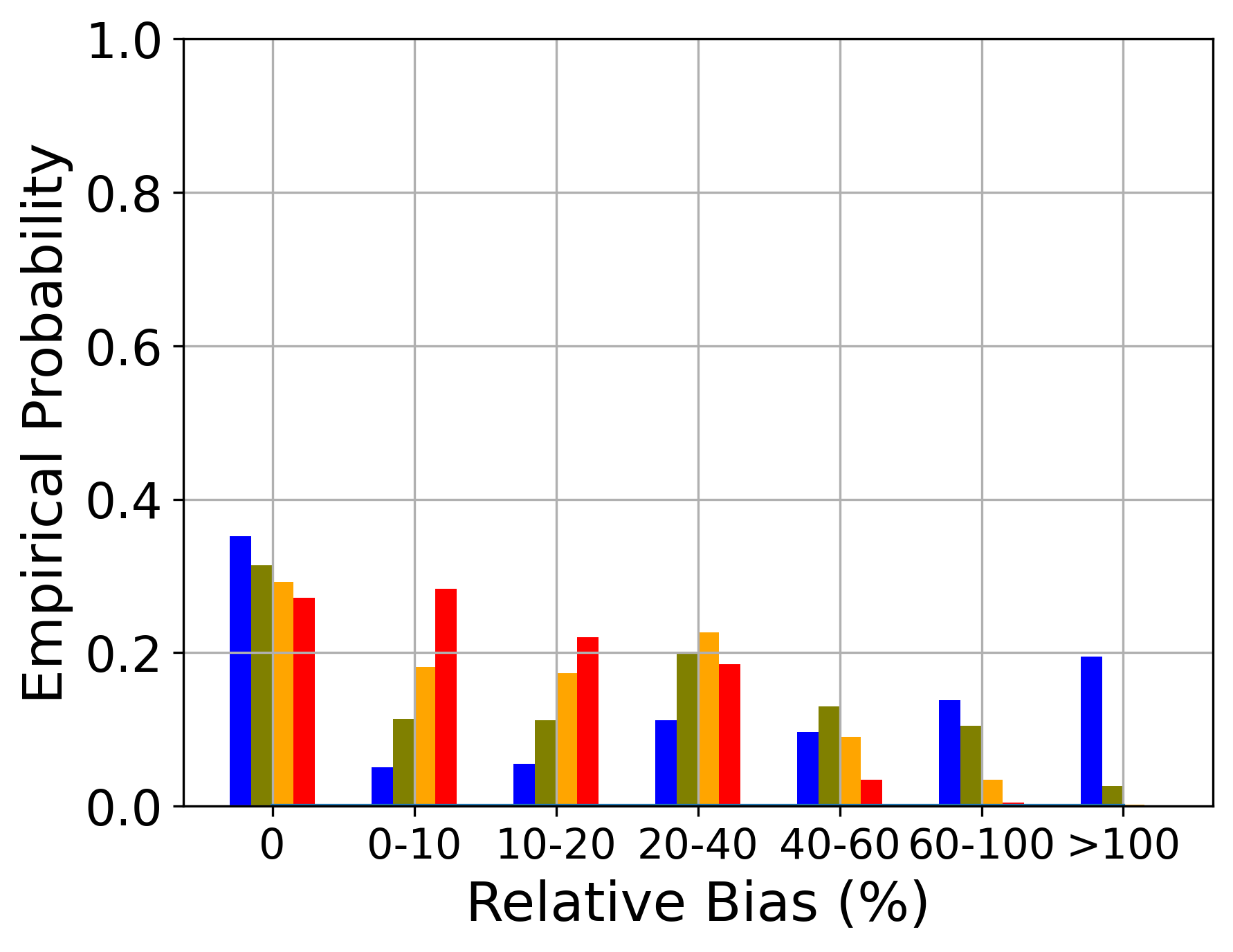}}\hfill
  \subfloat[][\textsf{Std }$50\%$, $\gamma=1.5$]{\includegraphics[width=.23\textwidth]{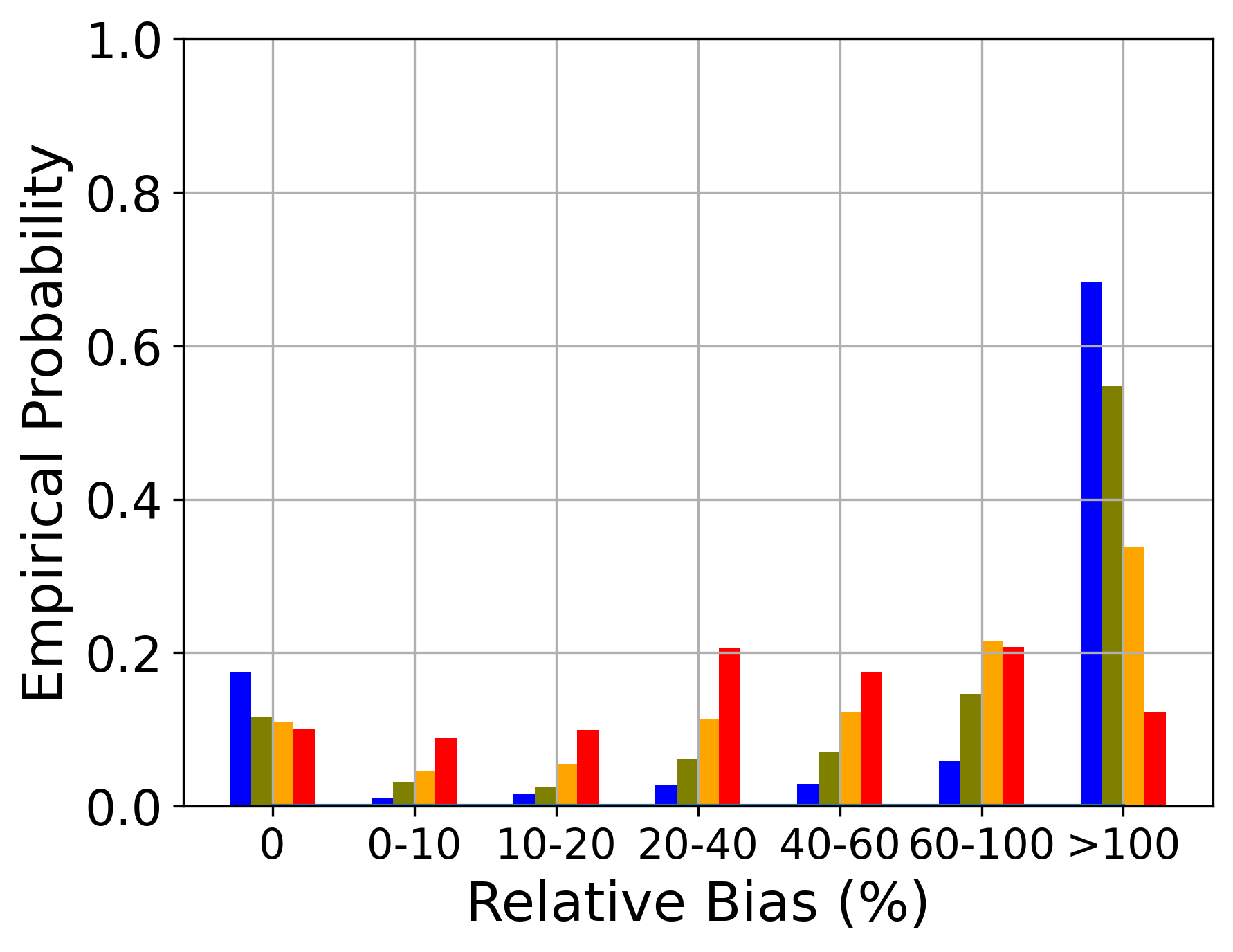}}\hfill
  \subfloat[][\textsf{Std }$100\%$, $\gamma=1.5$]{\includegraphics[width=.23\textwidth]{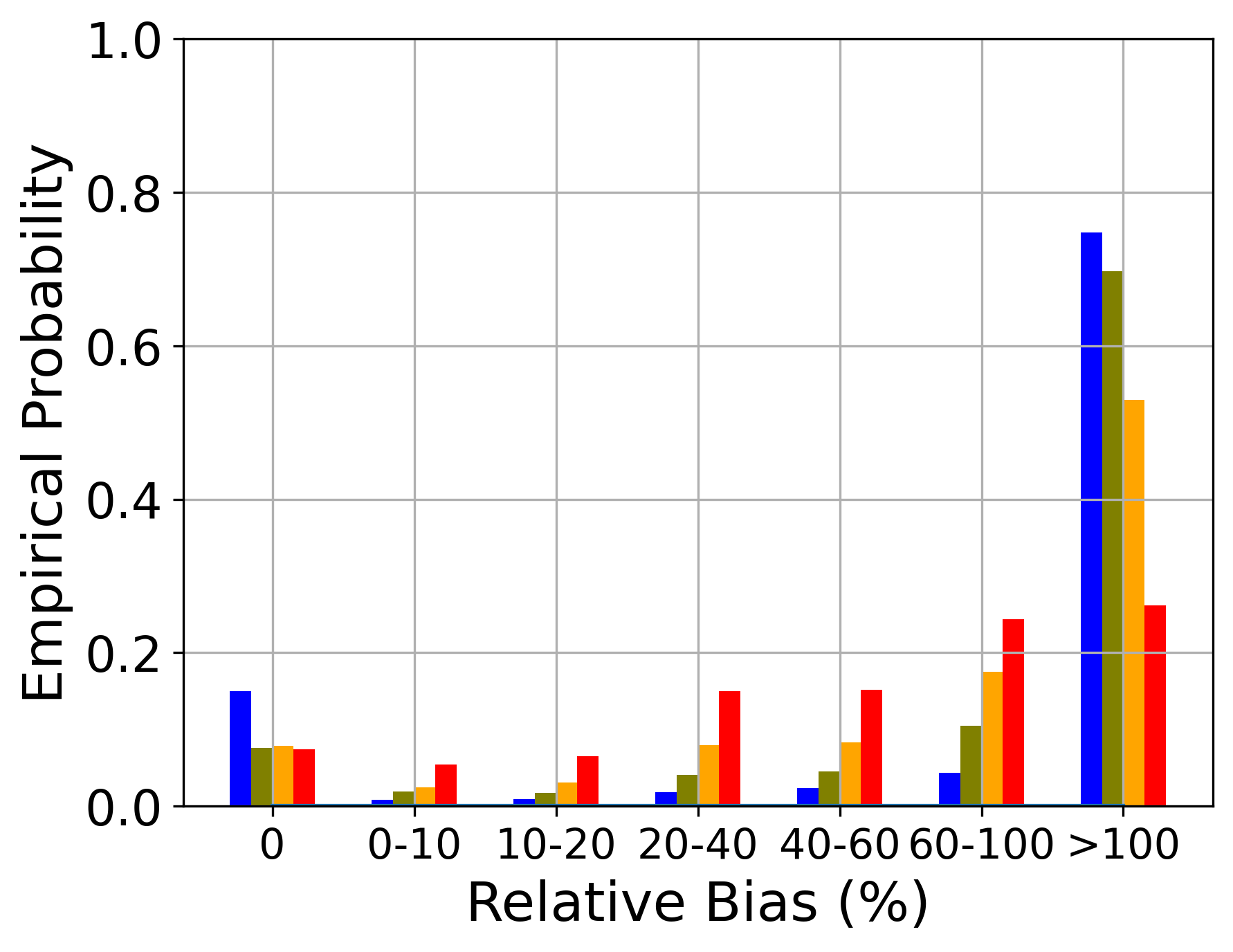}}\hfill
  \subfloat[][\textsf{Std }$200\%$, $\gamma=1.5$]{\includegraphics[width=.23\textwidth]{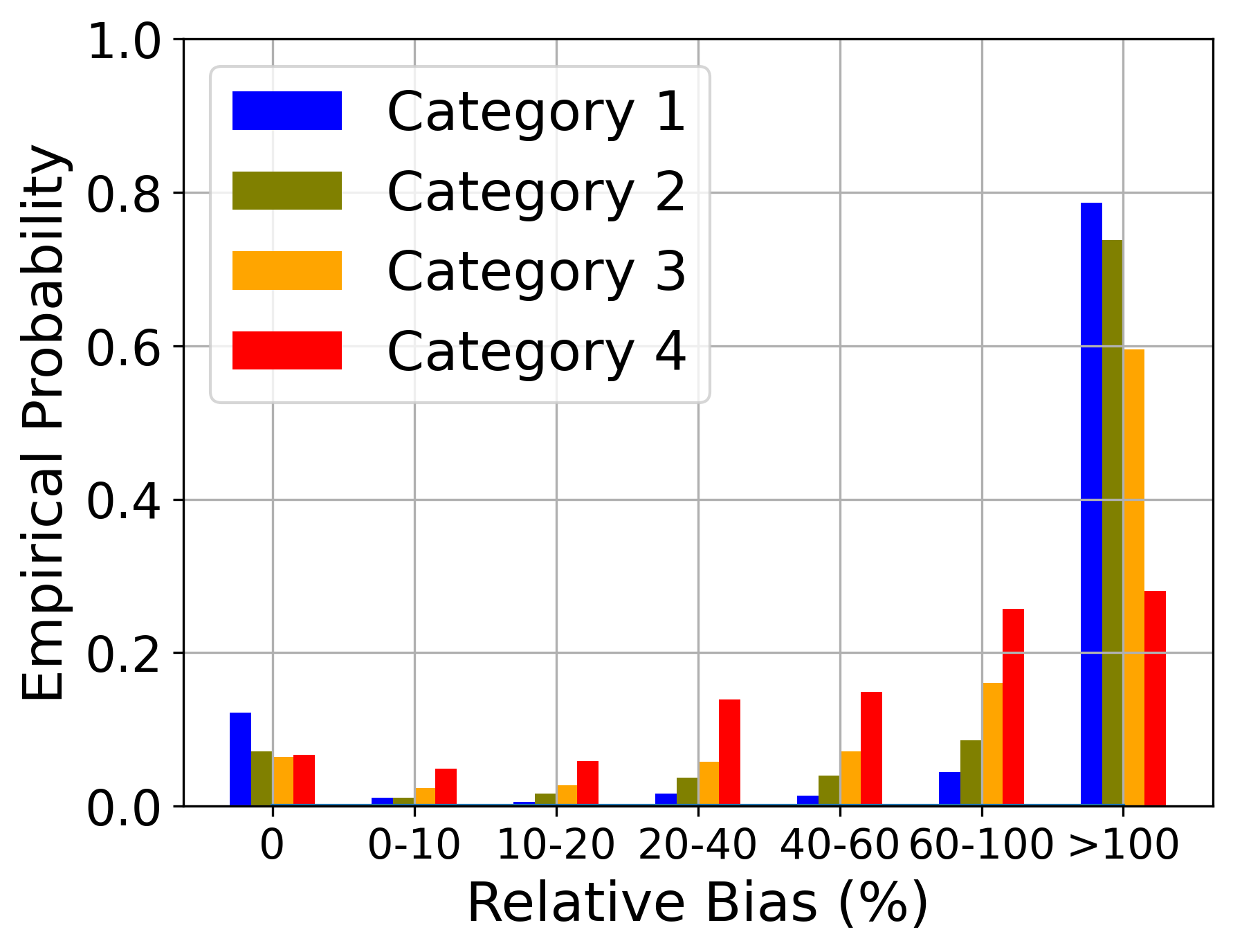}}\par
  \raisebox{35pt}{\parbox[b]{.05\textwidth}{}}
  \subfloat[][\textsf{Std }$20\%$, $\gamma=2$]{\includegraphics[width=.23\textwidth]{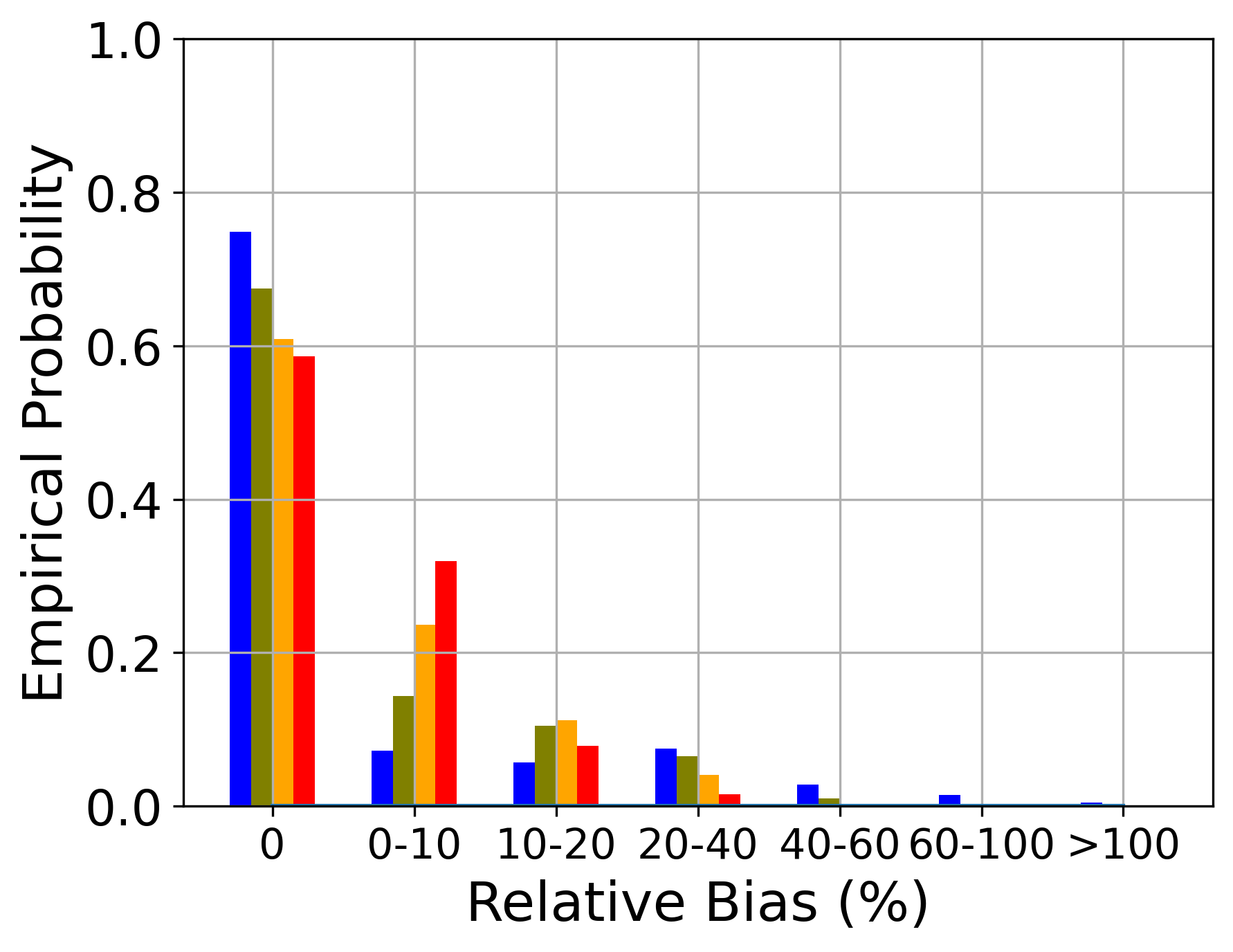}}\hfill
  \subfloat[][\textsf{Std }$50\%$, $\gamma=2$]{\includegraphics[width=.23\textwidth]{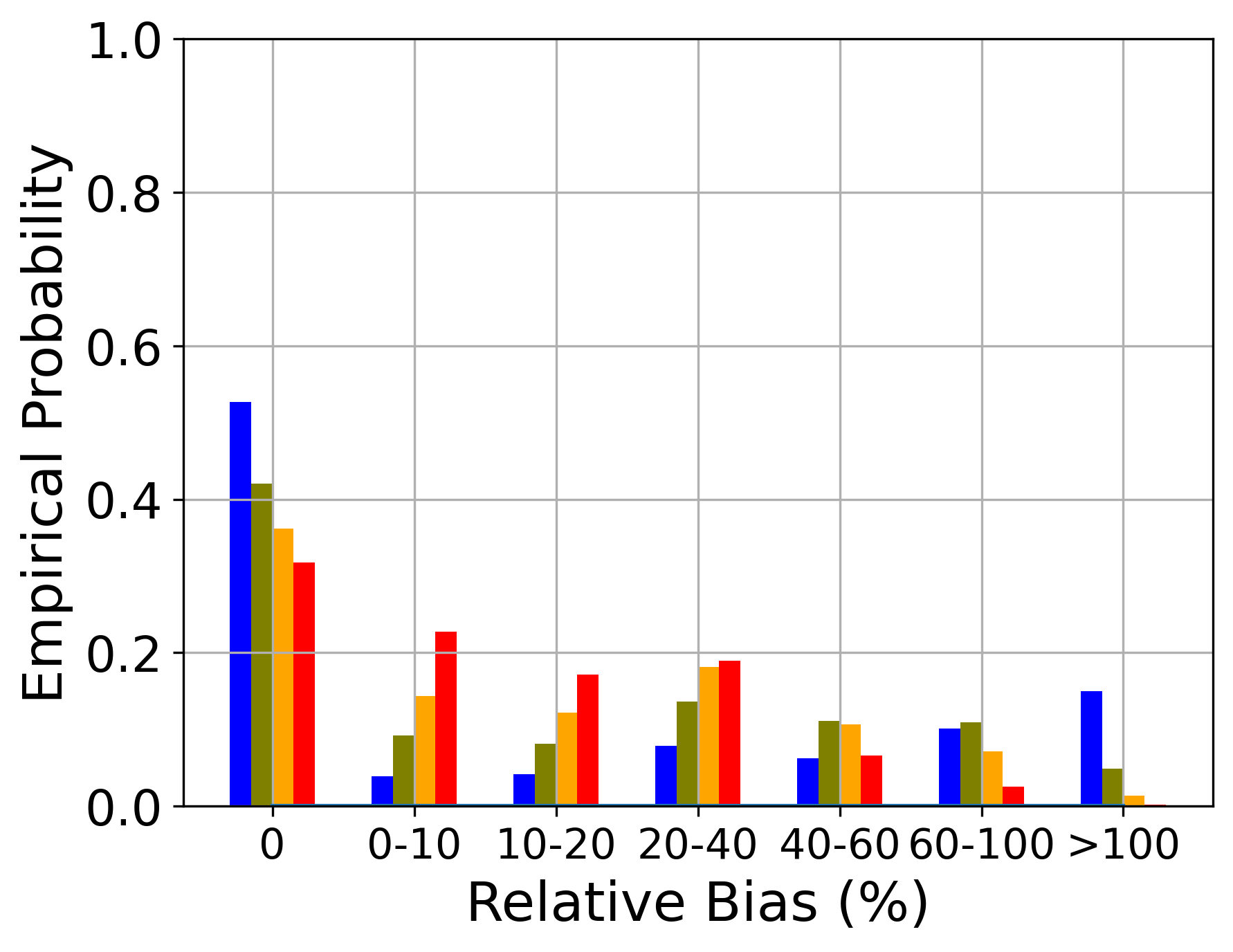}}\hfill
  \subfloat[][\textsf{Std }$100\%$, $\gamma=2$]{\includegraphics[width=.23\textwidth]{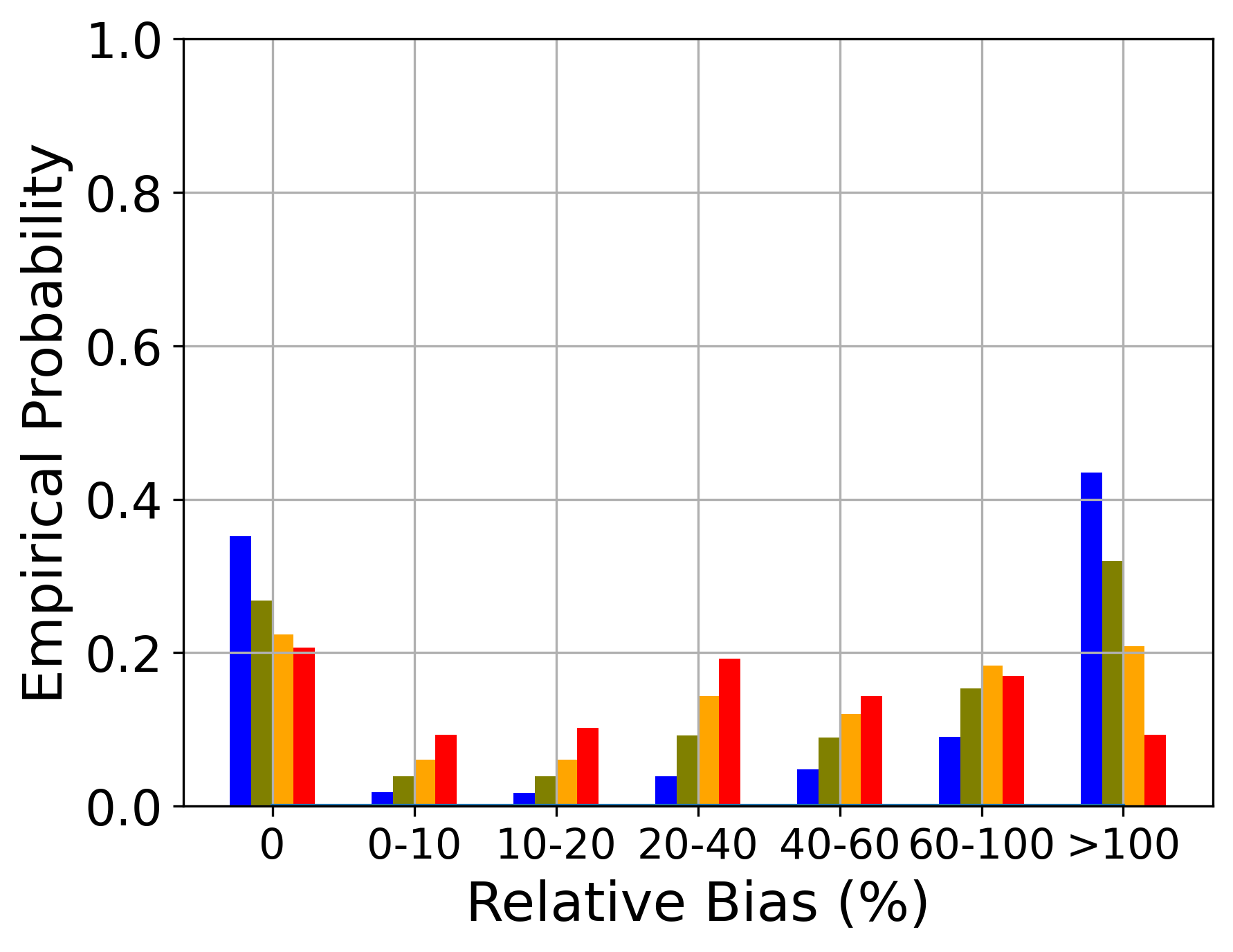}}\hfill
  \subfloat[][\textsf{Std }$200\%$, $\gamma=2$]{\includegraphics[width=.23\textwidth]{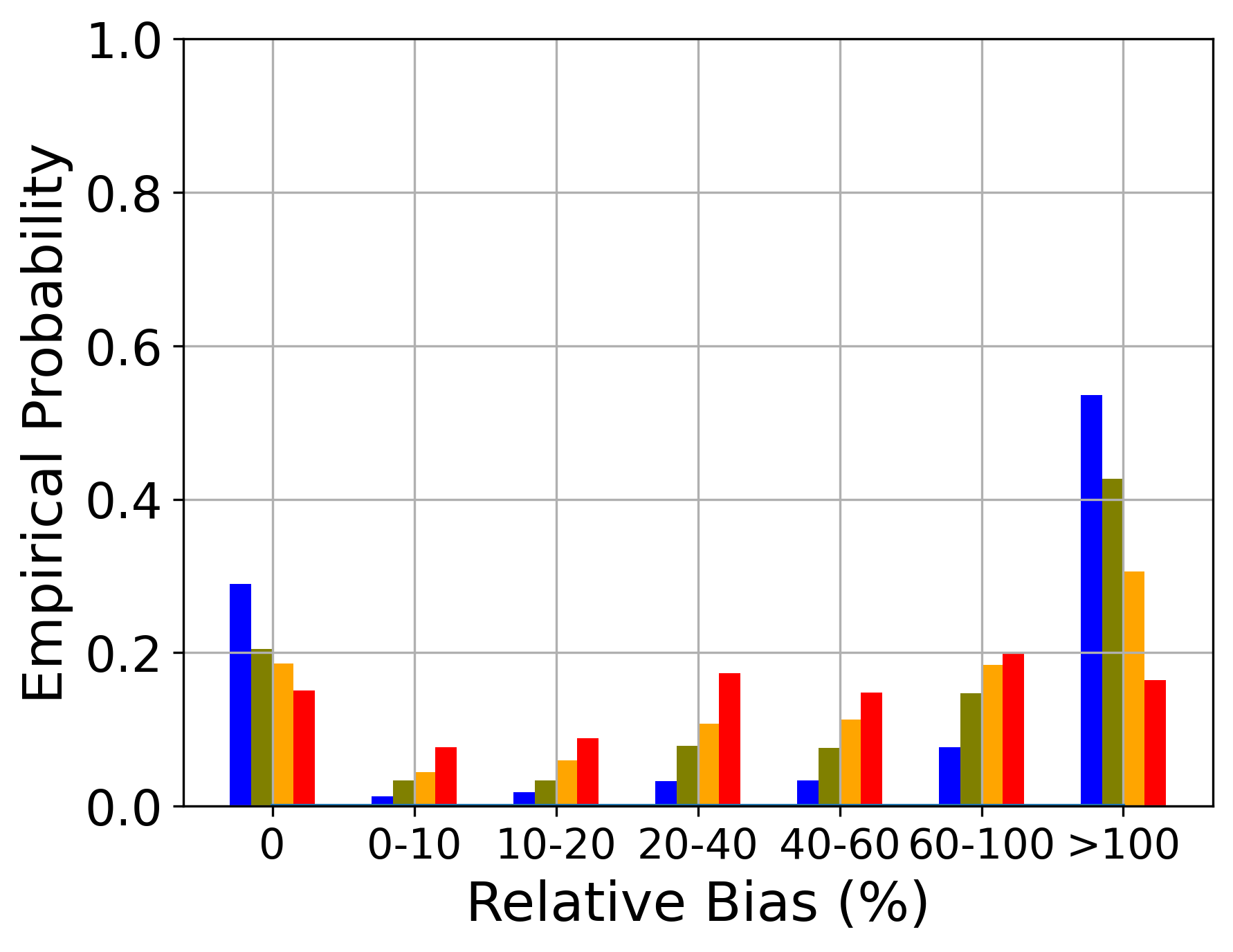}}\par
  \raisebox{35pt}{\parbox[b]{.05\textwidth}{}}
  \subfloat[][\textsf{Std }$20\%$, $\gamma=2.5$]{\includegraphics[width=.23\textwidth]{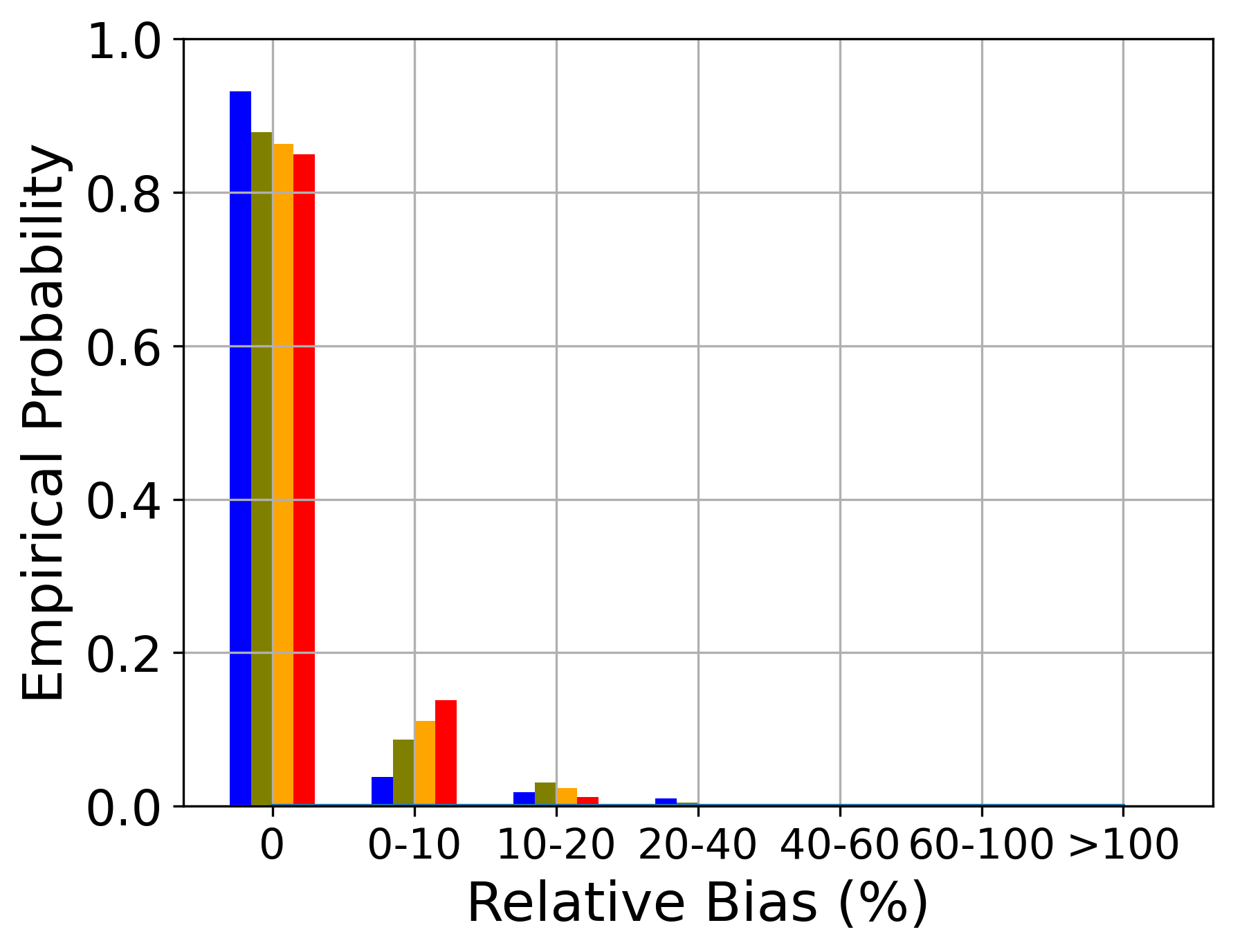}}\hfill
  \subfloat[][\textsf{Std }$50\%$, $\gamma=2.5$]{\includegraphics[width=.23\textwidth]{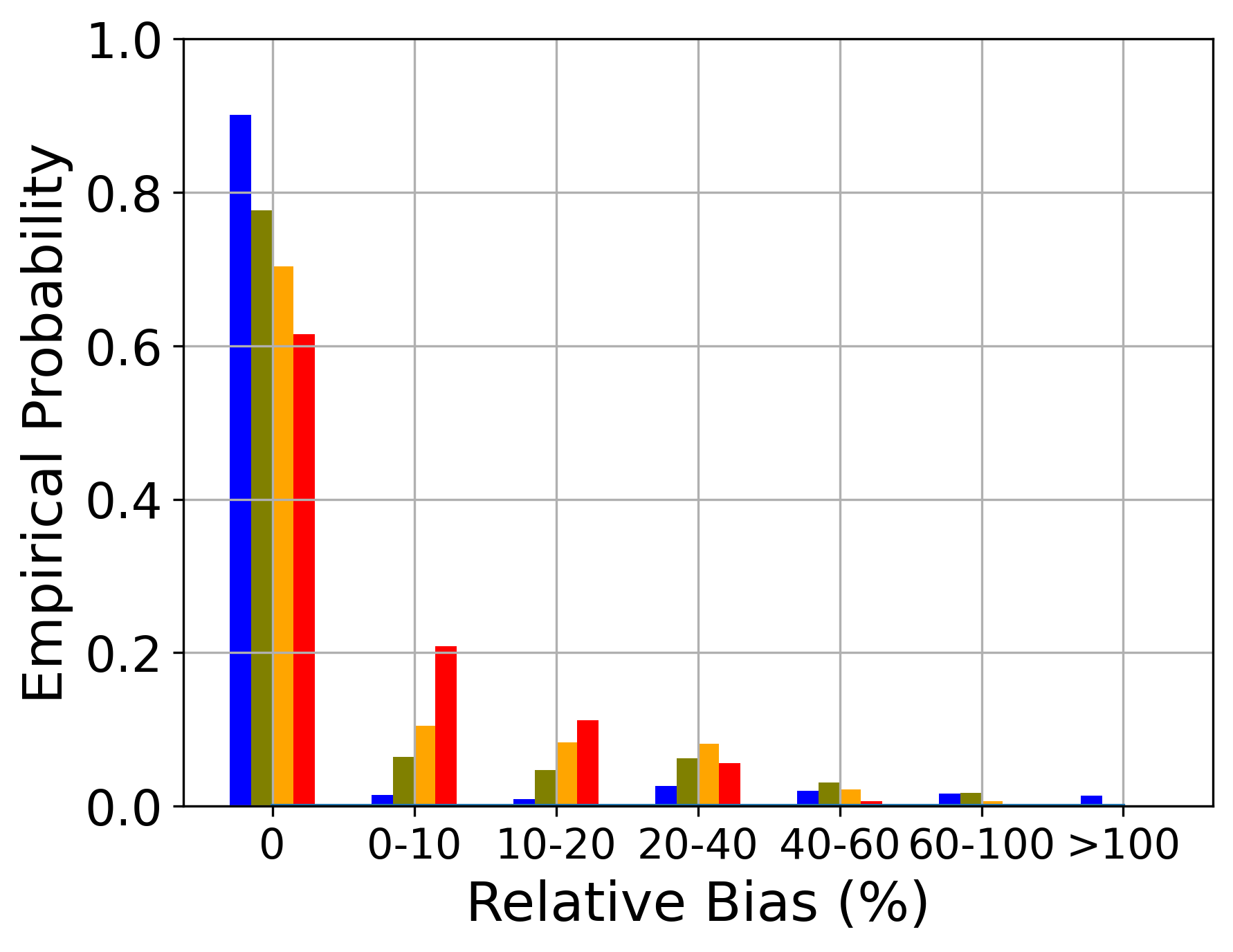}}\hfill
  \subfloat[][\textsf{Std }$100\%$, $\gamma=2.5$]{\includegraphics[width=.23\textwidth]{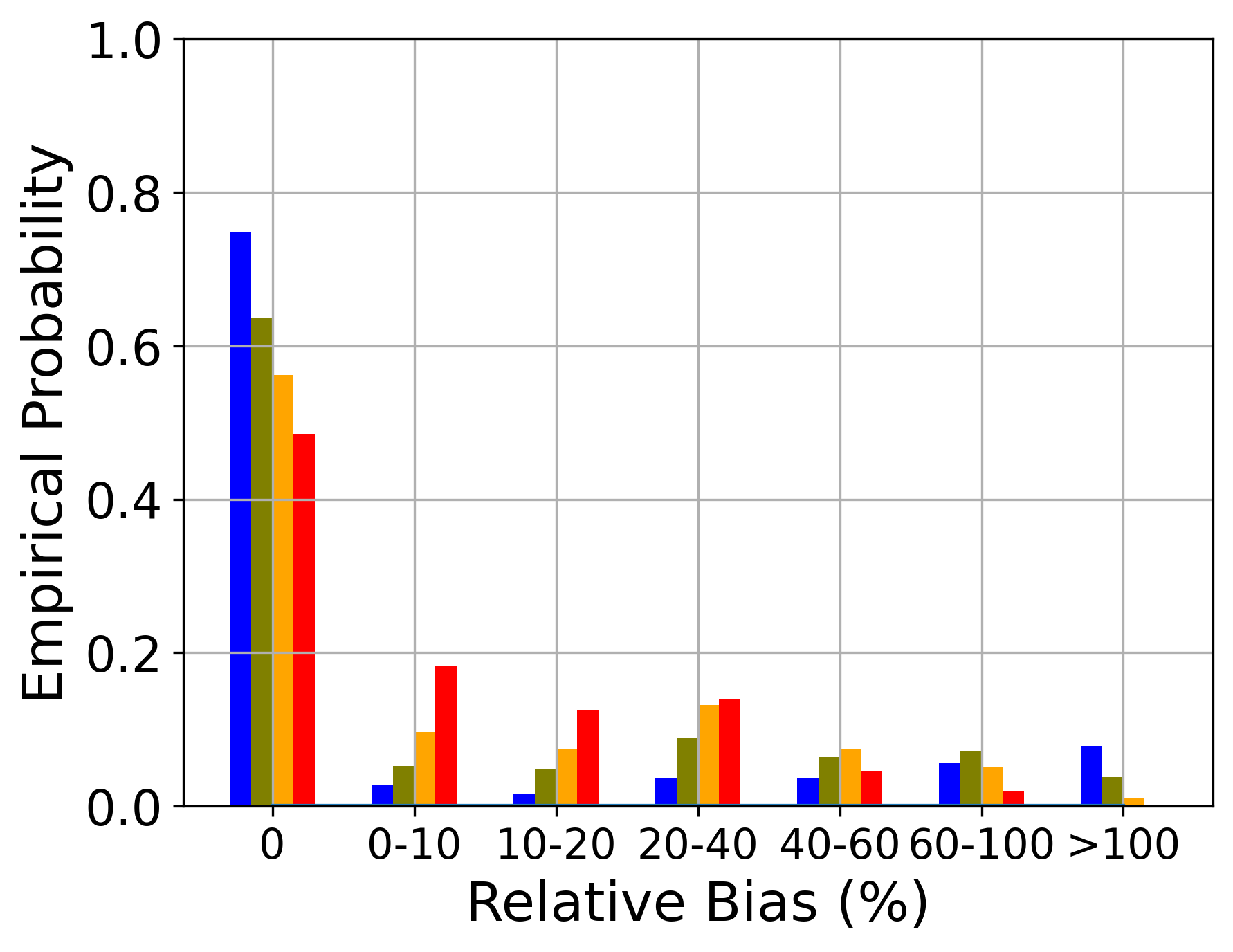}}\hfill
  \subfloat[][\textsf{Std }$200\%$, $\gamma=2.5$]{\includegraphics[width=.23\textwidth]{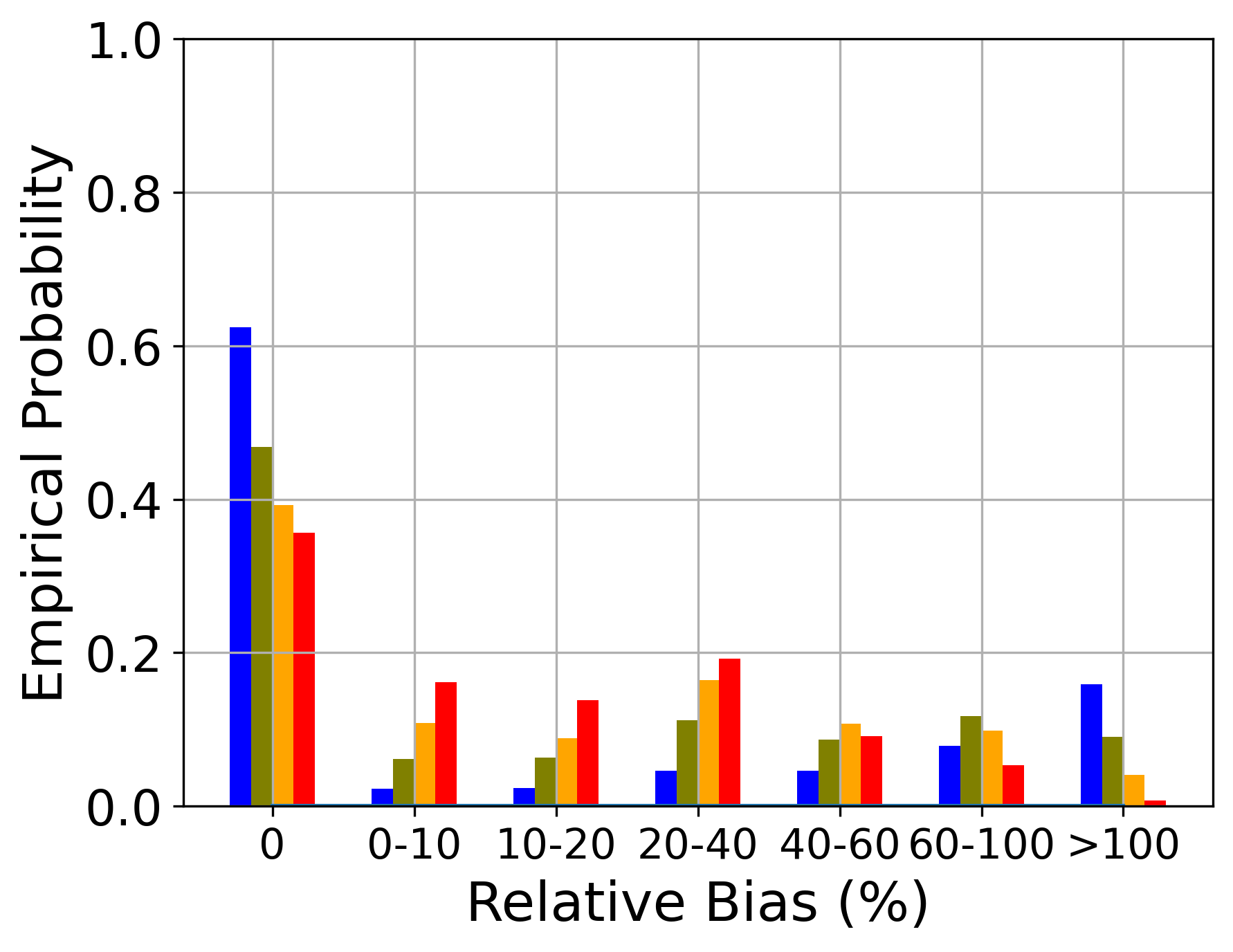}}\par
  \raisebox{35pt}{\parbox[b]{.05\textwidth}{}}
  \subfloat[][\textsf{Std }$20\%$, $\gamma=3$]{\includegraphics[width=.23\textwidth]{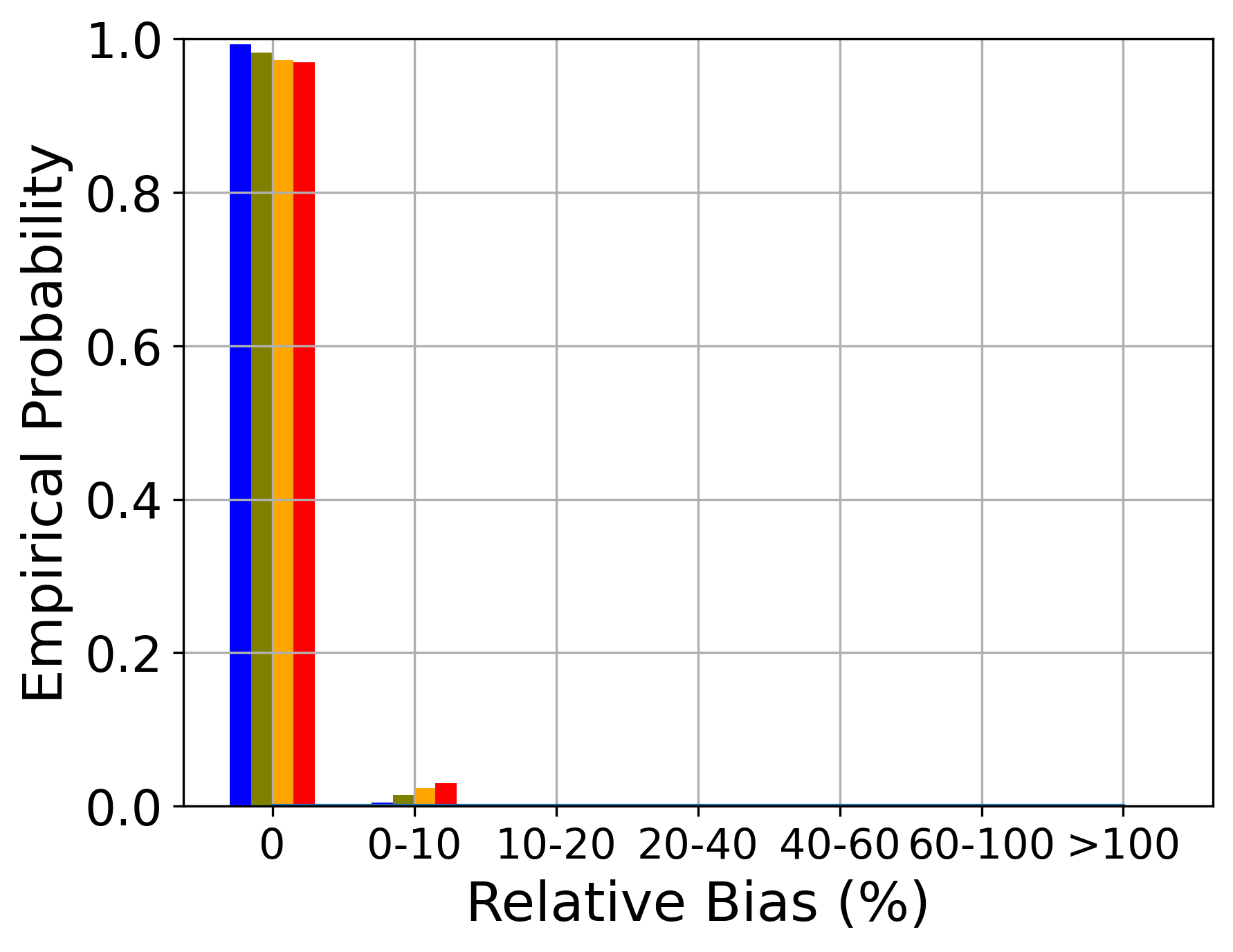}}\hfill
  \subfloat[][\textsf{Std }$50\%$, $\gamma=3$]{\includegraphics[width=.23\textwidth]{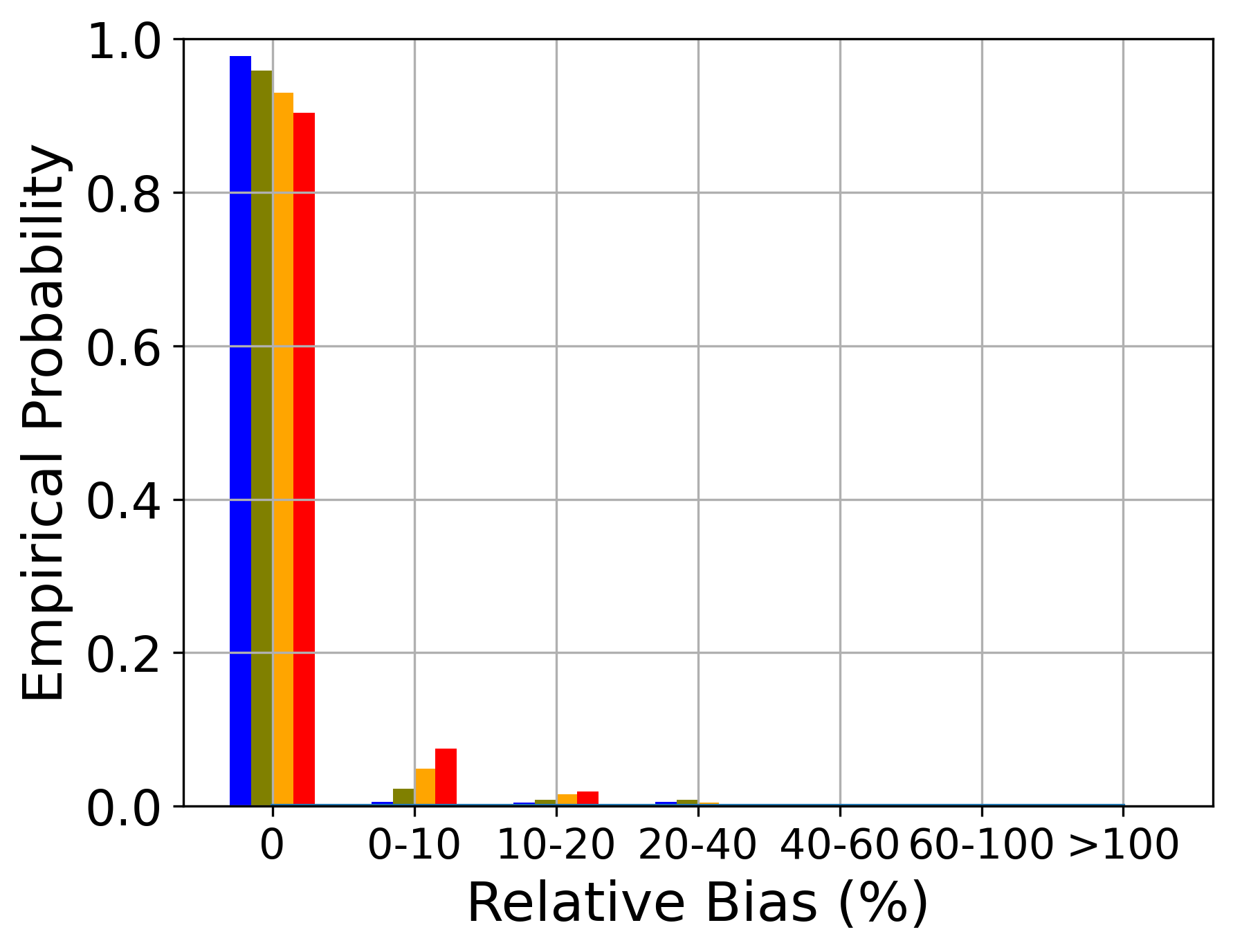}}\hfill
  \subfloat[][\textsf{Std }$100\%$, $\gamma=3$]{\includegraphics[width=.23\textwidth]{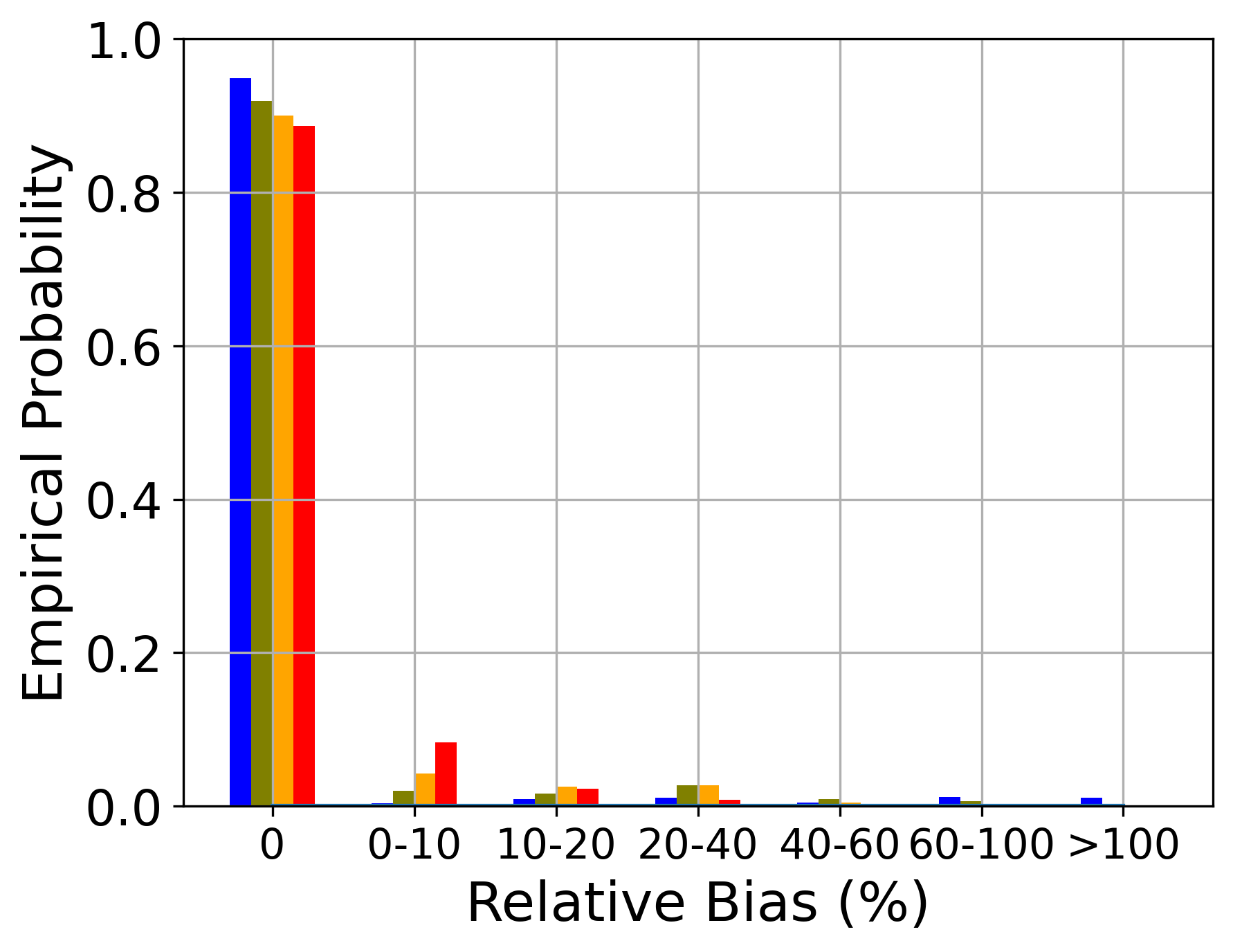}}\hfill
  \subfloat[][\textsf{Std }$200\%$, $\gamma=3$]{\includegraphics[width=.23\textwidth]{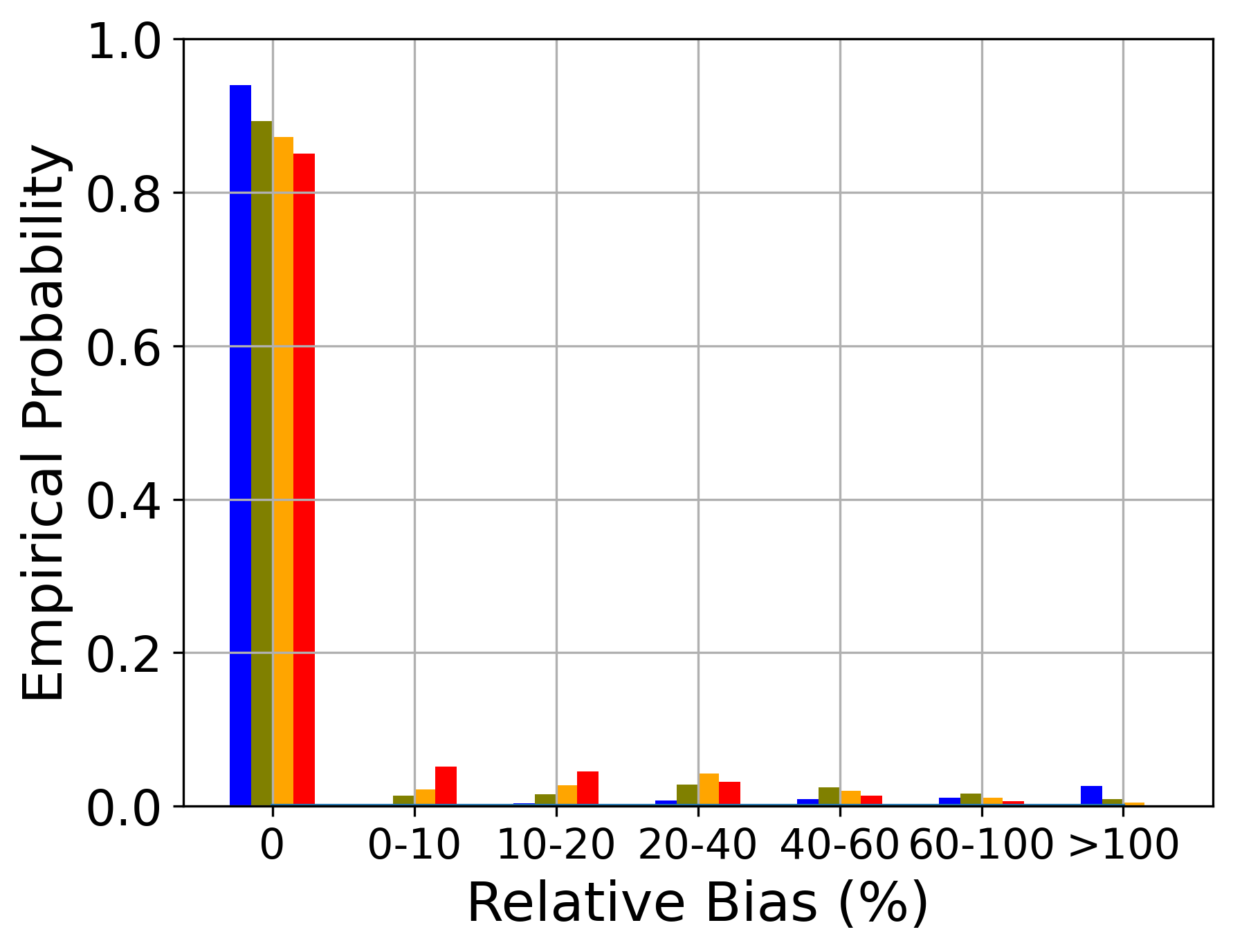}}
  \caption{Results for scale-free graphs generated using a starting $N$ value of $100$. From left to right, each row has results for a fixed power $\gamma$ and different levels of noise ($20\%$, $50\%$, $100\%$ and $200\%$). In each column, we have results at the same variance but different values of $\gamma$.}
  \label{fig:scalefree100}
\end{figure}

Similar to earlier results, higher levels of noise lead to a higher likelihood of incurring large relative bias across all categories of node pairs. 
At low levels of noise, \Qone{} node pairs still continue to be more robust to noise and more likely to remain unchanged compared to their \textsf{Category 4} counterparts (a consequence of Theorem \ref{thm:q_beta}).  

A striking observation is that in the case of scale-free graphs with lower values of $\gamma$ ($\gamma \leq 2$), \Qone{} node pairs are much more likely to incur significant amounts of relative bias ($> 100~\%$) compared to \Qfour{} pairs at moderate to high levels of noise. This is in sharp contrast with the results for our previous two graph classes where, typically, \Qfour{} pairs were \emph{worse-off} due to privacy. This is largely because of graph topology. When $\gamma \leq 2$, the graph has multiple densely connected centers that branch off into tree-like sub-graphs. A large proportion of \Qone{} pairs are located close to the centres and therefore have a large number of path alternatives. The \emph{path cardinality effect} increases their likelihood of incurring high bias. Further, \Qfour{} pairs are predominantly located on either side of connected centres---this means that they have, on average, the same number of path alternatives as their \Qone{} counterparts, but those paths have a high degree of overlap and only diverge near the centre. This causes \Qfour{} pairs to incur the same levels of absolute bias as the \Qone{} pairs, but they incur much smaller levels of relative bias because their paths are longer on average. This trend becomes less significant for $\gamma > 2$ due to change in the graph topology. As $\gamma$ increases, the number of nodes of high degree decrease significantly and the graph becomes less dense and more tree-like (as illustrated in Figure \ref{fig:scalefree100}). As a result, for most node-pairs, there exists a unique path to go from source to destination which explains the low levels of bias incurred across all node categories, i.e., increased robustness to privacy noise. In fact, as $\gamma$ approaches $3$, all node pairs have a greater than $80~\%$ chance of remaining unchanged---these likelihoods increase further and approach $100~\%$ at lower levels of noise.

\section{Discussion \& Future Work}\label{sec:disc}
In this work, we consider the problem of differentially private graph data release for downstream optimization tasks, particularly shortest path computation. We show that the noise introduced for privacy causes perceived shortest paths to shift from the true ones and thus incur positive bias. We provide analytical expressions to exactly compute or closely estimate the probability of incurring bias and infer how properties like how far two nodes are on the graph and how many alternate path choices they have, directly influence the probabilities and introduce disparities across different categories of source-destination pairs in terms of the degree of impact to privacy noise. Finally, we provide rigorous synthetic experiments on different classes of graphs to demonstrate the form and scale of disparities, in each case providing precise explanations from our theory on why such disparities occur. 

In this process, we highlight how different types of networks may face very different bias properties due to differential privacy noise and pre-processing to keep edge weights non-negative. This implies that there is a cautionary tale for planners using DP graph data, noting that design settings may not be re-usable across varying graph topologies and highlighting the importance of taking that topology into account when drawing conclusions. Our study helps with that direction, as it identifies graph properties (like sparsity and degree distribution) that affect and induce robustness to privacy noise and can inform network design for privacy-sensitive applications in the future. 

There are many interesting avenues of future work. For instance, since private graph data release affects shortest path computations as we show here, commuters on a road network may end up getting re-routed through other paths to reach their respective destinations. These effects may aggregate over the network and affect network-level traffic and congestion equilibria. It may also introduce sub-optimality in other types of layered decision problems, e.g., how to add new infrastructure to improve overall network performance. Characterizing these broader effects of privacy on networks is a key future direction.        

\bibliographystyle{plainnat}
\bibliography{mybib,fioretto}

\newpage 

\appendix
\section{Missing Proofs}\label{sec:proofs}
\subsection{Proof of Lemma \ref{lem:prob_2path}}
Note that the wrong path $P'$ can be chosen if and only if $w_{\widetilde G}(P') < w_{\widetilde G}(P^*)$. Therefore, 
\begin{align*}
    q &= \pr \left[ w_{\widetilde G}(P') < w_{\widetilde G}(P^*) \right] \\
    &= \pr \left[ w_{G}(P') + \sum_{e \in P'}Z(e) < w_{G}(P^*) + \sum_{e \in P^*}Z(e)   \right] \\
    &= \pr \left[ \sum_{e \in P'\setminus P^*}Z(e) - \sum_{e \in P^*\setminus P'}Z(e) < w_G(P^*) - w_G(P') \right] \\
    &= \pr \left[ \sum_{e \in P'\setminus P^*}Z(e) - \sum_{e \in P^*\setminus P'}Z(e) < -\alpha_{P',P^*} \right]\\
    &= \pr \left[ \sum_{e \in P'\setminus P^*}Z(e) + \sum_{e \in P^*\setminus P'}Y(e) < -\alpha_{P',P^*} \right].
\end{align*}
In the last step above, we substitute $Y(e) = - Z(e)$ for all $e \in P^*\setminus P'$. Note that $Y(e)$ and $Z(e)$ are identically distributed (because mean-zero Gaussian random variables are symmetric). Since each $Z(e), Y(e) \sim N(0, \sigma^2)$ and they are independent of each other, $\sum_{e \in P'\setminus P^*}Z(e) + \sum_{e \in P^*\setminus P'}Y(e) \sim N(0, \left| S_{P',P^*} \right|\sigma^2)$.
This implies:
\begin{align*}
    q &= \pr\left[ \frac{ \sum_{e \in P'\setminus P^*}Z(e) + \sum_{e \in P^*\setminus P'}Y(e) }{\sigma \sqrt{|S_{P',P^*}|}} < \frac{-\alpha_{P',P^*}}{\sigma \sqrt{|S_{P',P^*}|}} \right] \\
    &= \Phi\left( \frac{-\alpha_{P',P^*}}{\sigma \sqrt{|S_{P',P^*}|}} \right) = \Phi^{c}\left( \frac{\alpha_{P',P^*}}{\sigma \sqrt{|S_{P',P^*}|}} \right).
\end{align*}
The last step invokes the symmetry of a standard normal variable which allows, for any $a > 0$, $\Phi(-a) = \Phi^c(a)$. This concludes the proof of the lemma. 

\subsection{Proof of Theorem \ref{thm:q_beta}}
We can express $q_{\beta}$ as the following probability: 
\begin{align*}
    q_{\beta} &= \pr\left[ \text{shortest path on $\widetilde G$ is $\beta$-worse} \right] \\
    &= \pr\left[ \exists ~P \in \cP_{ij}^{\geq \beta}: w_{\widetilde G}(P) < w_{\widetilde G}(R) ~\forall~ R \in \cP_{ij}\setminus P \right] \\
    &\stackrel{(i)}{=} \sum_{P \in \cP_{ij}^{\geq \beta}} \pr\left[ w_{\widetilde G}(P) < w_{\widetilde G}(R) ~\forall~ R \in \cP_{ij}\setminus P  \right] \\
    &= \sum_{P \in \cP_{ij}^{\geq \beta}} \pr\left[ \bigcap_{R \in \cP_{ij}\setminus P } \{w_{\widetilde G}(P) < w_{\widetilde G}(R) \}   \right].
\end{align*}
The equality in step $(i)$ above follows from the fact that events of the type $\{ w_{\widetilde G}(P) < w_{\widetilde G}(R) ~\forall~ R \in \cP_{ij}\setminus P  \}$ are disjoint since two different paths cannot the best simultaneously (the event that two continuous random variables are equal, occurs with probability $0$). Now, for each $P \in \cP_{ij}^{\geq \beta}$, note that $P^* \in \cP_{ij} \setminus P$. Therefore, we have:
\[
     \pr\left[ \bigcap_{R \in \cP_{ij}\setminus P } \{ w_{\widetilde G}(P) < w_{\widetilde G}(R) \}  \right] \leq \pr\left[ w_{\widetilde G}(P) < w_{\widetilde G}(P^*)  \right] = \Phi^c \left(\frac{ \alpha_{P,P^*}}{\sigma \sqrt{|S_{P,P^*}|}}  \right),
\]
where the last equality follows from Lemma \ref{lem:prob_2path}. It is important to note that we cannot compute the probability of the intersection event in closed form because the individual events are not mutually independent (two paths may have overlapping edges). Summing over all $P \in \cP_{ij}^{\geq \beta}$, we derive the following upper bound: 
\[
       q_{\beta} \leq \sum_{P \in \cP_{ij}^{\geq \beta}}  \Phi^c \left(\frac{\alpha_{P,P^*}}{\sigma \sqrt{|S_{P,P^*}|}}  \right).
\]
Finally, noting that $\alpha_{P,P^*} \geq \beta$ for all $P \in \cP_{ij}^{\geq \beta}$ and from the definition of $S_{max}$, we have:
\[
        \Phi^c \left( \frac{\alpha_{P,P^*}}{\sigma \sqrt{|S_{P,P^*}|}} \right) \leq \Phi^c \left( \frac{\beta}{\sigma \sqrt{S_{max}}} \right) \quad \forall~P \in \cP_{ij}^{\geq \beta}. 
\]
This helps us simplify the upper bound even further and obtain the final result: 
\[
    q_{\beta} \leq \sum_{P \in \cP_{ij}^{\geq \beta}}  \Phi^c \left(\frac{\alpha_{P,P^*}}{\sigma \sqrt{|S_{P,P^*}|}}  \right) \leq \left| \cP_{ij}^{\geq \beta} \right| \cdot \Phi^c\left( \frac{\beta}{\sigma \sqrt{S_{max}}} \right).
\]

\subsection{Proof of Corollary \ref{corr:bound_beta}}
Note that showing $\pr\left[B_{ij} < \sqrt{2} \left(\sigma z^* \sqrt{S}\right) \right] \geq 1-\gamma$ is equivalent to showing that: 
\[
       \pr\left[B_{ij} \geq \sqrt{2} \left(\sigma z^* \sqrt{S}\right) \right] \leq \gamma,
\]
which again, is equivalent to showing $q_{\beta} \leq \gamma$ where $\beta = \sqrt{2} \left(\sigma z^* \sqrt{S}\right)$. Now, recall that we have already shown in Theorem \ref{thm:q_beta} that for any $\beta > 0$, we have:
\[
       q_{\beta} \leq \left| \cP_{ij}^{\geq \beta} \right| \cdot \Phi^c \left( \frac{\beta}{\sigma \sqrt{S_{max}}} \right).  
\]
We can construct a slightly more conservative upper bound on $q_{\beta}$ by noting that $|\cP_{ij}^{\geq \beta}| \leq |\cP_{ij}|$ and $S_{max} \leq 2S$ (in the worst case, all paths in $\cP_{ij}$ have $S$ edges and have no overlapping edges which leads to $S_{max} = 2S$). Therefore, 
\begin{align}\label{eq:rev_bound}
      q_{\beta} \leq \left| \cP_{ij} \right| \cdot \Phi^c \left( \frac{\beta}{\sigma \sqrt{2S}} \right). 
\end{align}
Hence, it is sufficient to show that when $\beta = \sqrt{2} \left(\sigma z^* \sqrt{S}\right)$, the revised upper bound in Equation~\ref{eq:rev_bound} is $\leq \gamma$. This can be verified easily by plugging in the value of $\beta$, as follows:
\begin{align*}
    |\cP_{ij}| \cdot \Phi^c \left( \frac{\beta}{\sigma \sqrt{2S}} \right) &= |\cP_{ij}| \cdot \Phi^c \left( \frac{ \sigma z^* \sqrt{2S}}{ \sigma \sqrt{2S}} \right) \\
    &= |\cP_{ij}| \cdot \Phi^c \left( z^* \right) \\
    &= |\cP_{ij}| \cdot \left(1 - \Phi(z^*) \right) \\
    &= |\cP_{ij}| \cdot \frac{\gamma}{|\cP_{ij}|}\\
    &= \gamma.
\end{align*}
This concludes the proof of the corollary. 

\section{A Special Case: Non-Overlapping Paths}\label{sec:nonoverlap}
\subsection{Exact characterization of $q_{\beta}$} 
We consider a special case where none of the paths in $\cP_{ij}$ have overlapping edges. In this case, we will show that it is possible to derive an exact expression for $q_{\beta}$. 

\begin{corr}\label{corr:non_overlap}
When paths in $\cP_{ij}$ have no overlapping edges, the probability $q_{\beta}$ can be computed exactly and is given by: 
\[
        q_{\beta} = \sum_{P \in \cP_{ij}^{\geq \beta}} \int_{-\infty}^{\infty} \prod_{R \in \cP_{ij}\setminus P } \Phi^{c}\left( \frac{t-w_G(R)}{\sigma \sqrt{n_R}} \right) \phi\left( \frac{t-w_G(P)}{\sigma \sqrt{n_P}} \right) dt.
\]
\end{corr}
\begin{proof}
The proof is similar to Theorem \ref{thm:q_beta}, the only difference being that we can compute the probability of the intersection event in closed form this time. We have already shown that: 
\[
   q_{\beta} = \sum_{P \in \cP_{ij}^{\geq \beta}} \pr\left[ \bigcap_{R \in \cP_{ij}\setminus P } \{w_{\widetilde G}(P) < w_{\widetilde G}(R) \}   \right].
\]
Using a conditioning argument, we can rewrite as follows: 
\begin{align*}
    q_{\beta} &= \sum_{P \in \cP_{ij}^{\geq \beta}} \int_{-\infty}^{\infty} \pr\left[ w_{\widetilde G}(R) > w_{\widetilde G}(P) \quad \forall~R \in \cP_{ij}\setminus P ~|~ w_{\widetilde G}(P) = t  \right]\cdot f_{w_{\widetilde G}(P)}(t)dt \\
    &= \sum_{P \in \cP_{ij}^{\geq \beta}} \int_{-\infty}^{\infty} \pr\left[ w_{\widetilde G}(R) > t \quad \forall~R \in \cP_{ij}\setminus P  \right]\cdot f_{w_{\widetilde G}(P)}(t)dt 
\end{align*}
Note that $w_{\widetilde G}(P) \sim N(w_{G}(P), n_P \sigma^2)$. Also, since no paths in $\cP_{ij}$ overlap, we have an intersection of independent events and therefore, 
\[
     \pr\left[ w_{\widetilde G}(R) > t \quad \forall~R \in \cP_{ij}\setminus P  \right] = \prod_{R \in \cP_{ij}\setminus P} \pr\left[ w_{\widetilde G}(R) > t \right] = \prod_{R \in \cP_{ij}\setminus P} \Phi^{c}\left( \frac{t-w_G(R)}{\sigma \sqrt{n_R}} \right).
\]
Plugging everything back in and substituting the probability density function for $w_{\widetilde G}(P)$, we obtain the final result. 
\end{proof}

\subsection{Comparison of $q_{\beta}$ with bounds from Theorem \ref{thm:q_beta}}
\begin{figure}[!ht]
  \centering
  \raisebox{20pt}{\parbox[b]{.11\textwidth}{}}%
  \subfloat[][]{\includegraphics[width=.40\textwidth]{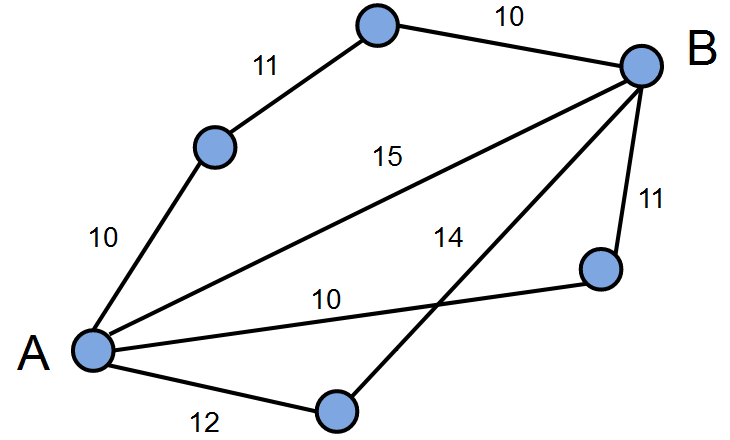}}\hfill
  \subfloat[][]{\includegraphics[width=.55\textwidth]{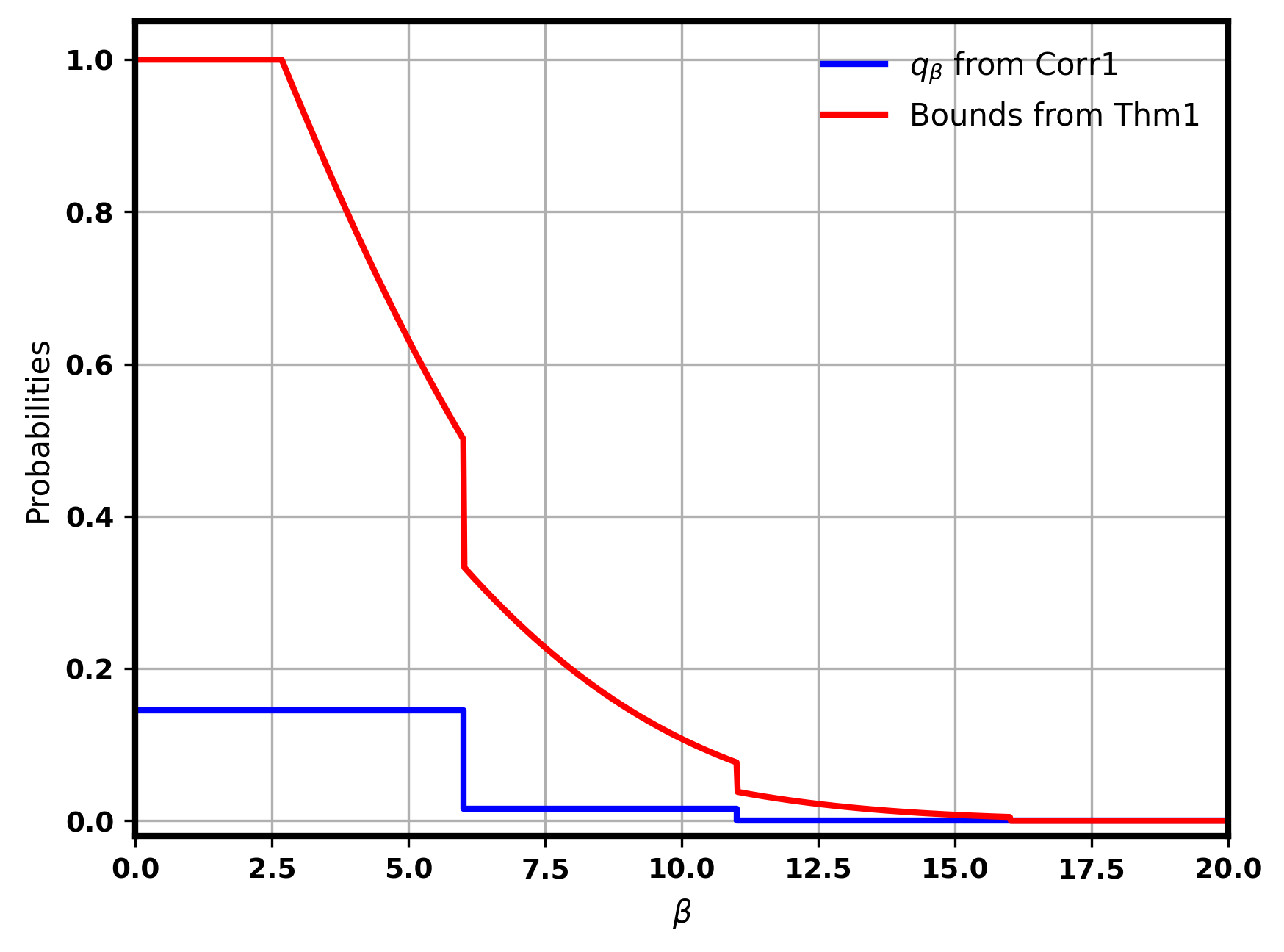}}\par
  \caption{(a) represents a toy graph with $4$ non-overlapping paths between nodes $A$ and $B$. The ground truth edge weights are indicated. For this graph, we compare the exact values of $q_{\beta}$ (given by Corollary \ref{corr:non_overlap}) and the general upper bounds (given by Theorem \ref{thm:q_beta}) in (b). As expected, the upper bounds are conservative. }
  \label{fig:comp_nonoverlap}
\end{figure}
We construct a toy graph with $6$ nodes and $4$ non-overlapping paths between source and destination nodes $A$ and $B$. The ground truth weights for all edges on the graph are indicated in sub-figure (a) above. We set $(\varepsilon, \delta) = (1, 0.01)$ which fixes the standard deviation of the noise $\sigma$. Now in (b), we plot the exact values of $q_{\beta}$ obtained from Corollary \ref{corr:non_overlap} along side the upper bound on $q_{\beta}$ provided by Theorem \ref{thm:q_beta} as a function of the bias $\beta$. At very low values of $\beta$, the upper bound is vacuous, however by the time $\beta \approx 0.5 \times \text{Mean Edge weight}$ (this graph has a mean edge weight of $11.625$), the upper bound begins to approximate the true probability $q_{\beta}$ quite well. This toy example demonstrates that for large graphs and for reasonable values of $\beta$, the upper bound provided in Theorem \ref{thm:q_beta} (which can be computed cheaply) can be used as an estimate for $q_{\beta}$.

\end{document}